\documentclass[12pt,oneside,reqno]{amsart}
\usepackage{geometry}
\usepackage[utf8]{inputenc}
\usepackage[T1]{fontenc}
\usepackage{verbatim}
\usepackage{amsmath}
\usepackage{amsthm}
\usepackage{mathabx}
\usepackage{amssymb}
\usepackage{amsfonts}
\usepackage{mathtools}
\usepackage{nccmath}
\usepackage{caption}
\usepackage{tikz}
\usepackage{xcolor}
\usepackage{enumitem}
\usepackage{braket}
\usepackage{mathrsfs}
\usepackage{verbatim}
\usepackage[english]{babel}
\usepackage{hyphenat}
\usepackage{bbm}
\usepackage{calligra}
\usepackage{empheq}
\usepackage{mathtools}
\usepackage{dutchcal}
\usepackage[T1]{fontenc}
\usepackage{lmodern}
\usepackage[nopatch=footnote]{microtype}
\usepackage{fmtcount}
\usepackage[dvipsnames]{xcolor}
\usepackage{comment}
\usepackage{soul}

\usepackage[numbers]{natbib}
\usepackage[colorlinks=true, linkcolor=Maroon, urlcolor=Maroon]{hyperref}

\newcommand{\RN}[1]{%
  \textup{\uppercase\expandafter{\romannumeral#1}}%
}

\newif\ifshowdetails
\showdetailsfalse

\ifshowdetails
  \newenvironment{details}{%
    \par\medskip\noindent
    \stepcounter{details}%
    {\color{Red}\bfseries Details \thedetails:}\ %
    \begingroup\color{RoyalBlue}\small
  }{%
    \endgroup\par\medskip
  }
\else
  \excludecomment{details}
\fi

\newcounter{details}[section]
\renewcommand{\thedetails}{\thesection.\arabic{details}}

\usepackage[english]{babel}

\allowdisplaybreaks
\usepackage{bm} 

\usepackage{graphicx,url}

\usepackage[svgnames]{xcolor}

\hypersetup{citecolor=blue}

\hypersetup{linkcolor=Maroon,citecolor=DarkBlue,urlcolor=red}

\usepackage{dsfont}
\usepackage{suetterl} 
\usepackage{pgfplots}	
\pgfplotsset{compat=1.18} 

\usepackage{caption}
\usepackage{subcaption}		
\usepackage{braket}
\usepackage{stmaryrd} 

\usepackage[T1]{fontenc}
\usepackage{currvita}

\usepackage{csquotes} 

\newcommand{\assumref}[1]{\hyperref[#1]{\textcolor{Maroon}{Assumption~\ref*{#1}}}}

\newcommand{\secref}[1]{\hyperref[#1]{\textcolor{Maroon}{Section~\ref*{#1}}}}
\newcommand{\hypref}[1]{\hyperref[#1]{\textcolor{Maroon}{Hypothesis~\ref*{#1}}}}

\newcommand{\appref}[1]{\hyperref[#1]{\textcolor{Maroon}{Appendix~\ref*{#1}}}}

\newcommand{\appsubref}[1]{\hyperref[#1]{\textcolor{Maroon}{Appendix~\ref*{#1}}}}

\newcommand{\thmref}[1]{\hyperref[#1]{\textcolor{Maroon}{Theorem~\ref*{#1}}}}
\newcommand{\lemref}[1]{\hyperref[#1]{\textcolor{Maroon}{Lemma~\ref*{#1}}}}
\newcommand{\propref}[1]{\hyperref[#1]{\textcolor{Maroon}{Proposition~\ref*{#1}}}}
\newcommand{\corref}[1]{\hyperref[#1]{\textcolor{Maroon}{Corollary~\ref*{#1}}}}
\newcommand{\remref}[1]{\hyperref[#1]{\textcolor{Maroon}{Remark~\ref*{#1}}}}
\newcommand{\figref}[1]{\hyperref[#1]{\textcolor{Maroon}{Figure~(\ref*{#1})}}}

\makeatletter
\newcommand{\hitem}[2]{%
  \item[\textbf{(H#1)}] \phantomsection\def\@currentlabel{#1}\label{#2}%
}
\makeatother

\renewcommand{\hypref}[1]{\hyperref[#1]{\textcolor{Maroon}{(H\ref*{#1})}}}

\usepackage{simpler-wick} 	
\newlength{\wdth}

\usepackage{tikz}
\usepgfplotslibrary{fillbetween}
\usetikzlibrary{patterns}

\usetikzlibrary{arrows}
\usetikzlibrary{arrows,positioning,arrows.meta} 

\allowdisplaybreaks		

\def\C{{\mathbb C}}
\def\N{{\mathbb N}}	
\def\R{{\mathbb R}}				 

\def\cB{{\mathcal B}}

\def\cE{{\mathcal E}}

\def\cH{{\mathcal H}}     \def \ch{{\mathcal h}}

\def\cL{{\mathcal L}}
     
\def\cN{{\mathcal N}}

\def\cR{{\mathcal R}}

\def\cX{{\mathcal X}}

\def\cO{{\mathcal O}}

\def \Resa{\mathcal{Res}}
\def \Rema{\mathcal{Rem}}
\def \Resaa{\mathcal{Res}^{(-)}}
\def \Resab{\mathcal{Res}^{(+)}}

\def \Reskaa{\mathcal{Res}_{k}^{(-)}}
\def \Reskab{\mathcal{Res}_{k}^{(+)}}
\def \Remka{\mathcal{Rem}_{k}}
\def \Remkaa{\mathcal{Rem}_{k}^{(-)}}
\def \Remkab{\mathcal{Rem}_{k}^{(+)}}
\def \Remkab{\mathcal{Rem}_{k}^{(+)}}		

\def \Lhar{\Lambda^{Har}_{k,k';j,j'}}
\def \Lp{\Lambda_{k,k';j,j'}^{+}}
\def \Lm{\Lambda_{k,k';j,j'}^{-}}
\def \Gm{\Gamma_{k,k';j,j'}^{-}}
\def \Gp{\Gamma_{k,k';j,j'}^{+}}
\def \Lpkkdjdj{{\Lambda_{k,k';j',j}}^{+}}
\def \Lmkkdjdj{\Lambda_{k,k';j',j}^{-}}
\def \Gmkkdjdj{\Gamma_{k,k';j',j}^{-}}
\def \Gpkkdjdj{\Gamma_{k,k';j',j}^{+}}
\def \FGRkj{{\Gamma_{k,k';j,j'}^{{FGR}}}}
\def \FGRk{{\Gamma_{k,k'}^{{FGR}}}}

\def \ms #1 {\mathscr{#1}}	

\def \pv {\operatorname{PV}}

\def\rchi{\raisebox{0.3ex}{$\chi$}}

\renewcommand{\Im}[1]{ \mathfrak{Im}\prc{ #1 }}
\renewcommand{\Re}[1]{ \mathfrak{Re}\prc{ #1 }}

\def\1{{{\mathds 1}}}

\newcommand{\ceil}[1]{\left\lceil #1 \right\rceil}

\newcommand{\supp}{\operatorname{supp}}

\def\eqn*{\begin{align*}}			
	\def\eeqn*{\end{align*}}		
\def\eqn{\begin{eqnarray}}				  
	\def\eeeqn{\end{eqnarray}}

\makeatletter
\def\namedlabel#1#2{\begingroup					
	#2%
	\def\@currentlabel{#2}%
	\phantomsection\label{#1}\endgroup
}
\makeatother

\theoremstyle{plain}
\newtheorem{theorem}{Theorem}[section]		
\newtheorem{proposition}[theorem]{Proposition}		   

\newtheorem{lemma}[theorem]{Lemma}

\newtheorem{remark}[theorem]{Remark}
\newtheorem*{remark*}{Remark}
		
\newtheorem{assumption}{Assumption}		
\numberwithin{equation}{section}

\def\Hf{\cH_f}

\def\veps{\vep}

\def\vfi{\phi} 

\def\cRep{\cR_\vep}

\def\vep{{\varepsilon}}

\def\eps{\epsilon}

\def\lam{\lambda}

\def\omg{\omega}

\makeatother

\makeatletter
\def\paragraph{\@startsection{paragraph}{4}%
  \z@\z@{-\fontdimen2\font}%
  {\normalfont\bfseries}}
\makeatother

\makeatletter
\def\@tocline#1#2#3#4#5#6#7{\relax
  \ifnum #1>\c@tocdepth 
  \else
    \par \addpenalty\@secpenalty\addvspace{#2}%
    \begingroup \hyphenpenalty\@M
    \@ifempty{#4}{%
      \@tempdima\csname r@tocindent\number#1\endcsname\relax
    }{%
      \@tempdima#4\relax
    }%
    \parindent\z@ \leftskip#3\relax \advance\leftskip\@tempdima\relax
    \rightskip\@pnumwidth plus4em \parfillskip-\@pnumwidth
    #5\leavevmode\hskip-\@tempdima
      \ifcase #1
       \or\or \hskip 2em \or \hskip 3em \else \hskip 4em \fi%
      #6\nobreak\relax
    \hfill\hbox to\@pnumwidth{\@tocpagenum{#7}}\par
    \nobreak
    \endgroup
  \fi}
\makeatother

\usepackage[affil-it]{authblk}
\usepackage{orcidlink}
\usepackage{amsaddr}
 
\newcommand\abs[1]{\left|#1 \right|}
\newcommand\norm[1]{\left\lVert#1 \right\rVert}

\newcommand{\jpp}[1]{{\langle #1 \rangle}}

\newcommand{\inp}[1]{ \left \langle #1 \right \rangle}
\newcommand{\ip}[2]{\langle #1, #2 \rangle}

\newcommand{\prc}[1]{ {\left( #1 \right)}}
\newcommand{\spr}[1]{ {\left[ #1 \right]}}
\newcommand{\brs}[1]{ \left\{ #1 \right\}}
\def\pt{\partial}

\newcommand{\cc}[1]{\overline{#1}}

\newcommand*\mcap{\mathbin{\mathpalette\mcapinn\relax}}
\newcommand*\mcapinn[2]{%
  \vcenter{\hbox{%
    \scalebox{0.7}{$\mathsurround=0pt
      \ifx\displaystyle#1\textstyle\else#1\fi\bigcap$}}}}

\newcommand*\mcupinn[2]{%
  \vcenter{\hbox{%
    \scalebox{0.7}{$\mathsurround=0pt
      \ifx\displaystyle#1\textstyle\else#1\fi\bigcup$}}}}

\begin{document} 
    

 

\title[Trapped bosons in mean field QED, nonlinear resonance cascades and dynamical BEC formation]{Trapped bosons in mean field QED, nonlinear resonance cascades and dynamical BEC formation}
    
\author{Thomas Chen$^\dagger$\orcidlink{0000-0003-2704-1454},\quad Ali Mezher$^\dagger$\orcidlink{0000-0002-9215-1804}}

\address{$^\dagger$Department of Mathematics, University of Texas at Austin, Austin, TX 78712, USA}
\email{tc@math.utexas.edu} 
\email{alimezher@utexas.edu}


\maketitle

\begin{abstract}
 In this paper, we study a system of bosons trapped in a confining potential, interacting with a quantized field of coherent photons in the mean field description of non-relativistic Quantum Electrodynamics (QED) obtained by N. Leopold and P. Pickl in \cite{leopold2017mean}. We derive the effective nonlinear cascade equations governing the emission and absorption of coherent photons by the boson subsystem in a combined weak coupling and macroscopic time scaling limit. We demonstrate that solutions to this nonlinear cascade determine a monotone decreasing energy flow in the boson subsystem, and thereby describe the dynamical formation of  a Bose-Einstein condensate (BEC) with full ground state occupation, under conservation of the total boson $L^2$ mass. We note that this process is crucially different from thermal relaxation to the ground state, and fundamentally depends on the nonlinear nature of the cascade dynamics.
\end{abstract}

\bigskip
\bigskip

\setcounter{tocdepth}{2}
\tableofcontents

\section{Introduction}
\label{intro}

In the present work, we study the behavior of the macroscopic coupled system arising from the mean-field limit of a non-relativistic quantum system of N bosons weakly coupled to a radiation field described by a coherent state (laser), as $N \to \infty$. More precisely, we consider the coupled system for the mild solution $(\vfi_t, u_t)$ to the mean-field PDE system derived by N. Leopold and P. Pickl in \cite{leopold2017mean}, namely, the system of nonlinear Hartree and half-wave equations
\begin{empheq}[left = \empheqlbrace]{align} 
i\partial_t \vfi_t &= \prc{-\Delta+V + \lambda (v \ast |\vfi_t|^2) + \eta \bigl(w \ast (u_t + \cc{u_t})\bigr)} \vfi_t,  \label{eq:intro-system-phi}\\
i\partial_t u_t &= \omega(-i\nabla)u_t + \eta \bigl(w \ast |\vfi_t|^2\bigr), \label{eq:intro-system-u} \\
(\vfi_0, u_0) & = (\vfi_t|_{t=0}, u_t|_{t=0}), \label{eq:intro-system-initial-data}
\end{empheq} 
where $V(x): \R^3 \to [0,\infty)$ is the confining potential, $v,w: \R^3 \to \R$ are two-body real-valued interaction functions that satisfy the hypothesis in \assumref{the-assumptions}. We will assume $\omega(\xi)=|\xi|$, $\xi\in\R^3$, $\eta>0$ small and $\lambda=\eta^2$. The initial data is in the product space
\begin{equation}
    (\vfi_0, u_0) \in H^1(\R^3) \times H^{1/2}(\mathbb{R}^3).
\end{equation}
For background on the derivation of effective PDEs from coupled quantum particle-field models, we refer to \cite{kirkpatrick2021large,kiscst,petrat2016new,pickl2011simple,rodnianski2009quantum, KP,pickl2,boers2016mean,fantechi2025quantum,ammari2022towards,chen2019mean}.  
\bigskip

\subsection{The effective resonance cascade}

The primary objective of this work is to rigorously derive and control the effective dynamics generated by the system in the weak-coupling, long-time regime defined by the scaling
\begin{equation}
T = \eta^2 t, \qquad \eta \to 0.
\end{equation}
We will refer to $t$ as the microscopic time and to $T$ as the macroscopic time. On the macroscopic timescale, second-order interactions with the radiation field accumulate, leading to an $\cO(1)$ evolution of the modal amplitudes of the trapped bosonic component. By projecting the dynamics onto the orthonormal eigenbasis of the confining Hamiltonian operator $\cH_0:=-\Delta+V$, we derive an effective equation for the sequence of probability amplitudes $\{F_k^\eta(T)\}_{k\ge0}$ associated with the discrete increasing, non-degenerate eigenenergy levels $\{E_k\}_{k\ge0}$ with the associated orthonormal eigenbasis $\{\chi_k\}_{k\ge0}$. We will assume $\chi_k$ to be real-valued, without any loss of generality.

The limiting evolution is given by a nonlinear \emph{resonance cascade equation} of the form
\begin{equation}
    \pt_T F_k(T) = \sum_{k'=0}^\infty M_{k,k'} |F_{k'}(T)|^2 F_k(T),
\end{equation}
where $M_{k,k'} \in \C$ contains the Hartree interaction, Lamb Shift, and Fermi Golden Rule contributions associated with the energy levels $E_k$ and $E_{k'}$. The transitions between discrete levels are strictly selected by the resonance condition on the field dispersion relation, namely, the condition is
\begin{equation}
    \label{eq:resonance-condition}
\omega(\xi) = |E_{k'} - E_k|.
\end{equation}
Crucially, the dynamics of this cascade are genuinely nonlinear. The leading resonant contribution is of Fermi-Golden-Rule (FGR) type. The transition rate at each level depends on the instantaneous occupations of other levels. The ground state, with no lower energetic outlet, creates highly asymmetric transition channels near the ground state. This nonlinear resonance setup drives the probability mass downward  under conservation of the boson $L^2$-mass, thereby creating a Bose-Einstein condensate (BEC) with full ground state occupation. Resonances in quantum theory have been extensively studied in the literature \cite{balslev1971spectral,simon1973resonances,simon1978resonances,okamoto1985complex, hunziker1990resonances,waxler1995time,soffer1995time,soffer1999resonances,king1991exponential, king1994resonant, abou2009theory,bach1998quantum,bach1999spectral,soffer1998nonautonomous,hunziker2005asymptotic,karpeshina2006perturbation,perry1981spectral}.
The proof that Bose-Einstein condensates correspond to the energy minimizer in interacting boson systems, and the derivation of nonlinear mean field equations to describe their effective dynamics has been a prominent research area for decades
\cite{MR1868990,MR2143817,MR2276262,MR4406847,MR3385343,MR3165917,MR3513141,MR3987176}, but the rigorous control of the dynamical formation of a BEC is currently an open problem, \cite{MR3742598,bach2022time,Reichl_2019}. The effective cascade equations emerging in the work at hand may, to some readers, appear vaguely similar to wave kinetic equations, which are derived in a significantly different setting; see, for instance, \cite{denghani2023inventiones,hanishatahzhu2024cpam} and the references therein.

\subsection{Spectral singularities and dispersive remainders}
The rigorous transition from the microscopic Hamiltonian PDE dynamics to the macroscopic cascade presents substantial analytical challenges. In the weak-coupling regime, the limit dynamics are governed by a transition matrix $M^\eta$. This matrix comprises an on-shell, energy-conserving resonant term (as in \eqref{eq:resonance-condition}), which represents the Fermi-Golden-Rule (FGR) transition rates between quantum states. It also includes an off-shell, principal-value term that corresponds to Lamb-shift energy renormalizations driven by quantum fluctuations.

These limit coefficients exhibit severe singularities due to small denominators in resolvents of the form $(|\nabla|-\Delta E+i\eta^2)^{-1}$. To rigorously resolve these singularities and establish well-defined macroscopic transition rates, we make use of a Limiting Absorption Principle (LAP) for the half-wave operator. Specifically, we establish uniform resolvent bounds between the polynomially weighted spaces $L^2_s (\mathbb{R}^3)=\{f\in L^2|\,\|\langle x\rangle^s f\|_{L^2}<\infty\}$ and $L^2_{-s}(\mathbb{R}^3)$ for $\frac12<s<1$. The spatial decay condition on the interaction kernel $w(x)$ then controls the frequency-space singularity at the resonance sphere. Consequently, demonstrating that these transition coefficients remain uniformly bounded as $\eta \to 0$ simultaneously validates the timescale $T=\eta^2 t$ and ensures the well-posedness of the limiting macroscopic cascade equations.

\smallskip

Uniform stationary bounds alone, however, are insufficient to guarantee strong convergence, as the non-resonant remainder terms—representing genuinely dispersive radiation—must also vanish on the macroscopic timescale. To address this, we invoke frequency-localized dispersive decay estimates for the half-wave propagator to prove that these remainder terms are governed by quantum fluctuations which disperse to spatial infinity fast enough that they vanish in the macroscopic limit.

\subsection{Main results}

With these analytical building blocks in place, we proceed to present the main theorems. We begin by specifying the functional framework and core assumptions on the Hamiltonian.

\begin{assumption}
\label{the-assumptions}
We assume that the following conditions on the confining potential $V$ and the interaction kernels $v, w$:
\begin{itemize}
\hitem{1}{asm:confining-potential} \textbf{Confining potential:} $V \in C^\infty(\mathbb{R}^3; \mathbb{R})$ is coercive, satisfying $V(x) \ge c_0 |x|^{1+\beta} - C_0$ for some $c_0, C_0>0$ and $\beta > 0$. Consequently, the operator $\cH_0 = -\Delta + V$ on $L^2(\mathbb{R}^3,dx)$ exhibits a purely discrete spectrum $\brs{E_k}_{k=0}^\infty$ with corresponding normalized eigenfunctions $\brs{\chi_k}_{k=0}^\infty$.

\hitem{2}{asm:spectral-genericity} \textbf{Spectral assumptions:} The spectrum of $\cH_0$ has the following properties.  
\begin{enumerate}
    \item The energy eigenvalues $E_j$ are rationally independent. That is, for any finite collection of indices $j_1 \ldots, j_m$, the relation $\sum_{l=1}^m n_l E_{j_l} = 0$ with integers $n_l$ implies that all $n_l$ are zero. 
    In addition, the four-mode gaps
    \begin{eqnarray}
        \delta_N:=\min_{\substack{0\leq k,k',j,j'\leq N \\
        |E_k-E_{k'}-E_j+E_{j'}|>0}}
        \Big\{
        \; |E_k-E_{k'}-E_j+E_{j'}|  \; 
        \Big\}
    \end{eqnarray}
    satisfy the Diophantine condition
    \begin{eqnarray}
        \delta_N\geq c_D\;(1+N)^{-\kappa}
    \end{eqnarray}
    for constants $c_D$, $\kappa>0$.
    \item 
The Fourier transform of products of eigenfunctions $\widehat{\chi_k \overline{\chi_{k'}}}$ do not vanish identically on any spheres centered at the origin.
\end{enumerate}

\hitem{3}{asm:interaction-generality} \textbf{Interaction regularity:} $v,w: \R^3 \to \R$ are even real-valued two-body interaction functions that satisfy the following conditions:
\begin{enumerate}
    \item The classical kernel $v \in W^{1,\infty}(\mathbb{R}^3,dx)$ is even. 
    \item The particle-field coupling kernel $w$ satisfies the polynomially weighted decay condition $w \in L^2_s(\mathbb{R}^3)$ for a fixed $\frac12<s<1$ with $s<\frac{1+\beta}2$.
    In addition, $\| |\nabla|^{-\frac12}w\|_{L^2_s}<C$, and $w \in H^{1/2}(\R^3)\mcap W^{3,1} (\R^3) \mcap L^\infty (\R^3)$, and by interpolation $ w \in L^{3/2}(\R^3)$. Moreover, $|\widehat w(\xi)|>0$ for all $\xi\in\R^3$.
\end{enumerate}
\end{itemize}
\end{assumption}

\begin{remark}
    In \lemref{lm-Vgeneric-spec-1-0}, and referring to \cite{MasonSigalotti2010}, we prove that for generic confining potentials $V$ endowed with the $L^\infty$ topology on bounded perturbations, the eigenvalues of $\cH_0$ are rationally independent. This residuality statement concerns the ambient space of confining potentials equipped with the $L^\infty$ topology on bounded perturbations, but it does not by itself prove residuality within the smooth subclass specified by \hypref{asm:confining-potential}. 
    
    We also note that given a potential $V_0$ satisfying \hypref{asm:confining-potential} so that $-\Delta+V_0$ has rationally dependent eigenvalues, a generic perturbation $V_0\rightarrow V_0+ V_1$ with $\|V_1\|_\infty<\epsilon$ suffices to render the spectrum rationally independent, see \lemref{lm-Vgeneric-spec-1-0}, for any $\epsilon>0$. Therefore, from the perspective of physical implementation, the rational independence of eigenenergies is a natural assumption.

    However, we note that the Diophantine property of four-mode gaps assumed in \hypref{asm:spectral-genericity} is not implied by \lemref{lm-Vgeneric-spec-1-0}, and we do not claim genericity for it.
\end{remark}

Let $(\vfi_t, u_t) \in C(\mathbb{R}; H^1(\mathbb{R}^3)) \times C(\mathbb{R}; H^{1/2}(\mathbb{R}^3))$ be the unique mild solution to the Cauchy problem \eqref{eq:intro-system-phi}–\eqref{eq:intro-system-u} subject to the initial data \eqref{eq:intro-system-initial-data}. We define the interaction picture modal amplitudes evaluated on the macroscopic time scale $T = \eta^2 t$ for $ t > 0 $ and $ \eta > 0$ by
\begin{equation}\label{eq:Feta-def-1-0-0}
F_k^\eta(T) := e^{i E_k T / \eta^2} \langle \chi_k, \vfi_{T/\eta^2} \rangle_{L^2_x}.
\end{equation}
Then, the effective nonlinear cascade dynamics in the weak-coupling limit are captured by \thmref{m:weak-coupling-limit}. Its proof follows from \eqref{eq:ODE-for-Fk(T)} and the bounds established in \secref{lamb-shift-and-resonant-bounds} and \secref{controlling-the-remainders}.

\begin{theorem}[Weak-coupling limit to the effective resonance cascade]
\label{m:weak-coupling-limit}
Assume \emph{\hypref{asm:confining-potential}-\hypref{asm:interaction-generality}}. 
Moreover, assume $\phi_0\in Y^1$, $u_0\in W^{3,1}$, with $\|\phi_0\|_{L^2}=1$, and a Hartree coupling $\lambda=\eta^2$.
For any fixed $T_0 > 0$, the sequence of modal amplitudes $F^\eta = \brs{F_k^\eta}_{k=0}^\infty$ converges strongly in $C([0, T_0]; \ell^2(\C))$ as $\eta \to 0$ to a limit $F \in C([0, T_0]; \ell^2(\C))$. Furthermore, $F(T)$ is the unique global solution to the effective nonlinear resonance cascade equation 
\begin{equation}
\label{eq:diag-cascade}
\partial_T F_k(T)=\sum_{k'\ge 0} M_{k,k'}\,|F_{k'}(T)|^2\,F_k(T)
\;\;\;,\;\;\;F_k(0)=\langle \chi_k,\phi_0\rangle\,,
\end{equation}
where the coefficients $M_{k,k'}:=M_{k,k';k,k'}+M_{k,k;k',k'}-\delta_{k,k'}M_{k,k;k,k}$, defined in \eqref{def:effective-coefficient-Mkkdjjd}, have the form 
\begin{eqnarray}
    \label{eq:M-form}
    M_{k,k'}&=& - i \Big(\Lambda_{k,k';k,k'}^{Har}
    +\Lambda_{k,k;k',k'}^{Har}
    -\delta_{k,k'}\Lambda_{k,k;k,k}^{Har}
    \nonumber\\
    &&\qquad- \Lambda_{k,k';k,k'}^{LS}
    - \Lambda_{k,k;k',k'}^{LS}
    + \delta_{k,k'}\Lambda_{k,k;k,k}^{LS}
    \Big)
    \nonumber\\
    &&
    - \FGRk \prc{\1_{k > k'} - \1_{k' > k}},
\end{eqnarray}
where
\begin{equation}
    \FGRk :=  \inp{w * (\chi_k \cc{\chi_{k'}}), \prc{\delta(\omg(-i \nabla) - |E_{k} - E_{k'}|)} w * (\cc{\chi_{k'}} \chi_{k} )} \geq 0
\end{equation}
 is symmetric in $(k,k')$. Here $\Lhar, \Lambda_{k,k';k,k'}^{LS} \in \R$ are the Hartree energy and Lamb shifts defined in \eqref{eq:definition-of-Lhar-coefficients}, and $\FGRk \in \R$ denotes the effective Fermi-Golden-Rule transition rates defined in \eqref{eq:definition-of-Gamma-kkdjjd}. The Fermi Golden Rule transitions in \eqref{eq:M-form} contain both the absorption of coherent photons for $k > k'$ and the emission of coherent photons for $k < k'$. 
\end{theorem}

\begin{lemma}
\label{lm-FGR-generic-1-1}
    Assume that the trapping potential $V$ satisfies hypotheses \hypref{asm:confining-potential} and \hypref{asm:spectral-genericity}, and the interaction potential $w$ satisfies \hypref{asm:interaction-generality} in \assumref{the-assumptions}, and $\omega(i\nabla)=|i\nabla|$. Then, writing $R_{k,k'}:=|E_k-E_{k'}|$, 
\begin{eqnarray}\label{eq-FGRk-pos-1-0}
    \FGRk &=& \int d\xi |\widehat w(\xi)|^2 \big|\widehat{\chi_k\widebar{\chi_{k'}}}(\xi)\big|^2  
    \delta(|\xi| - |E_k-E_{k'}|) 
    \nonumber\\
    &=& R_{k,k'}^2 \int_{S^2} d\sigma(\omega) \; 
    |\widehat w(R_{k,k'}\omega)|^2
    \big|\widehat{\chi_k\widebar{\chi_{k'}}}(R_{k,k'}\omega)\big|^2 
    \; \; >0
\end{eqnarray}
is nonzero for all $k,k'$ with $k\neq k'$. 
\end{lemma}

\begin{proof} 
%
By hypothesis \hypref{asm:spectral-genericity}, $\widehat{\rchi_k\widebar{\rchi_{k'}}}$ is not identically zero on any sphere $S_R$, and in particular not on the sphere of radius $R_{k,k'}=|E_{k} - E_{k'}|$.
The asserted positivity $\FGRk>0$ then follows from $|\widehat w(\xi)|>0$ assumed in \hypref{asm:interaction-generality}.
\end{proof}

The following result, which is proved in \secref{effective-resonance-cascade-eq-and-proof-of-BEC}, characterizes the global-in-time behavior of the autonomous macroscopic system \eqref{eq:diag-cascade}.

\begin{theorem}[Dynamical Bose-Einstein condensation]
\label{m:dynamical_bec}           
Under the hypotheses of \thmref{m:weak-coupling-limit}, including the quantitative Diophantine condition in \hypref{asm:spectral-genericity}, let $F^\eta(T)\in\ell^2$ be defined by \eqref{eq:Feta-def-1-0-0}, and let $F(T)\in\ell^2$ solve \eqref{eq:diag-cascade} with the same initial data $F(0)=F^\eta(0)\in\ell^2$. Then,  
\begin{equation}
    \lim_{\eta\rightarrow0}\sup_{T\in[0,T_\eta]}\norm{F^\eta(T) -  F(T) }_{\ell^2} =0\,,
\end{equation} 
for $T_\eta =C_\mu\log\langle 1/\eta\rangle$ where the constant is determined in \eqref{eq:mu-def-1-0} and \eqref{eq:Teta-def-1-0-0}. 

Assume that the initial ground state amplitude is strictly non-zero, i.e., $|F_0(0)| > 0$. Furthermore, as stated in \lemref{lm-FGR-generic-1-1}, assume that the transition rates to the ground state satisfy $\Gamma_{0,k'}^{FGR} > 0$ for all excited states $k' \ge 1$, although they may decay at an arbitrary rate as $k' \to \infty$. Then, 
\begin{enumerate}
    \item The $\ell^2$-mass is conserved, $\|F(T)\|_{\ell^2}=\|F(0)\|_{\ell^2}$ for all $T \ge 0$. In particular, the solution to the effective resonance cascade equation exists globally in time,
\begin{equation}
    \|F(T)\|_{\ell^\infty}\le \|F(T)\|_{\ell^2}=\|F(0)\|_{\ell^2}.
\end{equation}
    \item The solution to the effective resonance cascade equation exhibits complete Bose-Einstein condensation as $T \to \infty$, that is, the occupation mass is transferred toward lower-energy modes and converges entirely to the ground state. More precisely, assume $\sum_{k\geq0}|F_k(0)|^2=1$ with $0<|F_0(0)|\leq1$ and $\sum_{k\geq0}(E_k-E_0)|F_k(0)|^2<C$. Then, the occupation densities of all excited states decay to zero,
\begin{equation}
    \lim_{T \to \infty} \sum_{k' =1}^\infty \abs{F_{k'}(T)}^2 = 0.
\end{equation}
Consequently, the ground state absorbs the entire mass of the system,
\begin{equation}
    \lim_{T \to \infty} \abs{F_0(T)}^2 = \sum_{k=0}^\infty \abs{F_k(0)}^2 = 1,
\end{equation}
which constitutes the complete dynamical formation of a Bose-Einstein Condensate (BEC) in the boson subsystem. Furthermore, 
\begin{equation}
    \pt_T \sum_{k=0}^\infty E_k |F_k(T)|^2 \leq 0, \quad \text{and} \quad \lim_{T \to \infty} \sum_{k=0}^\infty E_k |F_k(T)|^2 = E_0 \norm{F(0)}_{\ell^2}^2,
\end{equation}
that is, the total energy of the particle subsystem decays monotonically with $T$, and concentrates in the ground state. 
\item Assume $0<|F_0(0)|\leq1$ and that only finitely many modes are excited at $T=0$. Let $K>0$ be the smallest index such that $F_k(0) = 0$ for all $k > K$, and let $\tilde{\Gamma}_K := \min_{1\leq k' \leq K} \Gamma_{0,k'}^{FGR}$. Then, the following bound on the convergence rate holds,
\begin{equation}
    \abs{F_0(T)}^2 \geq \frac{1}{1 + \frac{1 - \abs{F_0(0)}^2}{\abs{F_0(0)}^2} e^{- 2 \tilde{\Gamma}_K T}}
\end{equation} 
with $\norm{F(0)}_{\ell^2} = 1$.
\end{enumerate}
In particular, this implies 
\begin{eqnarray}
    \lim_{\eta\rightarrow0}|\braket{\chi_0,\phi_{T_\eta/\eta^2}}| & =& 1\;
    \nonumber\\
    \lim_{\eta\rightarrow0}
    |\braket{\chi_k,\phi_{T_\eta/\eta^2}}| & =& 0\;
    \rm{for \; all \;}k\geq1\,.
\end{eqnarray}
Consequently, by $L^2$ mass conservation, the solution $\phi_{T_\eta/\eta^2}$ to the original system \eqref{eq:boson-subsystem}, \eqref{eq:field-subsystem} converges to the ground state ray as $\eta\rightarrow0$ and $T_\eta\rightarrow\infty$.
\end{theorem}

The proofs are given in \thmref{thm:complete-BEC} and  \thmref{thm:convergence-eta}.

\begin{remark}
    We note that if \hypref{asm:spectral-genericity} is relaxed and some energy eigenvalues are allowed to be rationally dependent, the cascade dynamics becomes more intricate. We leave an analysis of this situation for future work.
\end{remark}

\begin{remark}
    In this work, we are deriving the cascade equations for macroscopic times $0<T<T_\eta$ with $T_\eta=C_\mu\log\braket{1/\eta}\rightarrow\infty$ in the limit $\eta\rightarrow0$ for a sufficiently small constant $C_\mu$, see \eqref{eq:mu-def-1-0} and \eqref{eq:Teta-def-1-0-0}. The scaling $t=T/\eta^2$ is used because the non-trivial cascade dynamics occurs at a time scale of $O(\eta^{-2})$; extracting it as the main mechanism leading to full ground state occupation is our main interest. 
    This implies complete ground state occupation in the original system \eqref{eq:intro-system-phi}, \eqref{eq:intro-system-u} in the diagonal limit $T_\eta\rightarrow\infty$ as $\eta\rightarrow0$. 
\end{remark}

\subsection{Further background}

We compare the system considered in this paper with several important, related but crucially different, dissipative systems.

\subsubsection{Particle in thermal photon field}
\label{sssec-Boltzmann-0-1}
We first consider the model of a quantum mechanical particle interacting with a photon field
in thermal equilibrium at inverse temperature $\beta$, with Hamiltonian
\begin{eqnarray}\label{eq-H-def-1-0}
\cH = \cH_{part} \otimes {\1} 
+ \eta \int dk \; \widehat w(k)
\Big( e^{-ikx}\otimes a_k +  e^{ikx}\otimes a_k^+\Big) 
+ {\1}\otimes \Hf
\end{eqnarray}
on the tensor product Hilbert space $\mathcal{H}=\mathcal{H}_{part}\otimes{\mathcal F}$ where $\mathcal{H}_{part}=L^2_x({\mathbb R}^3)$ is the particle Hilbert space, and $\cH_{part}=-\frac12\Delta_x$ is the particle Hamiltonian. ${\mathcal F}$ denotes the photon Fock space with vacuum vector $\Omega_f\in{\mathcal F}$, and creation and annihilation operators $a_k^+$, $a_k$ satisfying canonical commutation relations. $\cH_f=\int dk \, \omega(k) a_k^+ a_k$ is the photon field Hamiltonian.

In \cite{erdHos2002linear}, the linear Boltzmann equation is derived for the effective particle dynamics, in a kinetic scaling limit with macroscopic variables $(T,X)=\eta^2(t,x)$, while the photon degrees of freedom are averaged with respect to a Gibbs distribution, that is, $\langle A\rangle_{{\mathcal F}}=\frac1{Z_\beta}(e^{-\beta \Hf+\mu \cN_f}A)$ for Fock space observables $A$, where $\mu$ is the chemical potential, and where $\cN_f=\int dk a_k^+ a_k$ is the photon number operator. See also \cite{erdHos2000linear}.

The expected number of photons of energy $\omega(k)$ {\em absorbed} by the particle is   
\begin{eqnarray}
{\mathcal N}_{\beta,\mu}(k)=\frac{e^{-\beta \omega(k)+\mu}}{1-e^{-\beta \omega(k)+\mu }} \,,
\end{eqnarray}  
while the expected number of photons {\em emitted} is ${\mathcal N}_{\beta,\mu}(k)+1$. While the linear Boltzmann equation derived in \cite{erdHos2002linear} is spatially inhomogeneous, its spatially homogeneous counterpart has the form 
\begin{eqnarray} 
\partial_T f_T(V) = \int dU \sigma(U,V)(f_T(U)-f_T(V)).
\end{eqnarray} 
Here, $f_T(U)$ is the velocity space probability distribution for the particle, and the collision kernel is given by
\begin{eqnarray}
\sigma(U,V) &=& |\widehat w(U-V)|^2
\Big(({\mathcal N}_{\beta,\mu}(U-V)+1)
\; \delta(E(V)-E(U)+\omega(U-V)) 
\nonumber\\
&&
+ \; {\mathcal N}_{\beta,\mu}(U-V)
\; \delta(E(V)-E(U)-\omega(V-U))\Big) \,,
\end{eqnarray}
where $E(V)={V^2}/2$ is the kinetic energy of the particle.
The imbalance between emission and absorption rates drives $f_T(V)$ towards its thermal equilibrium solution, which can be shown to be of the form $f_\infty(V)= const. \, e^{-\beta E(V)}$, corresponding to a Maxwellian (respectively, a Gibbs state). That is, in the limit $T\rightarrow\infty$, the particle settles to its thermal equilibrium. It is important to note that in the zero-temperature limit, ${\mathcal N}_{\beta,\mu}(k)\rightarrow0$ as  $\beta\rightarrow\infty$, so that photon absorption disappears, and only photon emission persists. For other works on quantum Boltzmann dynamics, for instance, see \cite{chenhott,bach2022time,chen2025derivation,escobedo2015finite,escobedo2008derivation} and the references therein.
 
\subsubsection{Open quantum systems}
A fundamental problem in quantum statistical mechanics is to understand the thermal equilibration of a system of particles (electrons) confined to an atom or molecule that interacts with a radiation field in thermal equilibrium. The Hamiltonian still has the form \eqref{eq-H-def-1-0}, but the particle Hamiltonian now has the form $\cH_{part} =-\frac12\Delta_x+V(x)$, where $V$ is a confining potential (or a multi-particle version of this).

For a quantum mechanical atomic subsystem interacting with a quantized radiation field in thermal equilibrium, the system returns to its equilibrium state. That is, all excited eigenstates of $\cH_{part}$ turn into resonances due to the interaction with the thermal photon field. Therefore, any initial state of the atomic subsystem converges to the thermal equilibrium state as $t\rightarrow\infty$. For the full quantum field theoretical models, this is proven in \cite{bach1999spectral,bach1998quantum}. More recently, the return to equilibrium in systems of this type has been studied intensively in the context of quantum information theory; see \cite{ouyang2026approach,sigal2025propagation} and references therein.

\subsubsection{Effective Lindblad dynamics}
To determine the dynamical process of return to equilibrium, averaging over the equilibrium radiation field by taking a partial trace yields an effective model for the atomic subsystem described by an associated density matrix $\gamma_t$. Again defining a macroscopic time $T=\eta^2 t$, and taking a suitable limit, $\gamma_{T/\eta^2}$ converges to an effective density matrix $\rho_T$ as $\eta\rightarrow0$, which is governed by a Lindblad equation
\begin{eqnarray}
\partial_T \rho_T = \frac1i [\cH,\rho_T] + \sum_i\Big(L_i \rho_T L_i^*-\frac12\{L_i^* L_i,\rho_T\}\Big)
\end{eqnarray} 
where $L_i$, $i\in{\mathbb N}$, are linear operators accounting for the coupling of the subsystem to the thermal photon field, and $\{\,\cdot\,,\,\cdot\,\}$ is the anticommutator. Its trace is conserved, $\operatorname{Tr}(\rho_T)=\operatorname{Tr}(\rho_0)$, and $\rho_T$ converges to a thermal equilibrium state as $T\rightarrow\infty$. See
\cite{davies1974markovian,davies1976markovian,falconi2017scattering,jakvsic1997resonances,lindblad1976generators,gorini1976completely,ouyang2026approach,sigal2025propagation,fantechi2025quantum} and the references therein.

\subsubsection{Rayleigh scattering}
Another fundamentally important system is that of electrons confined to an atom or molecule interacting with the quantized radiation field in the framework of Rayleigh scattering in non-relativistic Quantum Electrodynamics. Initially excited atomic states undergo resonant decay, and asymptotically as $t\rightarrow\infty$, the coupled system factors into the renormalized ground state of the atomic system and scattering out states characterized by photons of frequencies determined by the eigenenergy differences in the atomic system (Bohr's frequency condition). See \cite{bach2007renormalized,bach2006infrared,bach2007infrared,frohlich2009spectral,frolich2002asymptotic}.

The system here is at zero temperature (there is no interaction with a thermal photon field), and thus, as discussed in \secref{sssec-Boltzmann-0-1}, it is driven by the resonant decay of excited particle states through the emission of photons, but photon absorption processes are negligible.  

\subsubsection{Coherent radiation field}

As a model for laser cooling, applied in the experimental creation of Bose-Einstein condensates, the following system is important. One assumes that the photons are described by a coherent state in ${\mathcal F}$, at initial time $t=0$, with initial wave function $u_0$. In a suitable mean field limit, it is proven that for $t>0$, the photon state continues to be coherent, of the form $\exp(a^+(u_t)-a(u_t))\Omega_f$. Assuming, in addition, a subsystem of $N$ particles with mean-field pair interactions coupled to the photons, it is proven in \cite{leopold2017mean} that as $N\rightarrow\infty$, one obtains the system \eqref{eq:intro-system-phi}-\eqref{eq:intro-system-initial-data} studied in the work at hand.

In contrast to the systems above, we obtain here that the rates of photon emission and absorption are {\em perfectly balanced}. This is manifested in the fact that both $u_t$ (emissions) and $\cc{u_t}$ (absorptions) in \eqref{eq:intro-system-phi} share the same coefficients. The radiation field is not in thermal equilibrium, as it is coherent. The $L^2$ mass of the particle subsystem is in itself conserved, as the bosons cannot escape the trap, given our assumptions \hypref{asm:confining-potential} on the potential $V$. As stated in \thmref{m:weak-coupling-limit}, we prove that with macroscopic time $T=\eta^2 t$, and in the limit $\eta\rightarrow0$, the eigenmodes in the particle subsystem solve a nonlinear resonance cascade equation that includes both photon emission and absorption contributions. In this context, $\ell^2$ conservation of the sequence of modes is a consequence of the fact that photon emission and absorption are balanced. 

In particular, we prove in \thmref{m:dynamical_bec} that all excited modes deplete as $T\rightarrow\infty$, while the ground state mode eventually absorbs all of the $L^2$ mass, which thereby forms a BEC. To understand the intuition behind this mechanism, we define, for any $K>0$ indexing an excited mode, the accumulated $L^2$ mass of all modes $k > K$, $m_K(T):=\sum_{k>K}|F_k(T)|^2$.  Then,

\begin{equation}
    \pt_t m_K(T) = 2 \sum_{k > K} \sum_{k' \geq 0} \FGRk \prc{\1_{k' > k} - \1_{k > k'}} |F_k(T)|^2 |F_{k'}(T)|^2.
\end{equation}
Hence, all modes $k' >k$ increase $m_K(T)$ through photon absorptions, while all modes $k' < k$ decrease $m_K(T)$ through photon emissions.  Contributions with both $k, k' > K$ are balanced and cancel each other exactly, while modes with $k' \leq K$ all yield negative contributions. We obtain in \thmref{thm:complete-BEC} that $m_K(T) \rightarrow 0$ for all $K > 0$, as $T\rightarrow\infty$. Therefore, the entire $L^2$ mass must concentrate in the ground state. {\em The nonlinear nature of the resonance cascade equations is absolutely crucial to enable this mechanism.}

\subsection{Organization of the paper}
\label{organization-of-the-paper}

The remainder of this paper is organized to provide a complete progression from the microscopic Hamiltonian dynamics to the BEC formation at the macroscopic time scale. In \secref{analysis-of-the-mean-field-equations}, we establish the functional setting, notation, and analytic preliminaries. \secref{conservation-laws} and \secref{lamb-shift-and-fermi-golden-rule} are devoted to the analysis of the limit mean-field PDE system, where we derive the fundamental conservation laws for the total $L^2$-mass and isolate the resonant interactions, rigorously identifying the microscopic Hartree, Lamb-shift, and Fermi-Golden-Rule operators. 

\smallskip

In \secref{effective-resonance-cascade-eq-and-proof-of-BEC}, we characterize the global-in-time behavior of the autonomous effective cascade equation, proving the monotone flow of energy towards lower spectral modes and establishing the dynamical formation of Bose-Einstein condensation. \secref{convergence-analysis:limit-to-effective-cascade} determines the error dynamics in the Duhamel formulation, explicitly setting up the framework to prove $F^\eta \to F$ as a strong limit as $\eta \to 0$.

\smallskip

The core functional analytic estimates required to validate this limit are presented in the final two sections. \secref{lamb-shift-and-resonant-bounds} resolves the singularity modes of the resonance shell (arbitrarily small denominators) of the effective interaction matrix; we deploy a half-wave version of the Limiting Absorption Principle in polynomially weighted spaces to prove the uniform $\mathcal{O}(1)$ boundedness of the resonant transition rates. \secref{controlling-the-remainders} establishes analytic control by invoking dispersive PDE estimates to prove that the non-resonant radiation remainders converge to zero in the limit $\eta \to 0$. 
\smallskip
Finally, \appref{appendix:analytic-preliminaries} collects the necessary technical tools from harmonic analysis and spectral theory. \appref{sec-appB-0-1} established global well-posedness of the initial value problem for $(\phi_t,u_t)$. 

\section{Analysis of mean field equations}
\label{analysis-of-the-mean-field-equations}
We consider the mean-field equations describing the system of trapped bosons in $ \R^3$ coupled to coherent photons (for instance, modeling a laser) with wave function $u_t$,

 \begin{empheq}[left=\empheqlbrace]{align}
i \pt_t \vfi_t &= \cH[u_t, \vfi_t]\vfi_t, \label{eq:boson-subsystem} \\
i \pt_t u_t&= \omg(-i \nabla_y) u_t + \eta (w * |\vfi_t|^2)\label{eq:field-subsystem}
\end{empheq}
where $v, w$ satisfy \assumref{the-assumptions}, and
\begin{equation}
    \label{eq:definition-of-H[u,vfi]}
 \cH[u_t,\vfi_t]:= -\Delta + V + \lam (v * |\vfi_t|^2) + \eta\, (w * (\cc{u_t}+u_t)),
\end{equation}
with initial data $ \vfi_0 \in H^1(\R^3)$, and $u_0 \in H^{1/2}(\R^3).$

We assume that the Schr\"{o}dinger operator of the free particle (boson) system has a purely discrete spectrum,
\begin{equation}
    \sigma ( - \Delta + V) = \brs{E_k: k \in \N_0},
\end{equation}
where  $E_0 < E_1  < E_2 < \ldots$ are eigenvalues, with ground state energy $E_0$. The corresponding orthonormal eigenfunctions are denoted by $\set{\rchi_k}_{k\geq1} \subset L^2(\R_x^3)$, that is,
\begin{equation}
    \prc{-\Delta + V} \, \rchi_k = E_k \, \rchi_k, \quad k \in \N_0.
\end{equation}   
We assume that the eigenvalues  $E_k$ are non-degenerate. Furthermore, it follows from \assumref{the-assumptions} that the energy differences are rationally independent. Specifically, if for any $k < k'$ and $j < j'$ (i.e.  $E_{k'}  -  E_{k} > 0$ and $E_{j'} - E_{j} > 0$) the relation 
\begin{equation}
    n (E_{k'} - E_{k}) = m (E_{j'} - E_{j})
\end{equation}
holds for nonzero integers $n,m \in \N$, then $n=m$ and
\begin{equation}
    E_k = E_j, \qquad E_{k'} = E_{j'}.
\end{equation}
In particular, for $k < k'$ and $j < j'$, 
\begin{equation}
    \label{eq:resonance-condition-for-energy-differences}
    \Delta E_{k,k';j,j'} := E_{k'} - E_{k}  -  (E_{j'} - E_{j})=0 \implies k = j, \quad k'=j'
\end{equation}
follows.

Moreover, we assume without any loss of generality that the eigenfunctions $\rchi_k: \R^3 \to \R$ are real-valued. Indeed, $\overline{(-\Delta+V)\rchi_k}=(-\Delta+V)\overline{\rchi_k}=E_k\overline{\rchi_k}$
implies that both $\rchi_k$ and its complex conjugate $\overline{\rchi_k}$ are eigenvectors; but since $E_k$ is non-degenerate, they are linearly dependent and differ at most by a constant phase factor $e^{i\theta}$, which can be chosen to equal 1.

\subsection{Local and global well-posedness}
\label{subsec:lwp-gwp}

In this subsection, we establish local well-posedness of the coupled mean-field
system. A well-known subtlety arises when treating the external potential $V$, in that the
propagator $U_t = e^{-it(-\Delta+V)}$ is not unitary on the standard Sobolev space
$H^1(\mathbb{R}^3)$ because the gradient does not commute with $V$. To account for
this circumstance, we perform the fixed-point argument in the form domain of the Hamiltonian.

Let $\cH_0 := -\Delta + V$. Since we assume $V(x) \ge -C_0$, the operator $\cH_0$ is 
bounded from below but not strictly positive. To construct a positive-definite norm, we define the shifted, strictly positive operator
\begin{equation}
\cH_{C_0} := -\Delta + V + C_0 + 1 \ge 1.
\end{equation}
The natural energy space for the particle is the form domain $Y^1(\mathbb{R}^3)$, 
defined as the completion of $C_c^\infty(\mathbb{R}^3)$ under the norm
\begin{equation}
\label{eq:Y1-norm-def}
\|\vfi_t\|_{Y^1}^2
:=
\|\cH_{C_0}^{1/2}\vfi_t\|_{L^2}^2
=
\|\nabla \vfi_t\|_{L^2}^2 + \int_{\mathbb{R}^3} V(x)|\vfi_t(x)|^2 \, dx + (C_0+1)\|\vfi_t\|_{L^2}^2.
\end{equation}
By the spectral theorem, the flow $e^{-it\cH_0}$ commutes with $\cH_{C_0}^{1/2}$, 
and is therefore an exact isometry on $Y^1(\mathbb{R}^3)$. Furthermore, because 
$V + C_0 + 1 \ge 1$, we have the continuous embedding $Y^1(\mathbb{R}^3) \hookrightarrow H^1(\mathbb{R}^3)$.

\bigskip

Now, consider the system
\begin{equation}
\label{eq:mf-system-wp}
\begin{cases}
i\partial_t \vfi_t=\Bigl(\cH_0 + \lam (v*|\vfi_t|^2) + \eta \, w*(u_t+\overline{u_t})\Bigr)\vfi_t,\vspace{0.5 em}\\ 
i\partial_t u_t=\omega(-i\nabla)\,u_t+\eta \,(w*|\vfi_t|^2)
\end{cases}
\end{equation}
with initial data
\begin{equation}
\vfi_0\in Y^1(\mathbb{R}^3),\qquad u_0\in H^{1/2}(\mathbb{R}^3).
\end{equation}

\begin{remark}[Compatibility of the dispersive initial data]
\label{rem:field-data-compatibility}
While the global well-posedness theory and energy conservation strictly require only finite energy for the initial field, $u_0 \in H^{1/2}(\mathbb{R}^3)$, the overarching analysis imposes the stronger assumption $u_0 \in W^{3,1}(\mathbb{R}^3)$ to guarantee sufficient time decay of the free dispersive evolution. By the embedding $W^{3,1}(\mathbb{R}^3) \hookrightarrow H^{1/2}(\mathbb{R}^3)$, which ensures the bound
\begin{equation}
\|u_0\|_{H^{1/2}(\mathbb{R}^3)} < \|u_0\|_{W^{3,1}(\mathbb{R}^3)}.
\end{equation}
Thus, the highly regular initial data requisite for the dispersive bounds trivially possesses finite energy, making the decay assumptions perfectly compatible with the global well-posedness framework.
\end{remark}

\begin{proposition}[Local well-posedness in the form domain]
\label{prop:lwp-mf-strong-final}
Assume that:
\begin{enumerate}
    \item $V:\mathbb{R}^3\to\mathbb{R}$ satisfies $V(x) \ge -C_0$ and $\cH_0$ is essentially self-adjoint on $L^2(\mathbb{R}^3)$;
    \item $v\in W^{1,\infty}(\mathbb{R}^3)$;
    \item $w\in H^{1/2}(\mathbb{R}^3)\mcap W^{1,1}(\mathbb{R}^3)$;
    \item $\vfi_0\in Y^1(\mathbb{R}^3)$ and $u_0\in H^{1/2}(\mathbb{R}^3)$.
\end{enumerate}
Then there exists time $t_0>0$ such that the system \eqref{eq:mf-system-wp} admits a unique local in time mild solution
\begin{equation}
(\vfi_t,u_t)\in C([0,t_0];Y^1(\mathbb{R}^3)\times H^{1/2}(\mathbb{R}^3)).
\end{equation}
In particular, the particle and field components satisfy the Duhamel formulas
\begin{align}
\vfi_t(x) &= e^{-it\cH_0}\vfi_0(x) - i\int_0^t e^{-i(t-\tau)\cH_0} \Bigl( \lambda (v*|\vfi_\tau(x)|^2) + \eta \, w*(u_\tau(y)+\overline{u_\tau(y)}) \Bigr)\vfi_\tau(x) \, d\tau,\label{eq:duhamel-particle-LWP} \\
u_t(y) &= e^{-it\omega(-i\nabla)}u_0(y) - i\eta\int_0^t e^{-i(t-\tau)\omega(-i\nabla)} (w*|\vfi_\tau(x)|^2) \, d\tau.
\label{eq:duhamel-field-LWP}
\end{align}
\end{proposition}
\noindent
The proof of \propref{prop:lwp-mf-strong-final} is given in \appref{app:lwp-mf-strong-final}.

\begin{theorem}[Global well-posedness in the form domain]
\label{thm:gwp-strong-final}
Assume the hypotheses of \propref{prop:lwp-mf-strong-final} and  \lemref{lem:energy-coercive-lower-bound} hold. If the boson $L^2$-mass and total energy are conserved along local solutions, then for every initial datum $(\vfi_0,u_0)\in Y^1(\mathbb R^3)\times H^{1/2}(\mathbb R^3)$, the mild solution exists globally in time
\begin{equation}
(\vfi_t,u_t)\in C([0,\infty);Y^1(\mathbb R^3)\times H^{1/2}(\mathbb R^3)).
\end{equation}
\end{theorem}

\begin{proof}
Let $(\vfi_t,u_t)$ be the maximal local solution on $[0,T_{\max})$. By the conservation of both the energy $\mathcal{E}$ and the mass $\|\vfi_t\|_{L^2}^2 = \|\vfi_0\|_{L^2}^2$, the left-hand side of \eqref{eq:energy-lower-bound-final} is strictly conserved and equal to $\mathcal{E}[\vfi_0,u_0] + \frac{C_0+1}{2}\|\vfi_0\|_{L^2}^2$. 
Thus, \lemref{lem:energy-coercive-lower-bound} yields
\begin{equation}
\frac12\|\vfi_t\|_{Y^1}^2 + \frac14\|\omega(-i\nabla)^{1/2}u_t\|_{L^2}^2 \le \mathcal{E}[\vfi_0,u_0] + \frac{C_0+1}{2}\|\vfi_0\|_{L^2}^2 + C\|\vfi_0\|_{L^2}^4.
\end{equation}
This provides the uniform a priori bound $\sup_{t < T_{\max}} \|\vfi_t\|_{Y^1}^2 < \infty$. 
Simultaneously, \lemref{lem:field-Hs-propagation} ensures that if $T_{\max} < \infty$, then $\sup_{t < T_{\max}} \|u_t\|_{H^{\frac12}}^2 < \infty$. 
Combining these, we obtain
\begin{equation}
\sup_{0\le t<T_{\max}}
\Bigl(
\|\vfi_t\|_{Y^1(\mathbb R^3)}^2
+
\|u_t\|_{H^{1/2}(\mathbb R^3)}^2
\Bigr)
<\infty.
\end{equation}
This contradicts the blow-up alternative for the local fixed-point map. Therefore, we obtain that 
$$T_{\max}=\infty.$$
This proves the claim. 
\end{proof}

\subsection{Conservation laws} 
\label{conservation-laws}
The coupled system of PDEs \eqref{eq:boson-subsystem} and \eqref{eq:field-subsystem} possesses the following conservation laws: The $L^2$-norm of the particle subsystem  \eqref{eq:boson-subsystem} is by itself conserved, that is 
\begin{equation}
    \label{eq:conservation-of-L2-norm-of-vfi}
    \pt_t \norm{\vfi_t}_{L^2}^2 = -{i} \inp{\vfi_t, \prc{ \cH[u_t, \vfi_t] - \cH[u_t, \vfi_t]} \vfi_t} = 0,
\end{equation}
but the $L^2$-norm of the photon subsystem \eqref{eq:field-subsystem}, due to 
\begin{equation}
    \pt_t \norm{u_t}^2_{L^2} = 2 \Im{\inp{u_t, \eta \, (w* |\vfi_t|^2)}} \neq 0,
 \end{equation}
is not conserved.

The total energy of the coupled system of bosonic particles and coherent photons,
\begin{equation}
    \cE[u_t, \vfi_t] = \inp{\vfi_t, \prc{\frac{1}{2} \prc{-\Delta + V} + \frac{\lam}{4} (v * |\vfi_t|^2) + \frac{\eta}{2} w * (\cc{u_t} + u_t)} \vfi_t}
    +\frac12\langle u, \omega(-i\nabla) u\rangle
\end{equation}
is conserved, that is, $\pt_t \cE[u_t, \vfi_t]=0$.

\subsection{Lamb shift and Fermi Golden Rule}  
\label{lamb-shift-and-fermi-golden-rule}
This section presents an analysis of resonant transitions between excited bosonic (particle) states due to the absorption and emission of coherent photons. We write the boson wave function as an expansion in terms of the eigenfunctions
\begin{equation}
    \label{eq:excited-state-ansatz-for-vfi}
    \vfi_t(x) = \sum_{k=0}^\infty A_k(t) e^{-i t E_k } \rchi_k(x).
\end{equation}
Equivalently,
\begin{equation}
    || \vfi_t ||_{Y^1}^2 = \sum_k E_k | A_k(t) |^2 + ( 1 + C_0 ) \sum_k  | A_k(t) |^2
\end{equation}
Therefore, the density of the bosons is given by
\begin{equation}
    \abs{\vfi_t(x)}^2 = \sum_{j,j'} A_j(t) \cc{A_{j'}(t)} e^{i t (E_{j'} - E_{j})} \rchi_j(x) \cc{\rchi_{j'}}(x).
\end{equation}
It will be established that the amplitude of the excited states $A_k(t)$, for $k \geq 1$, varies very slowly in time $t$, with $\abs{\pt_t {A_k(t)}} = \cO(\eta^2)$. 
\smallskip
First, we solve the mean-field equation \eqref{eq:field-subsystem} for the effective classical field $u_t$ by using the Duhamel formula,
\begin{equation}
    u_t(y) = e^{-i t \omg(-i \nabla)} u_0(y) - i \eta \int_0^t ds \, e^{-i (t-s) \omg(-i \nabla)} (w * |\vfi_s(x)|^2).
\end{equation}
The second term on the RHS is a source term that feeds energy from the particle subsystem into the coherent radiation field amplitude $u_t$. Substituting the above expression for $u_t$ into the particle equation \eqref{eq:boson-subsystem},
\begin{equation}
    i \pt_t \vfi_t = \prc{-\Delta + V  + \lam (v * |\vfi_t|^2)} \vfi_t  +  \Resa(t) + \Rema(t),
\end{equation}
where
\begin{equation}
    \label{eq:Resa(t) definition}
    \Resa(t) = \Resaa(t) + \Resab(t),
\end{equation}
is the resonant term, with 
\begin{align}
    \Resaa(t) &= - i \eta^2 \prc{w * \int_0^t ds \, e^{-i (t-s) \omg(-i \nabla)} (w * |\vfi_s|^2)} \vfi_t,\label{eq:Resaa(t)-term}\\ 
    \Resab(t) &= i \eta^2 \prc{w * \int_0^t ds \, e^{i (t-s) \omg(-i \nabla)} (w * |\vfi_s|^2)} \vfi_t, \label{eq:Resab(t)-term}
\end{align}
and 
\begin{equation}
    \label{eq:Rema(t)}
    \Rema(t) = 2\eta (w *  \Re{ e^{-i t \omg(-i \nabla)} u_0}) \vfi_t 
\end{equation}
is the initial remainder term.

Using the ansatz \eqref{eq:excited-state-ansatz-for-vfi},  we introduce the notations
\begin{align}
    \Reskaa(t) &:= e^{i t  E_k} \inp{\rchi_k, \Resaa(t) },\label{eq:Reskaa(t)-definition}\\
    \Reskab(t) &:= e^{i t  E_k} \inp{\rchi_k, \Resab(t) },\label{eq:Reskab(t)-definition}\\
    \Remka(t) &:= e^{i t  E_k} \inp{\rchi_k, \Rema(t) },\label{eq:Remka(t)-definition}
\end{align}
and taking the inner product with $\rchi_k$,

\begin{align}
    \label{eq:ODE-for-Ak(t)}
    i \pt_t A_k(t) &= \lam \sum_{k', j,j'} \Lhar e^{i t (E_{k'} - E_k - (E_{j'} - E_{j}))} \cc{A_{j'}(t)} A_j(t) A_{k'}(t)\\
     &+ \Reskaa(t) + \Reskab(t) + \Remka(t).
\end{align}
Here, the coefficients
\begin{equation}
    \label{eq:definition-of-Lhar-coefficients}
    \Lhar := \inp{\rchi_k, v *(\cc{ \rchi_{j'}} \rchi_j) \rchi_{k'}} = \inp{\cc{\rchi_{k'}} \rchi_k, v * (\cc{\rchi_{j'}} \rchi_j)} \in \R,
\end{equation} account for the Hartree interaction term. They are real-valued because $v$ and the orthonormal basis vectors $\set{\rchi_k}$ are real-valued functions. Here and in the sequel, for greater generality, we will nevertheless keep track of the expressions for the complex-valued basis vectors. Thus, we write $\cc{\rchi_{k'}}$ instead of $\rchi_{k'}$ in the inner product.

We determine the resonant terms as follows,

\begin{align}
    &\Reskaa(t) = e^{i t E_k} \inp{\rchi_k, \Resaa(t)}\nonumber\\*
=& \sum_{k', j,j'} \inp{\rchi_k, - i \eta^2 w* \prc{\int_0^t ds \,  e^{-i (t-s) \omg(-i \nabla)} (w * (\rchi_j \cc{\rchi_{j'}}))  e^{i s (E_{j'} - E_j)} A_j(s) \cc{A_{j'}(s)}} \rchi_{k'}}\nonumber\\*
& \times A_{k'}(t) e^{i t (E_k - E_{k'})}\\*
     =&- i\eta^2  \sum_{k', j,j'} \int_0^t ds\,  \inp{ w * (\cc{\rchi_{k'}} \rchi_{k}), e^{- i(t-s)[\omg(-i \nabla) - (E_{j} - E_{j'})]} w * (\cc{\rchi_{j'}} \rchi_{j})} \nonumber\\*
    & \times e^{i t [(E_{k} - E_{k'})- (E_{j} - E_{j'})]} A_j(s) A_{j'}(s) A_{k'}(t).
\end{align}

Therefore, we have
\begin{align}
    \Reskaa(t)  =& - \eta^2 \sum_{k', j,j'}   (\Lm + i \Gm) e^{i t (E_{k} - E_{k'} - (E_{j} - E_{j'}))} A_j(t) A_{j'}(t) A_{k'}(t)\nonumber\\
    & + \eta^2 \Remkaa(t),
\end{align}
 where $\Lm$ and $\Gm$ are real-valued (since $w$ and $\rchi_k$ are real-valued functions), and determine the Lamb shift and Fermi Golden Rule, respectively. In particular,
\begin{align}
    \Lm  + i \Gm  :=& \inp{w * (\cc{\rchi_{k'}} \rchi_k), \prc{\frac{1}{{\omega(-i \nabla)} - (E_{j} - E_{j'}) - i 0^+}} w * (\cc{\rchi_{j'}} \rchi_{j})} 
\end{align}
are bounded. \lemref{lm:FGR-LS-a-priori-1-0} yields an a priori bound; the FGR coefficients are, in fact, uniformly bounded in $k,k',j,j'$, as follows from \lemref{lem:uniform-FGR-bound}.

Moreover, we have introduced the remainder term
\begin{align}
    \label{eq:Remkm}
    \Remkaa(t) =& \mathcal{Rem}^{(-),\mathrm{bd}}_k(t) + \mathcal{Rem}^{(-),\mathrm{der}}_k(t) 
\end{align}
where
\begin{align}
    \label{eq:Remkmbd}
    \mathcal{Rem}^{(-),\mathrm{bd}}_k(t) :=&  \sum_{k', j,j'}   \inp{w *(\cc{\rchi_{k'}}\rchi_k), \prc{\frac{ e^{- i t \omega(-i \nabla)} }{\omg(-i \nabla) - (E_{j} - E_{j'}) - i 0^+}}  w *(\cc{\rchi_{j'}} \rchi_{j}) } \nonumber\\
    & \times e^{i t (E_{k} - E_{k'})} A_j(0) \cc{A_{j'}}(0) A_{k'}(t) 
\end{align}
is the boundary term obtained from integration by parts, and
\begin{align}
\label{Rmkmderminus} 
    \mathcal{Rem}^{(-),\mathrm{der}}_k(t)
    &:=
    \sum_{k',j,j'}
    \int_0^t ds\,
    \Bigg\langle
        w*(\overline{\rchi_{k'}}\rchi_k),
        \frac{1}{
            \omega(-i\nabla)-(E_j-E_{j'})-i0^+
        }
        \nonumber\\
    &\hspace{3.0cm}\times
        e^{-i(t-s)\omega(-i\nabla)}
        w*(\overline{\rchi_{j'}}\rchi_j)
    \Bigg\rangle
    \nonumber\\
    &\hspace{1.5cm}\times
    e^{it(E_k-E_{k'})}
    e^{is(E_j-E_{j'})}
    \partial_s
    \left(
        A_j(s)\overline{A_{j'}(s)}
    \right)
    A_{k'}(t)
\end{align}
contains the time derivative of the amplitudes.

Furthermore, we have 
\begin{align}
    \Reskab(t) =& e^{i t E_k} \inp{\rchi_k, \Resab(t)} \\*
        =& \sum_{k', j,j'} \inp{\rchi_k, i \eta^2 \prc{ w * \int_0^t ds \,  e^{i (t-s) \omg(-i \nabla) - is(E_j - E_{j'})}\, w * (\rchi_j \rchi_{j'})   } \rchi_{k'} }\nonumber \\
    & \times  A_j(s) A_{j'}(s)  A_{k'}(t) e^{i t (E_k - E_{k'})} \\
     =& i\eta^2 \sum_{k', j,j'}  \int_0^t ds \,   \inp{w * (\cc{\rchi_{k'}} \rchi_k),   \prc{e^{i (t-s)[\omg(-i \nabla)+ (E_{j} - E_{j'})]}} w * (\cc{\rchi_{j'}} \rchi_{j} } \nonumber \\
     & \times e^{i t [(E_{k} - E_{k'}) - (E_{j} - E_{j'}))]}  A_j(s)  \cc{A_{j'}(s)} A_{k'}(t).
     \end{align}
Therefore, 
\begin{align}
    \label{eq:Remkp}
    \Reskab(t) =& - \eta^2 \sum_{k', j,j'}   (\Lp + i \Gp) e^{i t (E_{k} - E_{k'} - (E_{j} - E_{j'}))}  A_j(t) A_{k'}(t) \cc{A_{j'}}(t)
    \nonumber\\
    & + \eta^2 \Remkab(t),
\end{align}
where $\Lp$ and $\Gp$ are real valued, and

\begin{align}
    \Lp  + i \Gp  :=& \inp{w * (\cc{\rchi_{k'}} \rchi_k), \prc{\frac{1}{{\omega(-i \nabla)} + (E_{j} - E_{j'}) + i 0^+}} w * (\cc{\rchi_{j'}} \rchi_{j})},
\end{align}
and we obtain the remainder term
\begin{align}\label{eq:Remkmplus}
    \Remkab(t) =&   \mathcal{Rem}^{(+),\mathrm{bd}}_k(t)+\mathcal{Rem}^{(+),\mathrm{der}}_k(t)
\end{align}
where
\begin{align}\label{eq:Remkmplusbd}
    \mathcal{Rem}^{(+),\mathrm{bd}}_k(t) :=&   
    \sum_{k', j,j'}   \inp{ w * (\cc{\rchi_{k'}} \rchi_k), \prc{\frac{1}{\omega(-i \nabla) + (E_{j} - E_{j'}) + i 0^{+}}}e^{ i t \omega(-i \nabla)} w *(\cc{\rchi_{j'}} \rchi_{j}) }  \nonumber\\
    & \times e^{i t (E_{k} - E_{k'})} A_j(0) \cc{A_{j'}(0)} A_{k'}(t)
\end{align}
and
\begin{align}
    \mathcal{Rem}^{(+),\mathrm{der}}_k(t)
    &:=
    \sum_{k',j,j'}
    \int_0^t ds\,
    \Bigg\langle
        w*(\overline{\rchi_{k'}}\rchi_k),
        \frac{1}{
            \omega(-i\nabla)+(E_j-E_{j'})+i0^+
        }
        \nonumber\\
    &\hspace{3.0cm}\times
        e^{i(t-s)\omega(-i\nabla)}
        w*(\overline{\rchi_{j'}}\rchi_j)
    \Bigg\rangle
    \nonumber\\
    &\hspace{1.5cm}\times
    e^{it(E_k-E_{k'})}
    e^{is(E_j-E_{j'})}
    \partial_s
    \left(
        A_j(s)\overline{A_{j'}(s)}
    \right)
    A_{k'}(t)
\label{Rmkmderplus}
\end{align}
in analogy to $\mathcal{Rem}^{(-),\mathrm{bd}}_k(t)$ and $\mathcal{Rem}^{(-),\mathrm{der}}_k(t)$.
\smallskip
The relation between $\Lm + i \Gm$ and $\Lp + i \Gp$ is given by the following symmetry properties,
\begin{equation}
   \Lp + i \Gp = \cc{\Lambda_{k,k';j',j} ^{-} + i \Gamma_{k,k';j',j} ^{-}}.
\end{equation}
Since  $\brs{\rchi_j}_j$ are real-valued functions, one gets 
\begin{equation}
\cc{ \Lm + i \Gm} = \Lpkkdjdj + i \Gpkkdjdj. 
\end{equation}
Observe that we have the decomposition 
\begin{align}
\frac{1}{\omega(-i \nabla) \pm (E_j - E_{j'}) \pm i 0^\pm}=&\pv\prc{\frac{1}{\omega(-i \nabla) \pm (E_{j} - E_{j'})}}\nonumber\\
&\mp i \pi \delta\prc{\omega(-i \nabla) \pm (E_{j} - E_{j'})}.
\end{align}
where
\begin{align}
\pv\prc{\frac{1}{\omega(-i \nabla) \pm (E_{j} - E_{j'})}} =&  \frac{\omega(-i \nabla) \pm (E_{j} - E_{j'})}{\abs{\omega(-i \nabla) \pm (E_{j} - E_{j'})}^2}.
\end{align}
Combining the Lamb shift contributions,
\begin{align}
    \Lambda^{LS}_{k,k';j,j'} :=& \Lambda_{k,k';j,j'} ^{-} + \Lambda_{k,k';j,j'} ^{+} = \Lm + \Lmkkdjdj.
\end{align}
Thus,
\begin{equation}
\Lambda^{LS}_{k,k';j,j'} = 2 \inp{w * (\cc{\rchi_{k'}} \rchi_k), \prc{\frac{\omg(-i \nabla) }{{\omega(-i \nabla)^2 - (E_{j} - E_{j'})^2}}} w * (\cc{\rchi_{j'}} \rchi_{j})}.
\end{equation}
The Fermi Golden Rule contributions are given by
\begin{align}
\Gamma_{k,k';j,j'} =& \Gm + \Gp = \Gm - \Gmkkdjdj\nonumber\\
=& \inp{\rchi_k, w * \prc{\delta(\omg(-i \nabla) - (E_{j} - E_{j'})) w * (\cc{\rchi_{j'}} \rchi_j) \rchi_{k'}}}\label{eq:Gamma-kkdjjd-2} \\
& - \inp{\rchi_k, w * \prc{\delta(\omg(-i \nabla) + (E_{j} - E_{j'})) w * (\cc{\rchi_{j'}} \rchi_j) \rchi_{k'}}}.\label{eq:Gamma-kkdjjd-3}
\end{align}
Thus, 
\begin{equation}
    \label{eq:definition-of-Gamma-kkdjjd}
\Gamma_{k,k';j,j'} = \FGRkj(\1_{\brs{j > j'}} - \1_{\brs{j < j'}}),
\end{equation}
where 
\begin{align}
    \FGRkj :=& \inp{\rchi_k, w * \prc{\delta(\omg(-i \nabla) - |E_{j} - E_{j'}|)} w * (\cc{\rchi_{j'}} \rchi_{j} ) \rchi_{k'}}\nonumber\\*
    =& \inp{w * (\rchi_k \cc{\rchi_{k'}}), \prc{\delta(\omg(-i \nabla) - |E_{j} - E_{j'}|)} w * (\cc{\rchi_{j'}} \rchi_{j} )}.
\end{align}
Here, we note that due to $ \omg(- i \nabla) \geq 0$, the term in \eqref{eq:Gamma-kkdjjd-2} on the RHS vanishes for $E_{j'} - E_j > 0$, and the term in \eqref{eq:Gamma-kkdjjd-3} vanishes for $E_{j'} - E_j > 0$. Recalling that the eigenvectors $\rchi_k$ are real-valued functions, we note the following symmetry properties.

\begin{remark}(Symmetry properties of the coefficients)
    \label{rmk:symmetry}
    If $v, w$ are even real-valued functions then the coefficients $\Lhar$,  $ \Lambda_{k,k'}^{LS}$ and $\Gamma_{k,k'}^{FGR}$ are symmetric under both $k \leftrightarrow k'$ and $j \leftrightarrow j'$. Moreover, 
    $\Gamma_{k,k'}^{FGR}$ are non-negative and 
\begin{align}
     \Gamma_{k,k';j,j'} &= \Gamma_{k',k;j,j'} =   -\Gamma_{k,k';j',j}
\end{align}
is antisymmetric under $j \leftrightarrow j'$. That is, 

\begin{align}
    \Lambda_{k,k';j,j'}^{LS}&= \Lambda_{k',k:j,j'}^{LS} = \Lambda_{k',k;j',j}^{LS}\\
    \Gamma_{k,k';j,j'}^{FGR} &= \Gamma_{k',k;j,j'}^{FGR} =  \Gamma_{k,k';j',j}^{FGR}\\
    \Lambda^{Har}_{k,k';j,j'}&= \Lambda^{Har}_{k,k';j',j} = \Lambda^{Har}_{k',k;j,j'}.
\end{align}
Moreover, we note that when $j=j'$, so that $|E_j - E_{j'}|=0$, the terms \eqref{eq:Gamma-kkdjjd-2} and \eqref{eq:Gamma-kkdjjd-3} cancel each other, so that $\Gamma_{k,k';j,j} = 0$.
\end{remark}

\bigskip

We match the strength of the Hartree interaction with the coupling to the radiation field by setting 
\begin{equation}
    \lam = \eta^2.
\end{equation}
Therefore, we obtain 
\begin{align}
     \pt_t A_k(t) =& - \eta^2 \sum_{k', j,j'} \spr{i (\Lhar- \Lambda_{k,k';j,j'}^{LS}) + \Gamma_{k,k';j,j'}}  e^{i t [(E_k - E_{k'} )-(E_j -E_{j'} )]}\\
    &\times  A_j(t)  \cc{A_{j'}(t)} A_{k'}(t) -i \eta^2 \Remkaa(t) -i \eta^2 \Remkab(t) - i \Remka(t).
\end{align}
Now, we define the rescaled time variable $T=\eta^2 t$ and the rescaled amplitude function 
\begin{align}
    F_\mu^\eta(T) :=& A_\mu(T/\eta^2),\quad \mu \in\set{k,k',j,j'},  \label{def:rescaled-time}\\
   M_{k,k';j,j'} :=&  - i \prc{ \Lhar- \Lambda_{k,k';j,j'}^{LS}}  - \Gamma_{k,k';j,j'},\label{def:effective-coefficient-Mkkdjjd}\\
     \Delta E_{k,k';j,j'} :=&  (E_k - E_{k'}) - (E_{j}- E_{j'}), \label{def:resonance-energy-gap}
\end{align}
so that 

\begin{align}
    \label{eq:ODE-for-Fk(T)}
    \pt_T F_k^\eta(T) =& \sum_{k', j,j'} M_{k,k';j,j'} e^{i {T \Delta E_{k,k';j,j'}}/{\eta^2}}  F_j^\eta(T)  \cc{F_{j'}^\eta(T)}   F_{k'}^\eta(T) \\
    &-i \Remkaa(T/\eta^2) -i \Remkab(T/\eta^2) - \frac{i}{\eta^2} \Remka(T/\eta^2).
\end{align}
Sometimes, we will write $\Delta E$ to refer to $ \Delta E_{k,k';j,j'}$ for brevity.

\section{Effective resonance cascade equation and proof of BEC}
\label{effective-resonance-cascade-eq-and-proof-of-BEC}
In this section, we study the effective resonance cascade equation on the time scale $ t = T / \eta^2$. Recall the definition of the resonance function $\Delta E$ in \eqref{def:resonance-energy-gap}, and the effective coefficient $M_{k,k';j,j'}$ in \eqref{def:effective-coefficient-Mkkdjjd}. We will show that the dominant contribution to \eqref{eq:ODE-for-Fk(T)} is given by the coefficients associated with $\Delta E =0$, 
\begin{align} \label{eq:eff-resonant}
    \pt_T F_k(T) =& \sum_{k', j,j'; \Delta E_{k,k';j,j'} = 0}  M_{k,k';j,j'} \cc{F_{j'}(T)}   F_j(T) F_{k'}(T).
\end{align}
Given the assumption of the rational independence of energy gaps  \eqref{eq:resonance-condition-for-energy-differences}, $\Delta E=0$ enforces the direct interaction
\begin{equation}
    k=k', \qquad j=j'
\end{equation}
and the exchange interaction
\begin{equation}
    k =j, \qquad k' = j'.
\end{equation}
We immediately observe that the effective resonance cascade equation reduces to an infinite system of coupled nonlinear ordinary differential equations (ODEs) in diagonal form. That is,  \eqref{eq:eff-resonant} reduces to
\begin{equation}
\partial_T F_k(T)=\sum_{k'\ge0} M_{k,k'}\,|F_{k'}(T)|^2\,F_k(T),
\qquad 
\end{equation}
where $M_{k,k'}:=M_{k,k';k,k'}+M_{k,k;k',k'}-\delta_{k,k'}M_{k,k;k,k}$.

\begin{theorem}
\label{thm:complete-BEC}
Consider the diagonal nonlinear cascade equation
\begin{equation}\label{eq:FkT-ODE-def-1-0-0}
    \partial_T F_k(T)=\sum_{k'\ge 0} M_{k,k'}\,|F_{k'}(T)|^2\,F_k(T)
    \;,\;\;\; F_k(0)=\langle\chi_k,\phi_0\rangle
\end{equation}
with $\sum_k (1+E_k)|F_k(0)|^2<C$ and $\|F(0)\|_{\ell^2}=1$.
The coefficients $M_{k,k'}= M_{k,k';k,k'}+M_{k,k;k',k'}-\delta_{k,k'}M_{k,k;k,k}$, defined in \eqref{def:effective-coefficient-Mkkdjjd}, have the form 
\begin{eqnarray}
    \label{eq:M-form-1-1}
    M_{k,k'}&=& - i \Big(\Lambda_{k,k';k,k'}^{Har}
    +\Lambda_{k,k;k',k'}^{Har}
    -\delta_{k,k'}\Lambda_{k,k;k,k}^{Har}
    \nonumber\\
    &&\qquad- \Lambda_{k,k';k,k'}^{LS}
    - \Lambda_{k,k;k',k'}^{LS}
    + \delta_{k,k'}\Lambda_{k,k';k,k'}^{LS}
    \Big)
    \nonumber\\
    &&
    - \FGRk \prc{\1_{k > k'} - \1_{k' > k}},
\end{eqnarray}
with 
\begin{equation}
    \FGRk :=  \inp{w * (\rchi_k \cc{\rchi_{k'}}), \prc{\delta(\omg(-i \nabla) - |E_{k} - E_{k'}|)} w * (\cc{\rchi_{k'}} \rchi_{k} )} > 0
\end{equation}
 symmetric in $(k,k')$ and $k\neq k'$. Assume that the initial ground state amplitude is strictly nonzero, i.e., $|F_0(0)| > 0$. Furthermore, as stated in \lemref{lm-FGR-generic-1-1}, recall that the transition rates to the ground state satisfy $\Gamma_{0,k'}^{FGR} > 0$ for all excited states $k' \ge 1$, though they may decay at an arbitrary rate as $k' \to \infty$. Then, 
\begin{enumerate}
    \item The infinite cascade is globally well-posed, and the $\ell^2$-mass is conserved, $\|F(T)\|_{\ell^2}=\|F(0)\|_{\ell^2}$ for all $T \ge 0$, so that
\begin{equation}
    \|F(T)\|_{\ell^\infty}\le \|F(T)\|_{\ell^2}=\|F(0)\|_{\ell^2}.
\end{equation}
    \item The solution to the effective resonance cascade equation exhibits complete macroscopic ground state occupation as $T \to \infty$, that is, the probability mass flows strictly toward the ground state. More precisely, assume $\sum_{k\geq0}|F_k(0)|^2=1$ with $0<|F_0(0)|\leq1$ and $\sum_{k\geq0}(E_k-E_0)|F_k(0)|^2<C$. Then, the occupation densities of all excited states decay to zero, and
\begin{equation}
    \lim_{T \to \infty} \sum_{k' =1}^\infty \abs{F_{k'}(T)}^2 = 0.
\end{equation}
Consequently, the ground state absorbs the entire mass of the system,
\begin{equation}
    \label{eq:complete-BEC-statement}
    \lim_{T \to \infty} \abs{F_0(T)}^2 = \sum_{k=0}^\infty \abs{F_k(0)}^2 = 1,
\end{equation}
which constitutes the complete dynamical formation of a Bose-Einstein Condensate (BEC) in the boson subsystem. Furthermore, $\pt_T \sum_k E_k \abs{F_k(T)}^2 \leq 0$, and
\begin{eqnarray}
    \lim_{T\rightarrow\infty} \sum_{k\geq0}E_k|F_k(T)|^2
    = E_0 \|F(0)\|_{\ell^2}^2\,.
\end{eqnarray}
That is, the total energy of the boson subsystem decays monotonically with $T$ and concentrates in the ground state.
\item Assume $0<|F_0(0)|\leq1$ and that only finitely many modes are excited at $T=0$. Let $K>0$ be the smallest index such that $F_k(0) = 0$ for all $k > K$, and let $\tilde{\Gamma}_K := \min_{1\leq k' \leq K} \Gamma_{0,k'}^{FGR}$. Then,  
\begin{equation}
    \abs{F_0(T)}^2 \geq \frac{1}{1 + \frac{1 - \abs{F_0(0)}^2}{\abs{F_0(0)}^2} e^{- 2 \tilde{\Gamma}_K T}}
\end{equation} 
with $\norm{F(0)}_{\ell^2} = 1$.
\end{enumerate}

\end{theorem}

\begin{proof}
We first prove global existence and uniqueness for the system of occupation probabilities $\{p_k(T):=|F_k(T)|^2\}_{k\geq0}$ in $\ell^1$. 
They satisfy the infinite system of ODEs
\begin{align}\label{eq-Tder-Fsq-1-0}
    \pt_T p_k(T) &= \sum_{k'} 2 \Re{ M_{k,k'} p_{k'}(T) } p_k(T)  \nonumber \\
    &= \sum_{k'} 2 \Gamma_{k,k'}^{FGR}  \big(\1_{\brs{k' > k}} - \1_{\brs{k' < k}}\big) p_{k'}(T) p_k(T).
\end{align} 
Local well-posedness in $\ell^1$ follows from a standard fixed-point argument using the Duhamel formula and the fact that $\sup_{k,k'}|\Gamma_{k,k'}^{FGR}|<\Gamma^*$ is uniformly bounded, \lemref{lem:uniform-FGR-bound}. Global well-posedness follows from conservation of $\ell^1$ mass, which we verify next.

Summing over all modes $k \ge 0$, the derivative of the total mass evaluates to:
\begin{align}
    \pt_T \sum_{k=0}^\infty p_k(T) &= \sum_{k, k'} 2 \Gamma_{k,k'}^{FGR}  \big(\1_{\brs{k' > k}} - \1_{\brs{k' < k}}\big) p_{k'}(T) p_k(T) = 0,
\end{align}
which vanishes identically due to the strict antisymmetry of the directional indicator $(\1_{\brs{k' > k}} - \1_{\brs{k' < k}})$ and the underlying symmetry of the transition rates $\Gamma_{k,k'}^{FGR} = \Gamma_{k',k}^{FGR}$. Therefore, the total probability mass is strictly conserved for all time, yielding $\sum_{k=0}^\infty p_k(T) = \sum_{k=0}^\infty p_k(0) = 1$.
\medskip
     
We next prove that the energy remains uniformly bounded. To this end, we consider the weighted sum over the spectrum. By exchanging the summation indices for $k < k'$, we obtain 
\begin{align}
    \label{eq:mass-flow-towards-low-energy-modes}
    \pt_T \sum_{k=0}^\infty E_k \; p_k(T) &= \sum_{k,k'} 2 \Gamma_{k,k'}^{FGR}  \big(\1_{\brs{k' > k}} - \1_{\brs{k' < k}}\big) E_k p_{k'}(T) p_k(T) \nonumber \\
    &= - \sum_{k < k'} 2 \Gamma_{k,k'}^{FGR}  (E_{k'} - E_k) \; 
    p_{k'}(T) p_k(T) \le 0
\end{align}
(to be precise, one should use regularized, non-negative, monotone increasing weights $e_k^{(R)}:=\min\{E_k-E_0,R\}$ in place of $E_k$, and use monotone convergence as $R\rightarrow\infty$).
Since the physical energy levels are strictly ordered ($E_{k'} > E_k$ for $k' > k$), the derivative is non-positive, implying that the probability mass flows strictly toward lower-energy modes and that the overall energy of the boson subsystem decreases monotonically. In particular, 
$\sum_k(1+E_k)p_k(T)\leq \sum_k(1+E_k)p_k(0)<C$ uniformly in $T$. 

To prove local well-posedness of \eqref{eq:FkT-ODE-def-1-0-0}, we define 
\begin{eqnarray}
    a_k(T):=\sum_{k'}M_{k,k'}p_{k'}(T)
\end{eqnarray}
which is uniformly bounded in $\ell^\infty$ due to \lemref{lem:bound-Fk2-Mkkd}. Then,
\begin{eqnarray}
    F_k(T):=\exp\Big(\int_0^TdS\; a_k(S)\Big)F_k(0)
\end{eqnarray} 
solves \eqref{eq:FkT-ODE-def-1-0-0} with the given initial data. Because $p_k(T)$ solves
\begin{eqnarray}
    \partial_T p_k(T)=2\Re{a_k(T) p_k(T)} \,, \;\;\;p_k(0)=|F_k(0)|^2\,,
\end{eqnarray}
so that
\begin{eqnarray}
    p_k(T):=\exp\Big(\int_0^TdS\;2\Re{ a_k(S)}\Big)p_k(0) \geq0 \,,
\end{eqnarray}
uniqueness of scalar ODEs implies that $p_k(T)=|F_k(T)|^2$.
Consequently, $F_k(T)$ solves \eqref{eq:FkT-ODE-def-1-0-0}.
To prove uniqueness, assume two solutions with the same initial data; then, both determine the same $a_k(T)$ and satisfy the same linear scalar ODE; uniqueness thus follows. Using the conservation of $\|F(T)\|_{\ell^2}=\|p(T)\|_{\ell^1}$, we may subsequently enhance local to global well-posedness.

With mass conservation and energy monotonicity established, we isolate the ground state ($k=0$) from the dynamic system. From \eqref{eq-Tder-Fsq-1-0}, we infer that its occupation density satisfies the differential equation
\begin{equation}
    \label{eq:ODE-for-ground-state-occupation}
    \pt_T \abs{F_0(T)}^2 = 2 \abs{F_0(T)}^2 \sum_{k' =1}^\infty \Gamma_{0,k'}^{FGR} \abs{F_{k'}(T)}^2.
\end{equation}
Integrating this directly yields an exponential growth factor:
\begin{equation}
    \abs{F_0(T)}^2 = \abs{F_0(0)}^2 \exp\left( 2 \int_0^T  \sum_{k' =1}^\infty \Gamma_{0,k'}^{FGR} \abs{F_{k'}(S)}^2 \, dS \right).
\end{equation}
Because the total mass is conserved ($\abs{F_0(T)}^2 \le 1$) and the initial ground state amplitude is strictly nonzero, the integral in the exponent must be globally bounded as $T \to \infty$:
\begin{equation}
    \int_0^\infty \sum_{k' =1}^\infty \Gamma_{0,k'}^{FGR} \abs{F_{k'}(S)}^2 \, dS < \infty.
\end{equation}
Crucially, the transition coefficients $\Gamma^{FGR}$ are uniformly bounded via the version of the Limiting Absorption Principle established in \lemref{lem:uniform-lap-for-R-eps-pm}, and the state vector is uniformly bounded in $\ell^2(\C)$ by mass conservation. Moreover, the time derivative of the integrand is uniformly bounded,
\begin{eqnarray}
    \Big|\partial_S \sum_{k' =1}^\infty \Gamma_{0,k'}^{FGR} \abs{F_{k'}(S)}^2  \Big|
    &\leq&\Gamma^*\sum_{k\geq1}|\partial_S \abs{F_{k'}(S)}^2|
    \nonumber\\
    &\leq & 2(\Gamma^*)^2 \sum_{k,k'} \abs{F_{k}(S)}^2\abs{F_{k'}(S)}^2
    \nonumber\\
    &=& 2(\Gamma^*)^2
\end{eqnarray}
where $\Gamma_{k,k'}^{FGR}\leq \Gamma^*$ for all $k,k'$ (see \lemref{lem:uniform-FGR-bound} below), and where we used \eqref{eq-Tder-Fsq-1-0} in passing to the second line.  
Thus, the integrand is uniformly continuous in time. By Barbalat's Lemma, the convergence of the infinite integral of a uniformly continuous, non-negative function guarantees that the integrand strictly vanishes
\begin{equation}
    \label{eq:integrand-in-exponent-is-integrable}
    \lim_{T \to \infty} \sum_{k' =1}^\infty \Gamma_{0,k'}^{FGR} \abs{F_{k'}(T)}^2 = 0.
\end{equation}
Since $\Gamma_{0,k'}^{FGR} > 0$ for all $k' \ge 1$, this immediately implies pointwise decay for every individual excited mode, namely $\lim_{T \to \infty} \abs{F_{k'}(T)}^2 = 0$.

To upgrade this pointwise decay to strong convergence in the total $L^2$-mass, we must establish uniform control over the high-energy tails. We compute the mass flow for modes above an arbitrary high-energy threshold $K \ge 1$,
\begin{equation}
    \pt_T \sum_{k > K} \abs{F_k(T)}^2 = \sum_{k > K} \sum_{k' \ge 0} 2 \Gamma_{k,k'}^{FGR} \big(\1_{\brs{k' > k}} - \1_{\brs{k' < k}}\big) \abs{F_{k'}(T)}^2 \abs{F_k(T)}^2.
\end{equation}
We split the inner sum over $k'$ into two regions: $k' > K$ and $k' \le K$. For the region where both $k, k' > K$, the sum identically vanishes due to the exact antisymmetry of the indicator term and the symmetry of $\Gamma$. We are left strictly with the cross-terms where $k' \le K$. Since $k > K$, this implies $k' < k$, causing the indicator function to evaluate to $-1$,
\begin{equation}
    \label{eq:monotonicity-of-tail-mass}
    \pt_T \sum_{k > K} \abs{F_k(T)}^2 = - \sum_{k > K} \sum_{k' \le K} 2 \Gamma_{k,k'}^{FGR} \abs{F_{k'}(T)}^2 \abs{F_k(T)}^2 \le 0.
\end{equation}
This monotonicity implies that the tail mass is uniformly bounded by its initial value
\begin{equation}
\sum_{k > K} \abs{F_k(T)}^2 \le \sum_{k > K} \abs{F_k(0)}^2.
\end{equation}
Because the initial state resides in $\ell^2(\C)$, for any arbitrary $\epsilon  > 0$, there exists a sufficiently large index $K(\epsilon )$ such that the initial mass tail is strictly bounded by $\epsilon $. We decompose the sum over all excited states into a finite low-energy block and the infinite high-energy tail,
\begin{equation}
    \sum_{k' \ge 1} \abs{F_{k'}(T)}^2 \le \sum_{k'=1}^{K(\epsilon )} \abs{F_{k'}(T)}^2 + \sum_{k' > K(\epsilon )} \abs{F_{k'}(0)}^2 \le \sum_{k'=1}^{K(\epsilon )} \abs{F_{k'}(T)}^2 + \epsilon .
\end{equation}
Taking the limit as $T \to \infty$, the finite sum of strictly vanishing terms evaluates to zero while the tail is bounded by $\epsilon $,
\begin{equation}
    \limsup_{T \to \infty} \sum_{k' \ge 1} \abs{F_{k'}(T)}^2 \le \sum_{k'=1}^{K(\epsilon )} \left( \lim_{T \to \infty} \abs{F_{k'}(T)}^2 \right) + \epsilon  = \epsilon .
\end{equation}
Since this inequality holds for any arbitrary $\epsilon  > 0$, we conclude the strong limit 
\begin{equation}\label{eq:boson-l2-conv-0-1-0}
    \lim_{T \to \infty} \sum_{k' \ge 1} \abs{F_{k'}(T)}^2 = 0 \,.
\end{equation} 
Invoking the global $L^2$-mass conservation, all residual probability mass must therefore accumulate entirely in the ground state
\begin{equation}
    \lim_{T \to \infty} \abs{F_0(T)}^2 = 1 - \lim_{T \to \infty} \sum_{k' \ge 1} \abs{F_{k'}(T)}^2 = 1,
\end{equation}
as claimed. 

To prove that in the limit $T\rightarrow\infty$, the expected energy converges to the ground state energy, we first note that $E_k-E_0>0$ for all $k\geq1$, and that for any $K>0$,
\begin{eqnarray}
    \label{eq:mass-flow-1-1}
    \lefteqn{
    \pt_T \sum_{k>K} (E_k-E_0) \abs{F_k(T)}^2 
    }
    \nonumber\\
    &=& \sum_{k,k'>K} 2 \Gamma_{k,k'}^{FGR}  \big(\1_{\brs{k' > k}} - \1_{\brs{k' < k}}\big) (E_k-E_0) \abs{F_{k'}(T)}^2 \abs{F_k(T)}^2 
    \nonumber\\
    &&\quad+ \sum_{k>K}\sum_{k'\leq K}2\Gamma_{k,k'}^{FGR}  \big(\1_{\brs{k' > k}} - \1_{\brs{k' < k}}\big) (E_k-E_0) \abs{F_{k'}(T)}^2 \abs{F_k(T)}^2 
    \nonumber \\
    &=& - \sum_{K<k < k'} 2 \Gamma_{k,k'}^{FGR}  (E_{k'} - E_k) \abs{F_{k'}(T)}^2 \abs{F_k(T)}^2 
    \nonumber\\
    &&\quad-\sum_{k>K}\sum_{k'\leq K}2\Gamma_{k,k'}^{FGR} (E_k-E_0) \abs{F_{k'}(T)}^2 \abs{F_k(T)}^2 
    \nonumber\\
    &\leq& 0.
\end{eqnarray}
Here we used the monotonicity of the eigenenergies ($E_{k'} > E_k$ for $k' > k$). Therefore, $\sum_{k>K} (E_k-E_0) \abs{F_k(T)}^2\leq \sum_{k>K} (E_k-E_0) \abs{F_k(0)}^2$ uniformly in $T$. 
Thus, for any arbitrary fixed $\epsilon>0$, there exists $K=K(\epsilon)$ such that $\sum_{k>K} (E_k-E_0) \abs{F_k(T)}^2<\epsilon$, uniformly in $T$. 
This implies that 
\begin{eqnarray} 
    \limsup_{T\rightarrow\infty} \sum_{k\geq1} (E_k-E_0)|F_k(T)|^2 
    &\leq&\lim_{T\rightarrow\infty}\sum_{1\leq k\leq K} (E_k-E_0)|F_k(T)|^2
    +\epsilon
    \nonumber\\
    &\leq&(E_K-E_0)\lim_{T\rightarrow\infty}\sum_{1\leq k\leq K} |F_k(T)|^2
    +\epsilon
    \nonumber\\
    &\leq&0+\epsilon
\end{eqnarray}
where we used \eqref{eq:boson-l2-conv-0-1-0}.
Because $\epsilon>0$ is arbitrary, this implies $\lim_{T\rightarrow\infty} \sum_{k\geq1} (E_k-E_0)|F_k(T)|^2 =0$. We assume $F_0(0)>0$ so that \eqref{eq:boson-l2-conv-0-1-0} holds.
From mass conservation, $\sum_{k\geq0}|F_k(T)|^2=1$, and energy monotonicity, $\sum_{k\geq0}E_k|F_k(T)|^2<\sum_{k\geq0}E_k|F_k(0)|^2<C_{E}$, therefore
\begin{eqnarray}
    \lim_{T\rightarrow\infty}\sum_{k\geq0} E_k|F_k(T)|^2
    &=&\lim_{T\rightarrow\infty}\sum_{k\geq0} E_0|F_k(T)|^2+
    \lim_{T\rightarrow\infty}
    \sum_{k\geq1} (E_k-E_0)|F_k(T)|^2  
    \nonumber\\
    &=& E_0 \lim_{T\rightarrow\infty}\sum_{k\geq0} |F_k(T)|^2
    +0
    \nonumber\\
    &=& E_0
\end{eqnarray}
follows.  

Assuming now that only finitely many modes are nonzero at initial time, let $K$ be the smallest integer such that $|F_k(0)|=0$ for all $k>K$.
Since $ \sum_{k > K} |F_k(T)|^2 \leq \sum_{k > K} |F_k(0)|^2=0$  by \eqref{eq:monotonicity-of-tail-mass}, we find from \eqref{eq:ODE-for-ground-state-occupation} that
\begin{align}
    \pt_T |F_0(T)|^2 =& 2 |F_0(T)|^2 \sum_{k' =1}^K \Gamma_{0,k'}^{FGR} |F_{k'}(T)|^2\\
\geq& 2 \tilde{\Gamma}_K |F_0(T)|^2 \sum_{k' =1}^K |F_{k'}(T)|^2\\
=& 2 \tilde{\Gamma}_K |F_0(T)|^2 (1 - |F_0(T)|^2) 
\end{align}
where we used the mass conservation $|F_0(T)|^2 + \sum_{k'=1}^K |F_{k'}(T)|^2 = \|F(T)\|_{\ell^2}^2 = 1$, and where $\tilde{\Gamma}_K := \min_{k' \leq K} \Gamma_{0,k'}^{FGR}$. To integrate this ordinary differential inequality, we define $G(T):= |F_0(T)|^2$. Separating variables and integrating over the interval $[0, T]$ yields 
\begin{equation}
\frac{G(T)}{1-G(T)} \ge \left(\frac{G(0)}{1-G(0)}\right) e^{2 \tilde{\Gamma}_K T}.
\end{equation}
Solving this algebraic inequality for $G(T)$,
\begin{equation}
    |F_0(T)|^2 \ge \frac{1}{ 1 + \frac{1 - |F_0(0)|^2}{|F_0(0)|^2} e^{- 2 \tilde{\Gamma}_K T} }.
\end{equation}
This proves the theorem. 
\end{proof}

\begin{lemma}[Uniform bound on the Fermi Golden Rule coefficients]
\label{lem:uniform-FGR-bound}
Assume that $w\in W^{1,1}(\mathbb{R}^3)$. 
Then the Fermi Golden Rule coefficients satisfy
\begin{equation}
    \Gamma^*
    :=
    \sup_{k\neq k'}
    \Gamma^{\mathrm{FGR}}_{k,k'}
    \leq 
    C\|w\|_{W^{1,1}(\mathbb{R}^3)}^2,
\end{equation}
where $C>0$ is independent of $k$ and $k'$.
\end{lemma}

\begin{proof}
Fix $k\neq k'$ and set
\begin{equation}
    R_{k,k'}:=|E_k-E_{k'}|>0.
\end{equation}
By the definition of the Fermi Golden Rule coefficient and spherical
coordinates,
\begin{align}
    \Gamma^{\mathrm{FGR}}_{k,k'}
    &=
    \int_{\mathbb{R}^3}
    |\widehat{w}(\xi)|^2
    \left|
        \widehat{\rchi_k\overline{\rchi_{k'}}}(\xi)
    \right|^2
    \delta\bigl(|\xi|-R_{k,k'}\bigr)\,d\xi
    \notag\\
    &=
    R_{k,k'}^2
    \int_{\mathbb{S}^2}
    |\widehat{w}(R_{k,k'}\omega)|^2
    \left|
        \widehat{\rchi_k\overline{\rchi_{k'}}}
        (R_{k,k'}\omega)
    \right|^2
    d\sigma(\omega).
\end{align}
Since the eigenfunctions are $L^2$-normalized,  
Cauchy-Schwarz  gives
\begin{equation}
    \left\|
        \widehat{\rchi_k\overline{\rchi_{k'}}}
    \right\|_{L^\infty}
    \leq
    C\|\rchi_k\overline{\rchi_{k'}}\|_{L^1}
    \leq
    C\|\rchi_k\|_{L^2}\|\rchi_{k'}\|_{L^2}
    =C.
\end{equation}
Moreover, since $w\in W^{1,1}(\mathbb{R}^3)$ and $\xi_j\widehat{w}(\xi) =-i\,\widehat{\partial_j w}(\xi)$, we find
\begin{equation}
    |\xi|\,|\widehat{w}(\xi)|
    \leq
    C\sum_{j=1}^3
    \|\widehat{\partial_j w}\|_{L^\infty}
    \leq
    C\|\nabla w\|_{L^1}.
\end{equation}
Evaluating this estimate on the resonant sphere
$|\xi|=R_{k,k'}$ yields
\begin{equation}
    R_{k,k'}^2
    |\widehat{w}(R_{k,k'}\omega)|^2
    \leq
    C\|\nabla w\|_{L^1}^2,
    \qquad \omega\in\mathbb{S}^2.
\end{equation}
Therefore,
\begin{align}
    \Gamma^{\mathrm{FGR}}_{k,k'}
    &\leq
    C\|\nabla w\|_{L^1}^2
    \int_{\mathbb{S}^2}
    \left|
        \widehat{\rchi_k\overline{\rchi_{k'}}}
        (R_{k,k'}\omega)
    \right|^2
    d\sigma(\omega)
    \nonumber\\
    &\leq
    C\|\nabla w\|_{L^1}^2.
\end{align}
The constant is independent of $k$ and $k'$. Taking the supremum over
all $k\neq k'$ proves
\begin{equation}
    \Gamma^*
    \leq
    C\|\nabla w\|_{L^1}^2
    \leq
    C\|w\|_{W^{1,1}}^2.
\end{equation}
This establishes the claim.
\end{proof}

\begin{lemma}[Uniform bound on effective coefficients]
\label{lem:bound-Fk2-Mkkd}
Assume the potential $V(x)$ satisfies the bound $\jpp{x}^{2s} \leq C( V(x)+C_0+1)$ on $\mathbb{R}^3$ for some $ \frac12<s<\frac{1+\beta}{2}$ (see \hypref{asm:confining-potential} in \assumref{the-assumptions}) and $C_V > 0$. Let the interaction kernel $w$ satisfy the weighted integrability condition $w \in L^2(\mathbb{R}^3, \jpp{x}^{2s} dx)$. 
Then, for any state $\vfi_t = \sum_k F_k^\eta(\eta^2 t) e^{-itE_k} \rchi_k$ with uniformly bounded particle energy $ \inp{\vfi_t, (-\Delta + V) \vfi_t} < C_E$, the effective coefficients $M_{k,k'}$ satisfy the uniform bound
\begin{equation}
    \sup_{T,\eta}\sup_k\sum_{k'} \abs{M_{k,k'}} \abs{F_{k'}^\eta}^2 \leq C_{E,s},
\end{equation}
where the constant $C'_{E,s} > 0$ is bounded for $s >1/2$ and depends on the value of the conserved total boson and photon energy, and on norms of the interaction potentials $v$ and $w$.
\end{lemma}

\begin{proof}
    Consider the coefficients $M_{k,k'}$ defined in \eqref{def:effective-coefficient-Mkkdjjd}
    \begin{equation}
           M_{k,k'} = -i \Lambda_{k,k'}^{Har} + i M_{k,k'}^{L,F},
\end{equation}
where  $\Lambda_{k,k'}^{Har}$ is the Hartree contribution defined in \eqref{eq:definition-of-Lhar-coefficients} with the exchange contributions $k=j,\, k'=j'$ and direct contributions $k=k'$, $j=j'$, respectively, 
\begin{equation}
    \Lambda_{k,k'}^{Har} = \inp{\rchi_k \cc{\rchi_{k'}}, v * (\rchi_k \cc{\rchi_{k'}}) }
    +
    \inp{\rchi_k \cc{\rchi_{k}}, v * (\rchi_{k'} \cc{\rchi_{k'}}) }
    -
    \delta_{k,k'}\inp{\rchi_k \cc{\rchi_{k}}, v * (\rchi_{k} \cc{\rchi_{k}}) },
\end{equation}
and 
    \begin{equation}
        M_{k,k'}^{L,F} :=   \Lambda_{k,k'}^{LS}  +i  \Gamma_{k,k'}^{-}.
    \end{equation}
The terms $\Lambda_{k,k'}^{LS},\,\Gamma_{k,k'}^{-} \in \R$ are the Lamb shift and Fermi Golden Rule contributions, respectively, given by \eqref{eq:M-form-1-1}.
\begin{eqnarray}\label{eq:Lamb-iGam-id-1-0}
    \Lambda_{k,k'}^{LS}  + i  \Gamma_{k,k'}^{-} 
    &=& \lim_{\veps \to 0}
    \inp{w * (\rchi_k \cc{\rchi_{k'}}) , \cRep^{-} \, w * (\cc{\rchi_{k'}} \rchi_k) }
    \nonumber\\
    &&\qquad+ \lim_{\veps \to 0} \inp{w * (\rchi_k \cc{\rchi_{k'}}) , \cRep^{+} \, w * (\cc{\rchi_{k'}} \rchi_k) }
    \nonumber\\
    &&\qquad\qquad
    +2\inp{w * (\rchi_k \cc{\rchi_{k}}) , |\nabla|^{-1}\, w * (\cc{\rchi_{k'}} \rchi_{k'}) }
    \nonumber\\
    &&\qquad\qquad\qquad
    -2\delta_{k,k'}\inp{w * (\rchi_k \cc{\rchi_{k}}) , |\nabla|^{-1}\, w * (\cc{\rchi_{k}} \rchi_{k}) }
\end{eqnarray}
where $\cRep^{\pm}$ are the resolvent operators defined by
\begin{equation}
    \cRep^{\pm} := \prc{\frac{1}{\omg(-i \nabla) \pm (E_{k} - E_{k'}) \pm i \veps}}.
\end{equation}
The term on the second line of \eqref{eq:Lamb-iGam-id-1-0} accounts for direct contributions and only contributes to $\Lambda_{k,k'}^{LS}$ but not to $\Gamma_{k,k'}^{-}$, while the term on the first line accounts for exchange contributions and contributes to both.
Since the contribution of the Hartree term is uniformly bounded, it gives a constant linear shift to the calculations. Thus, we focus on the contributions of the Lamb shift and the Fermi Golden Rule, which can be estimated together because they have the same structure. 
We write
\begin{eqnarray}
     M_{k,k'}^{L,F} &=:& \widetilde M_{k,k'}^{L,F} + 
    2\inp{w * (\rchi_k \cc{\rchi_{k}}) , |\nabla|^{-1}\, w * (\cc{\rchi_{k'}} \rchi_{k'}) }
    \nonumber\\
    &&\qquad\qquad\qquad
    -2\delta_{k,k'}\inp{w * (\rchi_k \cc{\rchi_{k}}) , |\nabla|^{-1}\, w * (\cc{\rchi_{k}} \rchi_{k}) }
\end{eqnarray}
Then we estimate the contributions from both resolvents $\cR_\veps^{\pm}$ in the same way. We thereby obtain
\begin{eqnarray}
    \abs{\widetilde M_{k,k'}^{L,F}} 
    \leq C_{s} \norm{\jpp{\cdot}^s w * (\rchi_k \cc{\rchi_{k'}}) }_{L^2}^2
\end{eqnarray}
from
\begin{align}
           \big|\langle w * &(\rchi_k(x) \cc{\rchi_{k'}(x)}), \cR_\veps^{\pm} w * (\cc{\rchi_{k'}(x) } \rchi_k(x)) \rangle
           \big|
           \\
           =& \abs{ \inp{\jpp{x}^s w * (\rchi_k(x) \cc{\rchi_{k'}(x)}), \jpp{x}^{-s} \cR_\veps^{\pm} w * (\cc{\rchi_{k'}(x) } \rchi_k(x)) }}\\
        \leq& \norm{\jpp{\cdot}^s w * (\rchi_k \cc{\rchi_{k'}})}_{L^2} \norm{\jpp{\cdot}^{-s} \cR_\veps^{\pm} w * (\cc{\rchi_{k'} } \rchi_k) }_{L^2}\\
        \leq& C_{s} \norm{\jpp{\cdot}^s w * (\rchi_k \cc{\rchi_{k'}})}_{L^2}^2,
\end{align}
where we used the limiting absorption principle proved in \lemref{lem:uniform-lap-for-R-eps-pm}
to obtain the last inequality. Define the linear operator 
\begin{equation}
    A \cc{\rchi_{k'}}(x) := \jpp{x}^s \prc{w * \rchi_k \cc{\rchi_{k'}}}(x).
\end{equation}
Then, we have 
\begin{align}
        \sum_{k'\geq0} \abs{ \widetilde M_{k,k'}^{L,F}}|F_{k'}^\eta(T)|^2 & \leq C_{s} 
        \sum_{k'\geq0} \norm{\jpp{\cdot}^s w * (\rchi_k \cc{\rchi_{k'}})}_{L^2}^2 |F_{k'}^\eta(T)|^2
        \nonumber\\
        &\leq C_s \|w\|_{L^2_s}^2
        \sum_{k'\geq0} \|\rchi_k \cc{\rchi_{k'}}\|_{L^1_s}^2
        |F_{k'}^\eta(T)|^2
        \nonumber\\
        &\leq C_s \|w\|_{L^2_s}^2
        \|\rchi_k\|_{L^2}^2 \sum_{k'\geq0} \|\rchi_{k'}\|_{L^1_s}^2|F_{k'}^\eta(T)|^2
        \nonumber\\
        &\leq C_s \|w\|_{L^2_s}^2 \sum_{k'\geq0} 
        (1+E_{k'})|F_{k'}^\eta(T)|^2
        \nonumber\\
        &\leq C_{E,s}\|w\|_{L^2_s}^2 
\end{align}
where the constant only depends on the value of the conserved energy. Here we used Peetre's inequality $ \jpp{y}^{2s} \leq C_s \jpp{y-y'}^{2s} \jpp{y'}^{2s}$, and coercivity of the trapping potential $V$, whereby
\begin{eqnarray}
    \int dy |\chi_k(y)|^2\langle y\rangle^{2s}
    &\leq&
    \int dy |\chi_k(y)|^2 (V(y)+C_0) 
    \nonumber\\
    &\leq&
    \Big\langle \chi_k, (-\Delta+V+C_0) \chi_k\Big\rangle 
    \nonumber\\
    &\leq&C_0(E_k+1) \,,
\end{eqnarray}
see \lemref{lem:moment-bound-for-phi}.
This implies that 
\begin{equation}
    \label{eq:bound-of-rows-of-Mkkd-1-2}
    \sup_{k}\sum_{k'\geq0}  \abs{\widetilde M_{k,k'}^{L,F}} |F_{k'}^\eta|^2\leq C_E.
\end{equation}
holds.

\smallskip

Furthermore, using similar arguments, we have 
\begin{eqnarray}
    \abs{ 
    \inp{w * (\rchi_k \cc{\rchi_{k}}) , |\nabla|^{-1}\, w * (\cc{\rchi_{k'}} \rchi_{k'}) } }
    &\leq& \|  w \|_{\dot H^{-\frac12}}^2 
    \| \cc{\rchi_{k}} \rchi_{k} \|_{L^1} \|   \|\cc{\rchi_{k'}} \rchi_{k'} \|_{L^1}
    \nonumber\\
    &=&  \|  w \|_{\dot H^{-\frac12}}^2 
    \| \rchi_{k} \|_{L^2}^2 \|\rchi_{k'} \|_{L^2}^2 
    \nonumber\\
    &\leq& \|  w \|_{\dot H^{-\frac12}}^2  \,.
\end{eqnarray}
This implies 
\begin{eqnarray}
    \sup_k\sum_{k'\geq0}
    \big|\inp{w * (\rchi_k \cc{\rchi_{k}}) , |\nabla|^{-1}\, w * (\cc{\rchi_{k'}} \rchi_{k'}) }\big| \;  |F_{k'}^\eta|^2
    &\leq&C\|  w \|_{\dot H^{-\frac12}}^2 \sum_{k'\geq0}  |F_{k'}^\eta|^2
    \nonumber\\
    &\leq&C \; \|  w \|_{\dot H^{-\frac12}}^2 \,,
\end{eqnarray}
using that $\| \chi_{k'}\|_{L^2}^2=1$, and $\sum_{k'} |F_k^\eta|^2=1$.

For the Hartree terms, we obtain
\begin{eqnarray}
    \sup_k \sum_{k'\geq0}|\Lambda_{k,k'}^{Har}|\;|F_{k'}^\eta|^2 
    &=& 
    \sup_k \sum_{k'\geq0}\|v\|_{L^\infty} 
    \|\chi_k\|_{L^2}^2\|\chi_{k'}\|_{L^2}^2 |F_{k'}^\eta|^2
    \nonumber\\
    &\leq &\|v\|_{L^\infty} 
\end{eqnarray}
from $\|\chi_k\|_{L^2}=1$ and $\sum_{k'\geq0}  |F_{k'}^\eta|^2=1$.

This proves the claim. 
\end{proof}

\section{Convergence analysis: limit to the effective cascade}
\label{convergence-analysis:limit-to-effective-cascade}
Consider the two systems in $\ell^2$,

\begin{empheq}[left=\empheqlbrace]{align}
    F(T) &\; {\rm solves} \; \pt_T F(T) = M[F] \label{eq:equation-of-F-1-0}\\
F^\eta(T)&=\Big(\langle\chi_k,e^{iT\cH_0/\eta^2}\phi_{T/\eta^2}\rangle\Big)_{k\geq0}\;{\rm\;where\;}\phi_t\;{\rm solves\;}\eqref{eq:boson-subsystem}, \eqref{eq:field-subsystem} 
    \label{eq:equation-of-F-eta}
    \\
    F(0)&=F^\eta(0)
\end{empheq}
where $F= \prc{F_k}_{k \geq 0}$, $F^\eta = \prc{F_k^\eta}_{k \geq 0}$, and  $M[F]$ is the effective cascade operator
\begin{equation}
M[F]_k := \sum_{k'} M_{k,k'} \abs{F_{k'}}^2 F_k,
\end{equation}
with $M_{k,k'}$ denoting the effective coefficient defined in \eqref{def:effective-coefficient-Mkkdjjd}, $ F_k^\eta(T) = A_k(T/\eta^2)$ is the rescaled amplitude function defined in \eqref{def:rescaled-time}.
The respective initial data are equal; that is, 
\begin{equation}
F_k(0) = F_k^\eta(0), \qquad \forall k \geq 0.
\end{equation}
For any initial data $\phi_0 \in L^2(\R^3)$ for the particle subsystem \eqref{eq:boson-subsystem}, \thmref{m:dynamical_bec} shows that the $\ell^2$-mass is conserved for both the effective cascade equation
\begin{equation}
\norm{F(T)}_{\ell^2}^2 = \norm{\phi_0}_{L^2}^2
\end{equation}
and the limit cascade equation
\begin{equation}
\norm{F^\eta(T)}_{\ell^2}^2 = \norm{\phi_0}_{L^2}^2.
\end{equation}
Moreover,
\begin{equation}
 \norm{F^\eta(T)}_{\ell^\infty} \leq \norm{F^\eta}_{\ell^2} =\norm{\phi_0}_{L^2}
\end{equation}
 is uniformly bounded for all $T \in \R$ and $\eta > 0$.

\begin{theorem}
\label{thm:convergence-eta}
The solutions to the systems \eqref{eq:equation-of-F-1-0} and \eqref{eq:equation-of-F-eta} with equal initial data, $F(0)=F^\eta(0)\in\ell^2$, satisfy
\begin{equation}
    \lim_{\eta\rightarrow0}\sup_{T\in[0,T_\eta]}\norm{F^\eta(T) -  F(T) }_{\ell^2} =0\,,
\end{equation} 
for $T_\eta =C_\mu\log\langle 1/\eta\rangle$ where the constant is determined in \eqref{eq:mu-def-1-0} and \eqref{eq:Teta-def-1-0-0}. In particular, this implies 
\begin{eqnarray}
    \lim_{\eta\rightarrow0}|\braket{\chi_0,\phi_{T_\eta/\eta^2}}| & =& 1\;
    \nonumber\\
    \lim_{\eta\rightarrow0}
    |\braket{\chi_k,\phi_{T_\eta/\eta^2}}| & =& 0\;
    \rm{for \; all \;}k\geq1\,.
\end{eqnarray}
Consequently, by $L^2$ mass conservation, the solution $\phi_{T_\eta/\eta^2}$ to the original system \eqref{eq:boson-subsystem}, \eqref{eq:field-subsystem} converges to the ground state ray as $\eta\rightarrow0$ and $T_\eta\rightarrow\infty$.
\end{theorem}

\begin{proof}
Our strategy consists of the following steps.
We first apply a spectral projection to separate low-energy from high-energy modes. The high-energy tail can be controlled uniformly in time due to the boundedness of the conserved energy; given any $\epsilon >0$, there exists $N=N(\epsilon)$ such that the tail accounting for modes $k>N$ is bounded by $O(\epsilon)$ in an appropriate norm.
The low-energy modes are then shown to be dominated by a finite-dimensional coupled system of ODEs of $N+1$ modes, while the remainder consists of terms coupling low to high modes, and low-energy terms vanishing as $\eta\rightarrow0$. Control of the low-energy terms vanishing as $\eta\rightarrow0$ requires weighted $L^2_s$ decay of the half-wave propagator combined with a limiting absorption principle, among other ingredients. 

Given an arbitrary $\epsilon>0$, let $N=\ceil{ C_E\epsilon^{-2/\theta}}$ so that \lemref{lem:low-en-phi-1-0} holds.
Let $I\subset\R_+$ be a compact interval that contains the set 
$$
    \Big\{|E_j-E_{j'}|\;\Big|\;\,0\leq j,j'\leq N\;,\;j\neq j'\;\Big\}\;\subset I\,.
$$ 
Furthermore, assume $s>\frac12$, and let $0<\gamma<\min\{s-\frac12,1\}$.

Denote by $P_N$ the spectral projection onto the span of the eigenspaces associated with the first $N+1$ eigenvalues $E_0<E_1<\cdots<E_N$ of the free boson Hamiltonian $\cH_0=-\Delta+V$. Let $\phi_t^<:=P_N \phi_t=\sum_{0\leq k\leq N}e^{iE_k T/\eta^2}F_k^\eta(T)\rchi_k$ and $\phi_t^>:=(1-P_N)\phi_t$ with $t=\frac{T}{\eta^2}$. 
From \lemref{lm-spectral-tight-1-0} and \lemref{lem:low-en-phi-1-0} follows that the high energy tail is bounded by
\begin{eqnarray}
    \|\phi^>_t\|_{L^2}^2 = \sum_{k>N}|F^\eta(T)|^2 < \eps^2 
\end{eqnarray}
uniformly in $T$ and $\eta$.

Then, we obtain the low-energy contributions
\begin{eqnarray}
    \label{eq:def-low-en-phi-1-0}
    i\partial_t\phi_t^< &=& \cH_0\phi_t^<+\eta^2 v*|\phi_t^< |^2 \phi_t^<
    \nonumber\\
    &&+\eta P_N \Big( w*(u_t^<+\overline{u_t^<})\phi_t^< \Big)
    + \mathcal{Rem_{\phi}^<(t)}
\end{eqnarray}
where 
\begin{eqnarray}
    u_t^< &:= e^{-it\omega(-i\nabla)}u_0 
    - i\eta \int_0^t e^{-i(t-\tau)|\nabla|} (w*|\phi_\tau^<|^2) \, d\tau.
    \label{eq:duhamel-field-ll-thm}
\end{eqnarray}
It follows from \lemref{lem:low-en-phi-1-0} that $\mathcal{Rem_{\phi}^<(t)}$ can be controlled by a bound of order $O(\langle T\rangle^2\epsilon\log\langle1/\eta^2\rangle)$.
This implies that $\phi^<_t$ satisfies the self-consistent nonlinear system obtained from setting $\mathcal{Rem_{\phi}^<(t)}$ to zero, up to errors that converge to zero in a diagonal limit $\epsilon=\epsilon(\eta)\rightarrow0$, $T\in[0,T_\eta]$ with $T_\eta\rightarrow\infty$ as $\eta\rightarrow0$, as we demonstrate next.

Accordingly, let $P_N F^\eta=(F_0^\eta,\dots,F_N^\eta,0,0,\dots)$ and $P_N F=(F_0,\dots,F_N,0,0,\dots)$, and $ P_N^\perp:=(\1-P_N)$. We control the difference between $F^\eta$ and $F$ in the system \eqref{eq:equation-of-F-1-0} by
\begin{eqnarray}\label{eq:Feta-F-diff-1-0}
    \| F^\eta(T)-F(T) \|_{\ell^2}
    &\leq& \| P_N (F^\eta(T)-P_N F(T)) \|_{\ell^2} 
    \nonumber\\
    &&\qquad\qquad
    + \| P_N^\perp F^\eta(T) \|_{\ell^2}
    + \| P_N^\perp F(T) \|_{\ell^2}
    \nonumber\\
    &\leq&
    \| P_N (F^\eta(T)-P_N F(T)) \|_{\ell^2} 
    + 2\;\epsilon  
\end{eqnarray} 
using \lemref{lm-spectral-tight-1-0}.

To control the difference $P_N (F^\eta(T)-P_N F(T))$, we introduce the truncation remainders
\begin{eqnarray}
    \mathcal{Q}_N(T) &:=& P_N(M[F(T)]-M[P_N F(T)])
    \nonumber\\
    \mathcal{Q}_N^\eta(T) &:=& P_N(M[F^\eta(T)]-M[P_N F^\eta(T)])
    \nonumber\\
    &&\qquad
    + P_N \mathcal{Rem}^\eta[F^\eta(T)] + P_N \mathcal{Rem_{\phi}^<}(T)
\end{eqnarray}
where we recall  the effective cascade operator
\begin{equation}
    M[F]_k = \sum_{k'} M_{k,k'} \abs{F_{k'}}^2 F_k  \,,
\end{equation}
and where $P_N \mathcal{Rem}^\eta$ is given in \eqref{eq:Remeta-def-1-0-0} and $\mathcal{Rem_{\phi}^<}$ in \eqref{eq:Rem-l2-def-1-1-0}.
Then, the low energy modes satisfy the following projected equations,
\begin{align}
    \pt_T P_N F &= P_N M[P_N F] +\mathcal{Q}_N(T) 
    \label{eq:equation-of-F}\\
    \pt_T P_N F^\eta &= P_N M[P_N F^\eta] + \mathcal{Q}_N^\eta(T) .  \label{eq:equation-of-F-eta-trunc}
\end{align}
Let 
\begin{eqnarray}
    \ch_N^\eta(T):=P_N F^\eta(T)-P_N F(T) \,.
\end{eqnarray}
Then, $\ch_N^\eta(0)=0$ since the initial data agree, and we find
\begin{eqnarray}
    \partial_T \ch_N^\eta = P_N(M[P_N F^\eta(T] - M[P_N F])
    + \mathcal{Q}_N^\eta  - \mathcal{Q}_N 
\end{eqnarray}
We define the linearized operator
\begin{equation}
    \label{def:the-linearized-operator}
    \cL[P_N F,P_N F^\eta] \, \ch_n^\eta := M[P_N F^\eta] - M[P_N F] \,.
\end{equation}
By \lemref{lem:L-uniform-bound-energy}, we have  strong continuity and uniform boundedness of the generator $\cL[P_N F,P_N F^\eta](T)$ in $\mathcal{B}(\ell^2(\C))$. Therefore, the operator family generates a unique, strongly continuous evolution system $\{U_\eta(T, S)\}_{0 \leq s \leq T \leq T_{max}}$. Because the remainder term is continuous in time, the solution to the Cauchy problem is given by the Duhamel formula
\begin{equation}
    \ch_N^\eta(T) = U_\eta(T, 0) \ch_N^\eta(0) + \int_0^T U_\eta(T, S)
    ( \mathcal{Q}_N^\eta  - \mathcal{Q}_N )(S) \, dS. 
\end{equation}
where $\ch^\eta(0) = 0$ by hypothesis. Furthermore, the standard growth estimate for the evolution family yields the uniform bound
\begin{equation}
     \|U_\eta(T, S)\|_{\mathcal{B}(\ell^2} \leq \exp\left( \int_S^T \|\cL[P_N F,P_N F^\eta](\tau)\|_{\mathcal{B}(\ell^2)} \, d\tau \right) \leq e^{C_L(T-S)}. 
\end{equation}
Let 
\begin{eqnarray}
    \cB_N^\eta(T):= \int_0^T dS\; (\mathcal{Q}_N^\eta  - \mathcal{Q}_N) \,,
\end{eqnarray}
so that
\begin{eqnarray}
    \lefteqn{
    \int_0^T dS\; U_\eta(T, S)
    ( \mathcal{Q}_N^\eta  - \mathcal{Q}_N )(S)  
    }
    \nonumber\\
    &=&
    \int_0^T dS \; U_\eta(T, S) \; \partial_S \cB_N^\eta(S)
    \nonumber\\
    &=&U_\eta(T,T)\cB_N^\eta(T)-U_\eta(T, 0) \;  \cB_N^\eta(0)
    - \int_0^T dS \; (\partial_S U_\eta(T, S)) \;   
    \cB_N^\eta(S)
    \nonumber\\
    &=&\cB_N^\eta(T) 
    - \int_0^T dS \; (\partial_S U_\eta(T, S)) \;   
    \cB_N^\eta(S)\,.
\end{eqnarray}
because $U_\eta(T,T)=\1$ and $\cB_N^\eta(0)=0$.
From \lemref{lem:L-uniform-bound-energy}, we infer that 
\begin{eqnarray}
    \|\partial_S U_\eta(T,S) \|_{\cB(\ell^2)}
    &\leq& \|U_\eta(T,S) \|_{\cB(\ell^2)}
    \|\cL[P_N F,P_N F^\eta](S)\|_{\mathcal{B}(\ell^2)}
    \nonumber\\
    &\leq& C_L e^{C_L T} \,.
\end{eqnarray}
This implies 
\begin{align*}
    \|\ch_N^\eta(T)\|_{\ell^2} &\leq 
    \Big\|\int_0^T U_\eta(T, S)  
    ( \mathcal{Q}_N^\eta  - \mathcal{Q}_N )(S)\, dS
    \Big\|_{\ell^2} 
    \nonumber\\
    &\leq 
    (1+C_L T e^{C_L T})\sup_{S\in[0,T]}\|\cB_N^\eta(S)\|_{\ell^2} 
\end{align*}  
From the bounds on the truncation remainder terms proven in \lemref{lem:low-en-phi-1-0} and \lemref{lm:QN-QNeta-bds-1-0},  it follows that for $T\in[0,T_\eta]$,
\begin{eqnarray}
    \sup_{S\in[0,T_\eta]}\|\cB_N^\eta\|_{\ell^2} & \leq &
    C_{\cE}\langle T_\eta \rangle^2 \;\epsilon\;\log\braket{1/\eta^2} \, + \,
    C_{s,\gamma,\cE}\langle T_\eta\rangle^2 \;(1+\delta_N^{-2})\; N\; \eta^{2\gamma}\log \langle 1/\eta^2\rangle
    \nonumber\\
    &&\qquad +C_{s,v,w}\braket{T_\eta}^2 
    \frac{(1+E_N)}{\delta_N}  N^{\frac32}
    \eta^2   \log \langle 1/\eta^2 \rangle\,.
\end{eqnarray}
Here we use the upper bound in \eqref{eq:EN-upplowbd-1-0} to conclude that
\begin{eqnarray}
    (1+E_N) \leq  C N^{\frac23} 
\end{eqnarray}
and we recall that
where we recalled that
\begin{eqnarray}
    \delta_N^{-1}\leq C N^{\kappa}
\end{eqnarray} 
by the Diophantine condition in \hypref{asm:spectral-genericity}.
We thus arrive at 
\begin{eqnarray}
    \sup_{S\in[0,T_\eta]}\|\cB_N^\eta\|_{\ell^2} & \leq &
    C_{\cE}\langle T_\eta \rangle^2 \;\epsilon \; \log \langle 1/\eta^2 \rangle\, 
    \nonumber\\
    &&\qquad + \,C_{s,\gamma,\cE,s,v,w}\braket{T_\eta}^2 
    N^{3+2\kappa}  
    \eta^{2\gamma}   \log \langle 1/\eta^2 \rangle\,
\end{eqnarray}
by consolidating terms in the upper bound.

Next, we observe that for 
\begin{eqnarray}\label{eq:Teta-def-1-0-0}
    T_\eta=C_\mu \log\langle1/\eta\rangle
    \;{\rm with} \;C_\mu=\frac{\mu}{C_L} 
\end{eqnarray}
we find
\begin{eqnarray}
    C_L T_\eta e^{C_L T_\eta} \leq C \eta^{-\mu} 
    \log \langle1/\eta\rangle
\end{eqnarray}
Therefore, 
\begin{eqnarray}
    \|\ch_N^\eta(T)\|_{\ell^2}
    &\leq&
    C \eta^{-\mu} 
    \log \langle1/\eta\rangle
    \sup_{S\in[0,T_\eta]}\|\cB_N^\eta(S)\|_{\ell^2} 
    \nonumber\\
    &\leq&
    C \eta^{-\mu} 
    (\log \langle1/\eta\rangle)^4
    \Big(\epsilon+N^{3+2\kappa}  
    \eta^{2\gamma}  \Big)
\end{eqnarray}
We next use that $N = \ceil{ C_E\epsilon^{-2/\theta}}$ with $\theta>0$ defined in \eqref{eq:theta-def-1-0}, so that   
\begin{eqnarray}
    \|\ch_N^\eta(T)\|_{\ell^2}
    &\leq&
    C    
    ( \log\langle1/\eta\rangle)^4
    \eta^{-\mu} \Big(\epsilon 
    +\epsilon^{-2(3+2\kappa)/\theta}\eta^{2\gamma}
     \Big) \,.
\end{eqnarray}
Next, we chose
\begin{eqnarray}
    \epsilon(\eta) =\eta^K \,
\end{eqnarray}
with
\begin{eqnarray}
    K:= \min\Big\{\frac{\gamma}{2}\;,\;\frac{\gamma \; \theta}{2(3+2\kappa)}\Big\}\;>0 \,,
\end{eqnarray}
so that 
\begin{eqnarray}
    \|\ch_N^\eta(T)\|_{\ell^2}
    &\leq&
    C  
    (\log\langle1/\eta\rangle)^4
    \eta^{-\mu} (\eta^K+\eta^\gamma) \,,
\end{eqnarray}
and we choose
\begin{eqnarray}\label{eq:mu-def-1-0}
    \mu:=\frac12\min\{K,\gamma\} \,.
\end{eqnarray}
Consequently, taking the limit as $\eta \to 0$ gives
\begin{equation} 
    \lim_{\eta \to 0} \sup_{T\in[0,T_\eta]}
    \|\ch_N^\eta(T)\|_{\ell^2} =0\,.
\end{equation}
From \eqref{eq:Feta-F-diff-1-0} and \lemref{lem:low-en-phi-1-0}, we thus find that
\begin{eqnarray}
    \lefteqn{
    \lim_{\eta\rightarrow0}
    \sup_{T\in[0,T_\eta]}\| F^\eta(T)-F(T) \|_{\ell^2}
    }
    \nonumber\\
    &\leq& 
    \lim_{\eta\rightarrow0}\Big(
    \sup_{T\in[0,T_\eta]}
    \Big(\| P_N (F^\eta(T)-P_N F(T)) \|_{\ell^2} 
    \Big)
    + C\;\epsilon(\eta)  \Big)
    \nonumber\\
    &\leq& 
    \lim_{\eta\rightarrow0}\Big(
    \sup_{T\in[0,T_\eta]}\|\ch_N^\eta(T)\|_{\ell^2}
    + C\;\epsilon(\eta) \Big)
    \nonumber\\
    &=&0 \,.
\end{eqnarray}
We conclude that 
\begin{eqnarray}
    \lim_{\eta\rightarrow0}
    \sup_{T\in[0,T_\eta]}\|F^\eta(T) - F(T)\|_{\ell^2} \;=\;0  
\end{eqnarray} 
along a diagonal limit with $T_\eta=C\log\braket{1/\eta}$.
As a consequence, we also obtain that
\begin{eqnarray}
    \lim_{\eta\rightarrow0}|\braket{\chi_0,\phi_{T_\eta/\eta^2}}| & =& 1\;
    \nonumber\\
    \lim_{\eta\rightarrow0}
    |\braket{\chi_k,\phi_{T_\eta/\eta^2}}| & =& 0\;
    \rm{for \; all \;}k\geq1\,.
\end{eqnarray}
We therefore conclude that the solution $\phi_{T_\eta/\eta^2}$ to the original system \eqref{eq:boson-subsystem}, \eqref{eq:field-subsystem} converges to the ground state ray as $\eta\rightarrow0$ and $T_\eta\rightarrow\infty$, under conservation of the $L^2$ mass.
\end{proof}

In the rest of this section, we prove the lemmas used for the proof of \thmref{thm:convergence-eta}.
From the uniform boundedness of the energy in the boson subsystem, we obtain the following spectral tightness property.

\begin{lemma}[Spectral tightness]
\label{lm-spectral-tight-1-0}
Recall that $V(x)>-C_0+c_0|x|^{1+\beta}$ for $\beta>0$, by the coercivity assumption \hypref{asm:confining-potential}.

There are positive constants $c_*, C_*, C^*$ depending on $V$, $c_0,C_0,\beta$ such that
\begin{eqnarray}\label{eq:EN-upplowbd-1-0}
    c_*(N+1)^\theta-C_* \leq E_N \leq C^* (N+1)^{\frac23}
\end{eqnarray}
with 
\begin{eqnarray}\label{eq:theta-def-1-0}
    \theta:=\frac{2(1+\beta)}{3(3+\beta)} \,.
\end{eqnarray}
Moreover, given any $\epsilon>0$, there exists $N=\ceil{ C_E\epsilon^{-2/\theta}}$ such that 
\begin{eqnarray}
    \Big(\sum_{k>N}|F^\eta(T)|^2\Big)^{\frac12}   
    \;,\;\;
    \Big(\sum_{k>N}|F(T)|^2\Big)^{\frac12} < \eps 
\end{eqnarray}
uniformly in $T$ and $\eta$.
\end{lemma}

\begin{proof}
 Writing $\alpha:=1+\beta$, we define the comparison operator
\begin{equation}
    \cH_c:=-\Delta+c|x|^\alpha.
\end{equation}
The coercivity assumption \hypref{asm:confining-potential} implies that $\cH_0\geq \cH_c-C_0$.
The minimax principle, therefore, implies
\begin{equation}\label{eq:minmax-comparison-Hc}
    E_N(\cH_0)
    \geq
    E_N(\cH_c)-C_0.
\end{equation}
Thus, it suffices to obtain a lower bound on the
eigenvalues of $\cH_c$.

Let
\begin{equation}
    \mathcal{N}_c(\Lambda)
    :=
    \#\{N\geq0:E_N(\cH_c)\leq\Lambda\}
\end{equation}
denote the eigenvalue counting function of $\cH_c$. Then, by Birman-Schwinger theory \cite{CFKS1987schrodinger},
\begin{equation}\label{eq:counting-function-phase-space-bound}
    \mathcal{N}_c(\Lambda)
    \leq
    C \int_{\R^3}
    \bigl(\Lambda-c|x|^\alpha\bigr)_+^{3/2}\,dx,
\end{equation}
for $\Lambda$ sufficiently large. 
The integral in \eqref{eq:counting-function-phase-space-bound} can be evaluated by scaling, yielding in spherical coordinates
\begin{align}
    \eqref{eq:counting-function-phase-space-bound}&=
    C
    \int_0^{(\Lambda/c)^{1/\alpha}}
    (\Lambda-cr^\alpha)^{3/2}r^2\,dr.
    \nonumber\\ 
    &=
    4\pi c^{-3/\alpha}
    \Lambda^{\frac32+\frac3\alpha}
    \int_0^1
    (1-s^\alpha)^{3/2}s^2\,ds
    \nonumber\\
    &=
    C_{\alpha,c}
    \Lambda^{\frac32+\frac3\alpha}.
\label{eq:phase-space-scaling}
\end{align}
where we applied the change of variables $r=(\Lambda/c)^{1/\alpha}s$.
It follows that
\begin{equation}
    \mathcal{N}_c(\Lambda)
    \leq
    C_{\alpha,c}
    (1+\Lambda)^{q_\alpha},
    \qquad
    q_\alpha
    :=
    \frac32+\frac3\alpha.
\label{eq:counting-function-polynomial}
\end{equation}
Thus,
\begin{equation}
    N+1
    \leq
    \mathcal{N}_c(E_N(\cH_c))
    \leq 
    C_{\alpha,c}
    \bigl(1+E_N(\cH_c)\bigr)^{q_\alpha},
\end{equation}
using the definition of
$E_N(\cH_c)$, combined with \eqref{eq:counting-function-polynomial}.

Hence,
\begin{equation}
    E_N(\cH_c)
    \geq
    c_{\alpha,c}
    (N+1)^{1/q_\alpha}
    -1.
\label{eq:Hc-eigenvalue-lower}
\end{equation}
where
\begin{equation}
    \frac{1}{q_\alpha}
    =
    \frac{2(1+\beta)}{3(3+\beta)} =: \theta,
\end{equation}
we infer that \eqref{eq:minmax-comparison-Hc} and
\eqref{eq:Hc-eigenvalue-lower} imply
\begin{equation}
    E_N
    \geq
    c_*
    (N+1)^{\theta}
    -C_*
\end{equation}
follows.

For the upper bound, we fix a cube $Q=(0,L)^3$. By smoothness of $V$, $V_Q:=\sup_{x\in Q}V(x)<\infty$. 
For $m\geq1$, let $D_m\subset H^1_0(Q)$, extended by zero to $\R^3$, be spanned by the following modes satisfying Dirichlet boundary conditions,
\begin{eqnarray}
    \prod_{\ell=1}^3
    \sin\big(\frac{\pi n_\ell x_\ell}{L}\big)
\end{eqnarray}
for $1\leq n_\ell\leq m$. Then, $dim (D_m)=m^3$, and for every nonzero $f\in S_m$,
\begin{eqnarray}
    \frac{\langle f, \cH_0 f\rangle}{\|f\|_{L^2}^2}
    \leq V_Q+\frac{3\pi^2 m^2}{L^2} \,.
\end{eqnarray}
We choose now $m=\lceil (N+1)^{1/3}\rceil$. Since $m^3\geq N+1$, the minimax principle yields
\begin{eqnarray}
    E_N\leq V_Q + \frac{3\pi^2}{L^2}\lceil (N+1)^{1/3}\rceil^2 \leq C^* (N+1)^{\frac23}
\end{eqnarray}
as claimed.

Next, assume that we are given an arbitrary $\epsilon>0$. Due to the energy bound \eqref{eq:energy-lower-bound-final}, we find that
\begin{equation}
    \sum_{k\geq0}(E_k+C_0+1)|F_k^\eta(T)|^2 < C_{\cE}
\end{equation}
where the constant $C_{\cE}$ only depends on the value of the conserved energy.
Then, using that
$E_k+C_0+1>0$, and the fact that the eigenvalues are increasing,  
\begin{align}\label{eq:tail-from-energy}
    \sum_{k>N}|F_k^\eta(T)|^2
    &\leq
    \frac{1}{E_{N+1}+C_0+1}
    \sum_{k>N}
    (E_k+C_0+1)|F_k^\eta(T)|^2
    \nonumber\\
    &\leq
    \frac{C_{\cE}}{E_{N+1}+C_0+1}
    \nonumber\\
    &\leq
    C_{\cE}'
    N^{-\theta}  
    \nonumber\\
    &\leq \epsilon^2 \,,
\end{align}
for $N=\ceil{ C_E\epsilon^{-2/\theta}}$.
All constants are independent of $\eta$ and $T$.

Moreover, we recall from \eqref{eq:mass-flow-towards-low-energy-modes} that
\begin{eqnarray}
    \sum_k E_k|F_k(T)|^2\leq \sum_k E_k|F_k(0)|^2 \,.
\end{eqnarray}
Following the same arguments as above, we obtain
\begin{eqnarray}
    \sum_{k>N}|F(T)|^2 \leq C \eps^2 
\end{eqnarray}
where we may choose $N$ large enough so that $C=1$.
\end{proof}
Exploiting this spectral tightness property, we separate the high- and low-energy modes as follows. 

\begin{lemma}
\label{lem:low-en-phi-1-0}
Given $\epsilon>0$, let $N=\ceil{ C_E\epsilon^{-2/\theta}}$, and denote by $P_N$ the spectral projection onto the span of the eigenspaces associated with the first $N+1$ eigenvalues $E_0<E_1<\cdots<E_N$ of the free boson Hamiltonian $\cH_0=-\Delta+V$. Let $\phi_t^<:=P_N \phi_t$ and $\phi_t^>:=(1-P_N)\phi_t$. Then, for any finite $0<T<\infty$ and $t=\frac{T}{\eta^2}$, 
    \begin{eqnarray}
    \label{eq:def-low-en-phi-1-0-in-lem}
        i\partial_t\phi_t^< &=& \cH_0\phi_t^<
        +\eta^2 P_N\Big( v*|\phi_t^< |^2 \phi_t^<\Big)
        \nonumber\\
        &&+\eta P_N \Big( w*(u_t^<+\overline{u_t^<})\phi_t^< \Big)
        + Rem_{\phi}^<(t)
    \end{eqnarray}
where 
    \begin{eqnarray}
        u_t^< &:= e^{-it\omega(-i\nabla)}u_0 
        - i\eta \int_0^t e^{-i(t-\tau)\omega(-i\nabla)} (w*|\phi_\tau^<|^2) \, d\tau.
        \label{eq:duhamel-field-ll-lem}
    \end{eqnarray}
Let
\begin{eqnarray}\label{eq:Rem-l2-def-1-1-0}
    \mathcal{Rem_{\phi}^<}(T):=\Big(
    \Big\langle\chi_k\;,\; \frac1{i\eta^2}e^{iT E_k/\eta^2}\;Rem_{\phi}^< (T/\eta^2)\Big\rangle\Big)_{k\geq0}
    \;\in\;\ell^2\,.
\end{eqnarray}
Then,
    \begin{eqnarray}
        \int_0^{T} dS \;\Big\| \mathcal{Rem_{\phi}^<}(S)\Big\|_{\ell^2} \;\leq\; 
        C_{u_0,v,w,\cE}\; \braket{T}^2 \; \epsilon\;
        \log\langle1/\eta^2\rangle      \,.
    \end{eqnarray}
The constant only depends on the interaction potentials $v$ and $w$.
Moreover,
    \begin{eqnarray}
        \|\phi_t^>\|_{L^2} <C \epsilon \,.
    \end{eqnarray}
The constants are independent of $N$, $\eta$, and $t$.
\end{lemma}

\begin{proof}
Let $\epsilon>0$ and $N=\ceil{ C_E\epsilon^{-2/\theta}}$.
Then, with $t=T/\eta^2$, we have
     \begin{eqnarray}
        \|\phi_t^>\|_{L^2}^2 = \sum_{k>N}|F_k^\eta(T)|^2 <C \epsilon^2 \,.
    \end{eqnarray}
from \lemref{lm-spectral-tight-1-0}.

Next, we find from $\partial_t \phi_t^<= P_N\partial_t \phi_t$ that the remainder term in \eqref{eq:def-low-en-phi-1-0} is given by
    \begin{eqnarray}
        Rem_\phi^<(t) &=& Rem^<_{Har}(t)+Rem^<_{phot}(t)
    \end{eqnarray}
where the Hartree error term is given by
    \begin{eqnarray}
        Rem^<_{Har}(t) &=& \eta^2 P_N\,v*(|\phi_t|^2-|\phi_t^<|^2)\phi_t 
        +\eta^2 P_N \,v*|\phi_t|^2 \phi_t^>
    \end{eqnarray}
and the error term stemming from the photon interaction is given by
    \begin{eqnarray}
        Rem^<_{phot}(t) &=& 
        \eta P_N\Big( w*(2\Re{u_t-u_t^<} )
        \phi_t^< \Big)
        \nonumber\\
        &&+\eta P_N\Big( w*(2\Re{u_t})\phi_t^> \Big)
    \end{eqnarray}
where
    \begin{eqnarray}
        u_t-u_t^< &=  
        - i\eta \int_0^t e^{-i(t-\tau)|\nabla|} (w*(|\phi_\tau|^2-|\phi_\tau^<|^2) \, d\tau
    \end{eqnarray}
holds.
    
First of all, we have
    \begin{eqnarray}
        \|u_t-u_t^<\|_{L^\infty} &\leq& 
        - \eta \int_0^t \Big\| (e^{-i(t-\tau)\omega(-i\nabla)} w)*(|\phi_\tau|^2-|\phi_\tau^<|^2) \Big\|_{L^\infty} \, d\tau
        \nonumber\\
        &\leq &\eta \int_0^t \Big\| e^{-i(t-\tau)\omega(-i\nabla)} w \|_{L^\infty}\, \||\phi_\tau|^2-|\phi_\tau^<|^2\|_{L^1} \, d\tau
        \nonumber\\
        &\leq &\eta \|w \|_{W^{3,1}}\|\phi_\tau+\phi_\tau^<\|_{L^2} \,
        \|\phi_\tau^>\|_{L^2} 
        \int_0^t \langle t-\tau \rangle^{-1}\, d\tau
        \nonumber\\
        &\leq &C\|w \|_{W^{3,1}}\braket{T} 
        \;\epsilon \; \eta \log\langle 1/\eta^2\rangle
    \end{eqnarray}
    for $t=\frac{T}{\eta^2}$, and using $\|\phi_\tau+\phi_\tau^<\|_{L^2} \leq 2\|\phi_\tau\|_{L^2}=2$ and $\|\phi_\tau^>\|_{L^2} <C\epsilon$, as well as the dispersive bound in \lemref{lem:initial-field-remainder} for the half wave propagator.

    Similarly, we have
    \begin{eqnarray}
        \|u_t\|_{L^\infty} &\leq& \|e^{-it\omega(-i\nabla)}u_0 \|_{L^\infty}+
        \eta \int_0^t \Big\| (e^{-i(t-\tau)\omega(-i\nabla)} w)*|\phi_\tau|^2 \Big\|_{L^\infty} \, d\tau
        \nonumber\\
        &\leq &\langle t\rangle^{-1}\|u_0\|_{W^{3,1}}
        +\eta \|w \|_{W^{3,1}}\|\phi_\tau\|_{L^2}^2
        \|\phi_\tau\|_{L^2} 
        \int_0^t \langle t-\tau \rangle^{-1}\, d\tau
        \nonumber\\
        &\leq &\langle T/\eta^2 \rangle^{-1} \|u_0\|_{W^{3,1}}\; 
        +C\|w \|_{W^{3,1}}\braket{T}
        \; \eta \log\langle 1/\eta^2\rangle  \,
    \end{eqnarray}
    from similar considerations.

    Therefore,
    \begin{eqnarray}
        \|Rem^<_{Har}(t)\|_{L^2} &\leq& 
        \eta^2 \|v\|_{L^\infty}\|\phi_t+\phi_t^<\|_{L^2}
        \|\phi_t^>\|_{L^2}\|\phi_t^<\|_{L^2}
        +\eta^2\|v\|_{L^\infty}\|\phi_t\|_{L^2}^2 \|\phi_t^>\|_{L^2}
        \nonumber\\*
        &\leq& 
        C \|v\|_{L^\infty}\epsilon \; \eta^2 \,.
    \end{eqnarray}
    Moreover,
    \begin{eqnarray}
        \|Rem^<_{phot}(t)\|_{L^2} &\leq& 
        C\eta \| w \|_{L^1} \|u_t-u_t^<\|_{L^\infty}\|\phi_t^< \|_{L^2}
        +C\eta \| w\|_{L^1}\|u_t\|_{L^\infty}\|\phi_t^> \|_{L^2}
        \nonumber\\
        &\leq &
        \langle T/\eta^2\rangle^{-1}\|u_0\|_{W^{3,1}}\| w\|_{L^1}\;\epsilon\;
        \nonumber\\
        &&\qquad
        +C\| w\|_{L^1}\|w \|_{W^{3,1}}\braket{T}
        \; \epsilon \; \eta^2 \log\langle 1/\eta^2\rangle   \,.
    \end{eqnarray}
Thus, for $t=\frac{T}{\eta^2}$, we define
\begin{eqnarray}
    \mathcal{Rem_{\phi}^<}(T):=\Big(\Big\langle\chi_k\;,\; \frac{1}{i\eta^2}e^{iT E_k/\eta^2}\;Rem_{\phi}^< (T/\eta^2)\Big\rangle\Big)_{k\geq0}
    \;\in\;\ell^2
\end{eqnarray}
we obtain 
\begin{eqnarray}
        \int_0^T dS \;\Big\|\mathcal{Rem_{\phi}^<}(S)\Big\|_{\ell^2}
        &=&
        \int_0^{T} dS \;\Big\|\frac{1}{\eta^2}Rem_{\phi}^<(S/\eta^2)\Big\|_{L^2}
        \nonumber\\
        &\leq& 
        C_{u_0,v,w,\cE}\; 
        \langle T\rangle^2  \; \epsilon\;\log\langle1/\eta^2\rangle      \,.
\end{eqnarray}
This proves the claim.
\end{proof}

\smallskip

We now pick an arbitrary but fixed, small $\epsilon>0$, and let $N=\ceil{ C_E\epsilon^{-2/\theta}}$ as assumed in the two preceding lemmata.
Spectral tightness ensures that the high-energy modes with $k>N$ only contribute negligible errors, and we may now control the limit $\eta\rightarrow0$ for the system defined by the low-energy modes.


\begin{lemma}[Bounding the linearized operator]
\label{lem:L-uniform-bound-energy}
For simplicity of notation, we write $\ch^\eta\equiv\ch_N^\eta$.
Let $\cL[P_N F, P_N F^\eta]$ be the linearized operator defined in \eqref{def:the-linearized-operator}, and $\ch^\eta = P_N F^\eta - P_N F$ be the difference sequence between the two solutions. Let the assumptions of \lemref{lem:bound-Fk2-Mkkd} be satisfied. Then, 
\begin{equation}
\|\cL[P_N F,P_N F^\eta]\ch^\eta\|_{\ell^2} \le \tilde{C} \|\ch^\eta\|_{\ell^2}
\end{equation}
for a uniform constant $\tilde{C}:=\tilde{C}(s,w,C_V,C_E) > 0$. 
\end{lemma}

\begin{proof}

Fix $0\leq k\leq N$, and expand the $k$-th component as 
\begin{equation}\prc{\cL[P_N F,P_N F^\eta] \ch^\eta }_k = \text{I}_k + \text{II}_k + \text{III}_k,
\end{equation}
where
\begin{align}
    \text{I}_k &:= \left(\sum_{k'} M_{k,k'} |F_{k'}^\eta|^2\right) \ch^\eta_k, \\
    \text{II}_k &:= F_k \sum_{k'} M_{k,k'} F_{k'}^\eta \cc{\ch^\eta_{k'}}, \\
    \text{III}_k &:= F_k \sum_{k'} M_{k,k'} \cc{F_{k'}} \ch^\eta_{k'}.
\end{align}
We first establish a uniform bound on the multiplier $a_k^\eta := \sum_{k'} |M_{k,k'}| |F_{k'}^\eta|^2$. 
By \lemref{lem:bound-Fk2-Mkkd}, we obtain 
\begin{equation}
 \|\text{I}\|_{\ell^2}^2 = \sum_k |a_k^\eta \ch^\eta_k|^2 \le \left( \sup_k a_k^\eta \right)^2 \sum_k |\ch^\eta_k|^2 \le (C'_E)^2 \|\ch^\eta\|_{\ell^2}^2. 
\end{equation}
To bound the off-diagonal term II, we apply the Cauchy--Schwarz inequality to the inner sum over $k'$,
\begin{equation}
|\text{II}_k|^2 \le |F_k|^2 \left( \sum_{k'} |M_{k,k'}| |F_{k'}^\eta|^2 \right) \left( \sum_{k'} |M_{k,k'}| |\ch^\eta_{k'}|^2 \right) = |F_k|^2 a_k^\eta \left( \sum_{k'} |M_{k,k'}| |\ch^\eta_{k'}|^2 \right). 
\end{equation}
Summing over $k$ and applying the uniform bound $a_k^\eta \le C'_E$, 
\begin{equation}
\|\text{II}\|_{\ell^2}^2 \le C'_E \sum_k |F_k|^2 \left( \sum_{k'} |M_{k,k'}| |\ch^\eta_{k'}|^2 \right). 
\end{equation}
Applying Tonelli's theorem and factoring out $|\ch^\eta_{k'}|^2$,
\begin{equation}
\|\text{II}\|_{\ell^2}^2 \le C'_E \sum_{k'} |\ch^\eta_{k'}|^2 \left( \sum_{k} |M_{k',k}| |F_k|^2 \right) =  C'_E \sum_{k'} |\ch^\eta_{k'}|^2 a_{k'}[F].
\end{equation}
By identical logic on $F$ instead of $F^\eta$, \lemref{lem:bound-Fk2-Mkkd} gives $\sup_{k'} a_{k'}[F] \le C'_E$. Therefore:
\begin{equation}
 \|\text{II}\|_{\ell^2}^2 \le C'_E (C'_E) \sum_{k'} |\ch^\eta_{k'}|^2 = (C'_E)^2 \|\ch^\eta\|_{\ell^2}^2. 
\end{equation}
The term III is structurally identical to the term II and satisfies the same bound. Summing the bounds for $\text{I}$, $\text{II}$, and $\text{III}$ via the triangle inequality concludes the proof.
\end{proof}

\begin{remark}[$\mathbb{R}$-linearity of the linearized operator]
\label{rem:R-linear}
Due to the presence of complex conjugation in the expansion of the cubic difference (specifically in terms $\text{II}$ and $\text{III}$), the operator $\cL[P_N F, P_N F^\eta]$ acts as an $\mathbb{R}$-linear bounded operator on the complex space $\ell^2$. Consequently, for the abstract Cauchy problem and the application of the Duhamel formula in the subsequent convergence theorem, we formally view $\ell^2$ as a real Banach space. This perspective preserves all norm bounds and operator estimates without requiring $\mathbb{C}$-linearity.
\end{remark}

\begin{lemma}\label{lm:QN-QNeta-bds-1-0}
Let  $\epsilon>0$ and $N=N(\epsilon)=\ceil{ C_E\epsilon^{-2/\theta}}$ so that \lemref{lem:low-en-phi-1-0} holds. Let $I\subset\R_+$ be a compact interval that contains the set 
$$
    \Big\{|E_j-E_{j'}|\;\Big|\;\,0\leq j,j'\leq N\;,\;j\neq j'\;\Big\}\;\subset I\,.
$$ 
Moreover, assume $s>\frac12$, and let $0<\gamma<\min\{s-\frac12,1\}$. 
Then, for any finite $0<T<\infty$, the truncation remainders
\begin{eqnarray}
    \mathcal{Q}_N(T) &:=& P_N(M[F(T)]-M[P_N F(T)])
    \nonumber\\
    \mathcal{Q}_N^\eta(T) &:=& P_N(M[F^\eta(T)]-M[P_N F^\eta(T)])
    \nonumber\\
    &&\qquad
    + P_N \mathcal{Rem}^\eta[F^\eta(T)] + P_N \mathcal{Rem_{\phi}^<}(T)
\end{eqnarray}
satisfy the following bounds,
\begin{eqnarray}\label{eq:QN-F-bd-1-1}
    \int_0^T dS \; \|\mathcal{Q}_N(S)\|_{\ell^2} \leq 
    C_{\cE,T}\; \epsilon
\end{eqnarray}
and 
\begin{eqnarray}  
    \Big\|\int_0^T dS \; \mathcal{Q}_N^\eta(S) \Big\|_{\ell^2}
    &\leq&   C_{u_0}\braket{T} \eta \log \langle 1/\eta^2\rangle  +
    C_{s,\gamma,\cE}\braket{T}\,(1+\delta_N^{-2})
    N\; \eta^{2\gamma} \;\log \langle 1/\eta^2\rangle\,
    \nonumber\\
    &&\qquad +C_{s,v,w}\braket{T} 
    \frac{(1+E_N)}{\delta_N}  N^{\frac32}
    \eta^2   \log \langle 1/\eta^2 \rangle
    \nonumber\\
    &&\qquad +C_{u_0,v,w,\cE}\braket{T}^2 \epsilon\;  \log \langle 1/\eta^2\rangle   \,.
\end{eqnarray}
\end{lemma}

\begin{proof}
First of all, we have, for $0\leq k\leq N$,
\begin{eqnarray}
    (\mathcal{Q}_N)_k &=& (M[F]-M[P_N F])_k 
    \nonumber\\
    &=&\sum_{k'>N}M_{k,k'}|F_{k'}|^2\;F_k
\end{eqnarray}
so that 
\begin{eqnarray}\label{eq:QN-F-bd-1-0}
    \|\mathcal{Q}_N\|_{\ell^2}^2&=&
    \sum_{k\leq N}\Big(\sum_{k'>N}
    |M_{k,k'}| \; |F_{k'}|^2\Big) \; 
    \Big(\sum_{k''>N}|M_{k,k''}| \; |F_{k}|^2 \;|F_{k''}|^2\Big) 
    \nonumber\\
    &\leq&
    \Big( \sup_k \sum_{k'>N}
    |M_{k,k'}| \; |F_{k'}|^2\Big) \; 
    \sum_{k''>N}\Big(\sum_{k\leq N} |M_{k,k''}| \; |F_{k}|^2 \Big)\;|F_{k''}|^2 
    \nonumber\\
    &\leq& 
    \Big(\sup_{k}\sum_{k'>N} |M_{k,k'}| \; |F_{k'}|^2\Big)
    \Big(\sup_{k''}\sum_{k\leq N} |M_{k,k''}| \; |F_{k}|^2\Big)
    \Big(\sum_{k''>N} |F_{k''}|^2\Big)
    \nonumber\\
    &\leq& 
    C_{\cE}^2
    \Big(\sum_{k''>N} |F_{k''}|^2\Big) 
    \nonumber\\
    &\leq& 
    C_{\cE}^2 \; \epsilon^2
\end{eqnarray}
where we use \lemref{lem:bound-Fk2-Mkkd} and $|M_{k,k'}|=|M_{k',k}|$. The constant $C_{\cE}$ only depends on the  energy $\sum_{k}|F_k(0)|^2$ at initial time $T=0$.

Next we consider  
\begin{eqnarray} 
    \mathcal{Q}_N^\eta(T) &:=& P_N(M[F^\eta(T)]-M[P_N F^\eta(T)])
    \nonumber\\
    &&\qquad
    + P_N \mathcal{Rem}^\eta[F^\eta(T)]
    + P_N \mathcal{Rem_{\phi}^<}(T)
\end{eqnarray}
Similarly as in \eqref{eq:QN-F-bd-1-0}, we obtain
\begin{eqnarray} 
    \| P_N(M[F^\eta(T)]-M[P_N F^\eta(T)]) \|_{\ell^2}
    \;\leq\; C_{\cE}^2 \; \epsilon
\end{eqnarray}
where the constant $C_{\cE}$ only depends on the conserved energy of the combined system of bosons and photons. 
From \eqref{eq:Rem-pm-def-1-0}  and \eqref{eq:Rem-u0-def-1-0}, we have
\begin{eqnarray}\label{eq:Remeta-def-1-0-0}
    P_N\mathcal{Rem}^\eta &=&\mathcal{Rem}^{\Delta E\neq0}(t) +
    \mathcal{Rem}_{u_0}(t)
    \nonumber\\
    &&
    +\mathcal{Rem}^{(-),bd}(t)+\mathcal{Rem}^{(+),bd}(t)
    \nonumber\\
    &&
    +\mathcal{Rem}^{(-),der}(t)+\mathcal{Rem}^{(+),der}(t)
\end{eqnarray}
where 
\begin{eqnarray}
    \lefteqn{
    \mathcal{Rem}^{\Delta E\neq0}_k(t=T/\eta^2)
    }
    \nonumber\\
    &:&=
    \sum_{\substack{ 0\leq k', j,j'\leq N \\\Delta E_{k,k';j,j'} \neq 0} }  P_N  M_{k,k';j,j'} 
    e^{iT\Delta E_{k,k';j,j'}/\eta^2} \cc{F_{j'}(T)}   F_j(T) F_{k'}(T)
     \,.
\end{eqnarray}
Using \lemref{lem:initial-field-remainder}, \lemref{lm:Rem-der-bd-1-0}, \lemref{lm:Rem-bd-bd-1-0}, we obtain that
\begin{eqnarray}
    \lefteqn{
    \|\int_0^T dS \; P_N\mathcal{Rem}^\eta(S)\|_{\ell^2}
    }
    \nonumber\\
    &\leq&   C_{u_0}\braket{T} \eta \log \langle 1/\eta^2\rangle  +
    C_{s,\gamma,\cE}\braket{T}\,(1+\delta_N^{-2})
    N\; \eta^{2\gamma} \;\log \langle 1/\eta^2\rangle\,
    \nonumber\\
    &&\qquad +C_{s,v,w}\braket{T} 
    \frac{(1+E_N)}{\delta_N}  N^{\frac32}
    \eta^2   \log \langle 1/\eta^2 \rangle
\end{eqnarray}
where $0<\gamma<s-\frac12\leq\frac12$.
Moreover, \lemref{lem:low-en-phi-1-0} yields
\begin{eqnarray}
        \int_0^{T} dS \;\Big\| \mathcal{Rem_{\phi}^<}(S)\Big\|_{\ell^2} < 
        C_{u_0,v,w,\cE} \; \epsilon\;
        \langle T\rangle^2 \;\log\langle1/\eta^2\rangle      \,.
    \end{eqnarray}
Collecting all terms in the upper bound, we arrive at the claim.
\end{proof}

\section{Lamb Shift and resonance bounds}
\label{lamb-shift-and-resonant-bounds}

In the scaling $T = \eta^2 t$, the interactions between the confined bosons and the radiation field produce $\cO(1)$ effective transition rates (the Fermi Golden Rule) and energy renormalizations (the Lamb shift). In this section, we rigorously establish that these resonant coefficients are finite, well-defined, and uniformly bounded in the limit $\eta \to 0$, controlling the frequency-space singularity via the Limiting Absorption Principle in weighted spaces.

\begin{lemma}\label{lm:FGR-LS-a-priori-1-0}
    Let $s>\frac12$. Assume that $\gamma<\min\{s-\frac12,1\}$, and that $I\subset\R_+$ is a compact interval so that $|E_j-E_{j'}|\in I$.
    The Lamb shift and Fermi Golden Rule coefficients are bounded,
    \begin{equation}
        |\Lm  + i \Gm| < C_{s,\gamma,I}\|w\|_{L^2_s}^2
        \sqrt{1+E_{\min\{k,k'\}}}\sqrt{1+E_{\min\{j,j'\}}}
    \end{equation} 
    where the constant depends on $s$, $I$, $\gamma$.
\end{lemma}

\begin{proof}
    Invoking the limiting absorption principle in \lemref{lem:uniform-lap-for-R-eps-pm}, we find
\begin{eqnarray}
    |\Lm  + i \Gm|&<&C
    \|({\omega(-i \nabla)} - (E_{j} - E_{j'}) - i 0^+)^{-1}\|_{L^2_s\rightarrow L^2_{-s}}
    \nonumber\\
    &&\qquad
    \|w * (\cc{\rchi_{k'}} \rchi_k)\|_{L^2_s} \,
    \|w * (\cc{\rchi_{j'}} \rchi_{j}) \|_{L^2_s}
    \nonumber\\
    &<&C_{s,\gamma,I}\|w\|_{L^2_s}^2
    \|\cc{\rchi_{k'}} \rchi_k\|_{L^1_s}  \,
    \|\rchi_{j}\rchi_{j'}\|_{L^1_s} 
\end{eqnarray}
Here used Peetre's inequality, $\langle x\rangle^s\leq C \langle x-y\rangle^s\langle y\rangle^s$, to obtain
\begin{eqnarray}
    \|w * (\cc{\rchi_{k'}}\, \rchi_k)\|_{L^2_s} 
    &<&C
    \|w\|_{L^2_s}\|\cc{\rchi_{k'}} \rchi_k\|_{L^1_s}
    \nonumber\\
    &<&C
    \|w\|_{L^2_s} \|\rchi_{\min\{k,k'\}}\|_{L^2_s}\|\rchi_{\max\{k,k'\} }\|_{L^2}  
\end{eqnarray}
Passing to the last line, we used Cauchy-Schwarz to put the $\langle x\rangle^s$ weight on the eigenfunction associated with the smaller eigenenergy. 
The $L^2_s$ norms of the eigenfunctions are uniformly bounded by $\| \rchi_{k} \|_{L^2_s} < C \, \sqrt{1+E_k}$, which is a consequence of the coercivity of the trapping potential,
\begin{eqnarray}\label{eq:chi-L2s-bd-1-0}
    \|\rchi_{k}\|_{L^2_s}^2
    &=&\int dx \langle x\rangle^{2s}|\rchi_k(x)|^2
    \nonumber\\*
    &<&
    \int dx (V(x)+C)|\rchi_k(x)|^2
    \nonumber\\*
    &<&C \int dx (-\Delta+V(x)+C)|\rchi_k(x)|^2
    \nonumber\\*
    &<&C(1+E_k)
\end{eqnarray}
This proves the claim.
\end{proof}

\begin{lemma}[Uniform first-moment bound for the boson density]
\label{lem:moment-bound-for-phi}
Let $d=3$ and recall from \hypref{asm:confining-potential} in \assumref{the-assumptions} that the external trapping potential satisfies $V(x) \geq c|x|^{1+\beta} - C_0$ for some constants $c > 0$, $C_0 \geq 0$, and $\beta>0$. 
Let $(\phi_t, u_t)$ be a solution to the macroscopic dynamics such that the boson mass is conserved, $\|\phi_t\|_{L^2} = \|\phi_0\|_{L^2}$, and the total energy $\mathcal{E}[\phi_t, u_t]$ is bounded by $E_0$ for all $t \geq 0$, with
\begin{equation}
    \begin{aligned}
        \mathcal{E}[\phi, u] = & \frac{1}{2} \int_{\R^3} |\nabla \phi(x)|^2 \, dx + \frac{1}{2} \int_{\R^3} V(x) |\phi(x)|^2 \, dx + \frac{\lambda}{4} \int_{\R^3} (v * |\phi|^2)(x) |\phi(x)|^2 \, dx \\
        & + \frac{\eta}2 \int_{\R^3} (w * (\bar{u} + u))(x) |\phi(x)|^2 \, dx + \frac{1}{2} \int_{\R^3} \bar{u} \, \omega(-i \nabla) u \, dx.
    \end{aligned}
\end{equation}
Assume that the interaction kernel is symmetric ($w(x) = w(-x)$) and satisfies $w \in L^{3/2}(\R^3) \mcap L^\infty(\R^3)$, and let the field dispersion relation be $\omega(\xi) = |\xi|$. Then, there exists a constant $C > 0$ depending only on the initial data and system parameters, $\{E_0, \|\phi_0\|_{L^2}, w, \lambda, \eta, c, C_0\}$, such that for all $t \ge 0$
\begin{equation}\label{eq:phi-coerc-bd-1-0}
    \int_{\R^3} |x|^r |\phi_t(x)|^2 \, dx \leq C < \infty.
\end{equation}
for $0\leq r \leq 1+\beta$.
In particular, the first moment of the density $\rho_t = |\phi_t|^2$ is uniformly bounded in time.
\end{lemma}

\begin{proof}
By isolating the potential energy term from the total energy bound $\mathcal{E}[\phi_t, u_t] \leq E_0$, and discarding the non-negative boson kinetic energy $\frac{1}{2}\|\nabla \phi\|_{L^2}^2 \ge 0$, we obtain the upper bound
\begin{equation}
    \frac{1}{2} \int V(x) |\phi|^2 \leq E_0 - \frac{\lambda}{4} \int (w * |\phi|^2) |\phi|^2 - \frac{\eta}2 \int (w * (\bar{u} + u)) |\phi|^2 - \frac{1}{2} \| |\nabla|^{1/2} u \|_{L^2}^2.
\end{equation}
Since $w \in L^\infty$, Young's convolution inequality trivially bounds the Hartree interaction,
\begin{equation}
    -\frac{\lambda}{4} \int (w * |\phi|^2) |\phi|^2 \leq \frac{|\lambda|}{4} \|w\|_{L^\infty} \|\phi\|_{L^2}^4.
\end{equation}
We control the boson-field coupling term, utilizing the symmetry of $w$ to transpose the convolution, applying H\"{o}lder's inequality followed by Young's inequality, then invoking the fractional Sobolev continuous embedding $\dot{H}^{1/2}(\mathbb{R}^3) \hookrightarrow L^3(\mathbb{R}^3)$ to absorb the field variable
\begin{equation}
    \begin{aligned}
        - \eta \int (\bar{u} + u) (w * |\phi|^2) &\leq 2 \eta \|u\|_{L^3} \|w * |\phi|^2\|_{L^{3/2}} \\
        &\leq 2\eta \|u\|_{L^3} \|w\|_{L^{3/2}} \|\phi\|_{L^2}^2 \\
        &\leq 2C_S \eta \|w\|_{L^{3/2}} \|\phi\|_{L^2}^2 \| |\nabla|^{1/2} u \|_{L^2},
    \end{aligned}
\end{equation}
where $C_S$ is the Sobolev embedding constant. To handle the remaining dependence on the field $u$, we combine this coupling bound with the negative definite field kinetic energy and complete the square ($bX - \frac{1}{2}X^2 \le \frac{1}{2}b^2$),
\begin{equation}
    2C_S \eta \|w\|_{L^{3/2}} \|\phi\|_{L^2}^2 \norm{|\nabla|^{1/2} u}_{L^2} - \frac{1}{2} \norm{|\nabla|^{1/2} u}_{L^2}^2  \leq 2 C_S^2 \eta^2 \|w\|_{L^{3/2}}^2 \|\phi\|_{L^2}^4.
\end{equation}
Consolidating these estimates, the potential energy is bounded by the conserved quantities
\begin{equation}
    \frac{1}{2} \int V(x) |\phi|^2 \leq E_0 + \frac{|\lambda|}{4} \|w\|_{L^\infty} \|\phi\|_{L^2}^4 + 2 C_S^2 \eta^2 \|w\|_{L^{3/2}}^2 \|\phi\|_{L^2}^4.
\end{equation}
Applying the coercivity assumption on the trapping potential, $V(x) \geq c|x|^{1+\beta} - C_0$, we conclude
\begin{equation}
    \frac{c}{2} \int |x|^r |\phi|^2 \leq \frac{1}{2} \int V(x) |\phi|^2 + \frac{C_0}{2} \|\phi\|_{L^2}^2 \leq C(E_0, \phi_0, w, \lambda, \eta)
\end{equation}
for $0\leq r\leq 1+\beta$.
Because the right-hand side is composed entirely of time-independent conserved constants, the uniform moment bound holds for all $t \ge 0$.
\end{proof}

\begin{remark}
The uniform moment bound
\begin{equation}\sup_{t \ge 0} \int_{\mathbb{R}^3} |x| \rho_t(x) \, dx \le C\end{equation}
serves a critical dual purpose. Physically, it establishes the global-in-time macroscopic stability of the boson system; the coercive trap $V(x)$ prevents unbounded expansion despite continuous energy injection from the quantized radiation field and internal scattering. Analytically, this spatial localization is the fundamental prerequisite for the weak-coupling limit. The spatial moment bound guarantees the $C^1$-regularity of the Fourier-transformed density $\widehat{\rho}_t(\xi)$,
$$ |\widehat{\rho_t}(\xi + \delta) - \widehat{\rho_t}(\xi)| \le \int_{\mathbb{R}^3} |\delta| |x| |\rho_t(x)| \, dx = |\delta| \int_{\mathbb{R}^3} |x| |\rho_t(x)| \, dx \,,
$$ 
which is necessary to evaluate the highly singular principal-value resolvents (the Lamb shift) on the resonance sphere $|\xi| = |\Delta E|$ without divergence, thereby ensuring the effective cascade generator remains uniformly bounded.
\end{remark}

\section{Bounds on the remainders}
\label{controlling-the-remainders}

Let  $\epsilon>0$ and $N=\ceil{ C_E\epsilon^{-2/\theta}}$ so that \lemref{lem:low-en-phi-1-0} holds.  
In this section, we show that the remainders of the truncated system  \eqref{eq:equation-of-F} corresponding to \eqref{eq:Remkm}, \eqref{eq:Remkmplus}, namely 
\begin{align}\label{eq:Rem-pm-def-1-0}
    \mathcal{Rem}_k^{(\pm)}(t) =& \mathcal{Rem}^{(\pm),\mathrm{bd}}_k(t) + \mathcal{Rem}^{(\pm),\mathrm{der}}_k(t) 
\end{align}
and
\begin{equation}\label{eq:Rem-u0-def-1-0}
    \Rema_{u_0}(t) = 2 \eta (w *  \Re{ e^{-i t | \nabla|} u_0}) \vfi_t,
\end{equation}
all vanish as $\eta \to 0$. 

Before bounding the remainder term involving the freely evolved initial field $e^{-i t |\nabla|} u_0$, we note a subtle but critical scaling issue. In the macroscopic cascade equation, the remainder is evaluated on the time scale $T = \eta^2 t$ and multiplied by the scaling factor $\eta^{-2}$. Due to the free dispersive decay of the radiation field $\mathcal{O}(\jpp{t}^{-1})$, this scaled remainder behaves as $\mathcal{O}(\eta^{-1})$ near $T=0$, which is singular as $\eta \to 0$. Therefore, this remainder cannot be bounded uniformly in $L^\infty_T$. Instead, we exploit the rapid decay by establishing a strong $L^1_T$ vanishing bound over the macroscopic time interval, which is sufficient to close the Duhamel expansion.

\begin{lemma}[Free evolution term]
\label{lem:initial-field-remainder}
Assume for the initial photon field $u_0 \in W^{3,1}(\mathbb{R}^3)$ and $ w \in L^1(\mathbb{R}^3)$, and that $\|\vfi_t\|_{L^2} = 1$ for all $t \ge 0$ for the boson wave function. Under these regularity conditions, the free half-wave evolution satisfies the uniform dispersive decay estimate
\begin{equation}
    \|e^{-it|\nabla|}u_0\|_{L^\infty(\mathbb{R}^3)} \le C_0 \|u_0\|_{W^{3,1}} \langle t \rangle^{-1}
\end{equation}
for all $t \ge 0$, where $\jpp{t}:=(1 + t^2)^{1/2}$. 
Define the remainder contribution from the initial field
\begin{equation}
\label{eq:Rem_u0(t)}
    \mathcal{Rem}_{u_0}(t) := 2\eta \Big( w * \Re{ e^{-i t |\nabla|} u_0 } \Big) \vfi_t,
\end{equation}
and its associated modal coefficients $\mathcal{Rem}_k(t) := e^{i t E_k} \inp{\rchi_k, \mathcal{Rem}(t)}$.
Then, for any finite macroscopic time $T > 0$, the time-integrated, scaled remainder vanishes strongly in $\ell^2(\C)$ as $\eta \to 0$, that is 
\begin{equation}
    \int_0^T \left\| \frac{1}{\eta^2} \left( \mathcal{Rem}_{u_0,k}(s/\eta^2) \right)_{k \in \mathbb{N}} \right\|_{\ell^2(\C)} ds \;\le\; C \braket{T} \|w\|_{L^1}\|u_0\|_{W^{3,1}} \eta \log\langle1/\eta^2\rangle .
\end{equation}
\end{lemma}

\begin{proof}
By Parseval's identity and the orthonormality of the basis $\{\rchi_k\}$, the $\ell^2$-norm of the discrete modal sequence maps isometrically to the $L^2$-norm of the continuous spatial state,
\begin{equation}
    \left\| \left( \mathcal{Rem}_{u_0,k}(t) \right)_{k \in \mathbb{N}} \right\|_{\ell^2(\C)} = \|\mathcal{Rem}_{u_0}(t)\|_{L^2(\mathbb{R}^3)}.
\end{equation}
Utilizing the $L^2$-normalization of the macroscopic boson state ($\|\vfi_t\|_{L^2} = 1$), we separate the radiation field via H\"{o}lder's inequality, followed by Young's convolution inequality. That is, given $w \in L^1(\mathbb{R}^3)$,
\begin{align*}
    \|\mathcal{Rem}_{u_0}(t)\|_{L^2} &\le 2\eta \left\| w * \Re{ e^{-i t \omg(-i \nabla)} u_0 } \right\|_{L^\infty} \|\vfi_t\|_{L^2} \\
    &\le 2\eta \|w\|_{L^1} \left\| e^{-i t |\nabla|} u_0 \right\|_{L^\infty}.
\end{align*}
Applying the dispersive estimate for the free half-wave propagator yields the uniform-in-time continuous decay bound, 
\begin{equation}
    \|\mathcal{Rem}(t)\|_{L^2} \le C_0 \eta \|w\|_{L^1}\|u_0\|_{W^{3,1}} \langle t \rangle^{-1}.
\end{equation}
To evaluate the impact on the macroscopic cascade dynamics, we integrate the kinematically scaled remainder over the macroscopic time interval $s \in [0, T]$. Substituting the microscopic bound yields the integral
\begin{eqnarray}
    \int_0^T \frac{1}{\eta^2} \|\mathcal{Rem}_{u_0}(S/\eta^2)\|_{L^2} \, dS &\le& \int_0^T \frac{C_0 \eta \|w\|_{L^1}\|u_0\|_{W^{3,1}}}{\eta^2 \langle S/\eta^2 \rangle} \, dS
    \nonumber\\
    &=& C_0 \|w\|_{L^1}\|u_0\|_{W^{3,1}} \eta \int_0^T \frac{1}{\sqrt{\eta^4 + S^2}} \, ds,
    \nonumber\\
    &<& C \|w\|_{L^1} \|u_0\|_{W^{3,1}}
    \eta \;\ln\frac{T}{\eta^2} \,.
\end{eqnarray}   
which vanishes as $\eta \to 0$ for any finite macroscopic time $T > 0$. 
\end{proof}

\begin{remark}[Regularity of initial data]
While  \lemref{lem:initial-field-remainder} assumes $u_0 \in W^{3,1}(\mathbb{R}^3)$, the mathematically sharp requirement for the $O(t^{-1})$ free dispersive estimate in three dimensions is $u_0 \in B^2_{1,1}(\mathbb{R}^3) \hookrightarrow H^{1/2}(\R^3)$. If one restricts the data to standard Lebesgue-Sobolev spaces, any fractional regularity $s > 2$, such that $u_0 \in W^{s,1}(\mathbb{R}^3)$, is sufficient to guarantee the embedding $W^{s,1}(\mathbb{R}^3) \hookrightarrow B^2_{1,1}(\mathbb{R}^3)$ (but $B_{1,1}^2(\R^3) \hookrightarrow W^{s,1}(\R^3)$ for $s \leq 2$). 

However, we intentionally impose the integer regularity condition $W^{3,1}(\mathbb{R}^3)$ to establish a globally consistent functional framework for the full interacting system. In subsequent sections, controlling the nonlinear resonant interactions between the coupled fields requires distributing derivatives across product terms. Because fractional Sobolev spaces lack a clean Leibniz product rule, estimating the nonlinearities in the Duhamel expansion under the assumption $s > 2$ is analytically cumbersome. By assuming $W^{3,1}(\mathbb{R}^3)$ from the outset, we secure the free evolution decay while ensuring the data possesses sufficient integer regularity to close the nonlinear energy estimates using standard calculus.
\end{remark}

We recall from \eqref{eq:Remkm} the remainder term
\begin{align}
    \label{eq:Remkm-1-1}
    \Remkaa(t) =& \mathcal{Rem}^{(-),\mathrm{bd}}_k(t) + \mathcal{Rem}^{(-),\mathrm{der}}_k(t) 
\end{align}
where
\begin{align}
    \label{eq:Remkmbd-1-1}
    \mathcal{Rem}^{(-),\mathrm{bd}}_k(t) :=&  \sum_{k', j,j'}   \inp{w *(\cc{\rchi_{k'}}\rchi_k), \prc{\frac{ e^{- i t \omega(-i \nabla)} }{\omg(-i \nabla) - (E_{j} - E_{j'}) - i 0^+}}  w *(\cc{\rchi_{j'}} \rchi_{j}) } \nonumber\\
    & \times e^{i t (E_{k} - E_{k'})} A_j(0) \cc{A_{j'}}(0) A_{k'}(t) 
\end{align}
is the boundary term obtained from integration by parts, and
\begin{align}\label{Rmkmder-1-1} 
    \mathcal{Rem}^{(-),\mathrm{der}}_k(t)
    &:=
    \sum_{k',j,j'}
    \int_0^t ds\,
    \Bigg\langle
        w*(\overline{\rchi_{k'}}\rchi_k),
        \prc{\frac{1}{
            \omega(-i\nabla)-(E_j-E_{j'})-i0^+
        }}
        \nonumber\\
    &\hspace{3.0cm}\times
        e^{-i(t-s)\omega(-i\nabla)}
        w*(\overline{\rchi_{j'}}\rchi_j)
    \Bigg\rangle
    \nonumber\\
    &\hspace{1.5cm}\times
    e^{it(E_k-E_{k'})}
    e^{is(E_j-E_{j'})}
    \partial_s
    \left(
        A_j(s)\overline{A_{j'}(s)}
    \right)
    A_{k'}(t)
\end{align}
contains the time derivative of the amplitudes.
The energy modes are labeled by $0\leq j,j',k,k'\leq N(\epsilon)$.

Moreover, we recall from \eqref{eq:Remkmplus}
\begin{align}\label{eq:Remkmplus-1-1}
    \Remkab(t) =&   \mathcal{Rem}^{(+),\mathrm{bd}}_k(t)+\mathcal{Rem}^{(+),\mathrm{der}}_k(t)
\end{align}
with definitions given in \eqref{eq:Remkmplusbd} and
and \eqref{Rmkmderplus}.

\begin{lemma}
\label{lm:Rem-der-bd-1-0}
Let  $\epsilon>0$ and $N=N(\epsilon)=\ceil{ C_E\epsilon^{-2/\theta}}$ so that \lemref{lem:low-en-phi-1-0} holds. Let $I\subset\R_+$ be a compact interval that contains the set 
$$
    \Big\{|E_j-E_{j'}|\;\Big|\;\,0\leq j,j'\leq N\;,\;j\neq j'\;\Big\}\;\subset I\,.
$$ 
Moreover, assume $s>\frac12$, and let $0<\gamma<\min\{s-\frac12,1\}$. Then, for any finite $0<T<\infty$,
\begin{eqnarray}
    \label{eq:Remkmder-1-2}
    \int_0^T dS\; \|(\mathcal{Rem}^{(\pm),\mathrm{der}}_k(S/\eta^2))_{k}\|_{\ell^2}
    < C_{s,\gamma,I,\cE,w,T}  
    N \eta^{2\gamma}\log\langle 1/\eta^2\rangle 
\end{eqnarray}
converges to zero as $\eta\rightarrow0$.
\end{lemma}

\begin{proof}
Let $\psi(t):=e^{it\cH_0}\phi_t=\sum_k A_k(t)\rchi_k$. Then, 
\begin{eqnarray}
    \label{eq:mf-system-wp-1-1}
    i\partial_t \psi_t&=&
    e^{it\cH_0}\Big(\eta^2 (v*|\vfi_t|^2) + \eta \, w*(u_t+\overline{u_t})\Bigr)\vfi_t,
    \nonumber\\
    u_t&=&e^{-it|\nabla|}\,u_0-i\eta\int_0^t d\tau \,
    (e^{-i(t-\tau)|\nabla|}\,w)*|\phi_\tau|^2
\end{eqnarray}
with initial data
\begin{equation}
\vfi_0\in Y^1(\mathbb{R}^3),\qquad u_0\in W^{3,1}(\mathbb{R}^3).
\end{equation}
It follows that
\begin{eqnarray}
    \|u_t\|_{L^\infty}&<&C\langle t \rangle^{-1}\|u_0\|_{W^{3,1}}
    +\eta \int_0^t d\tau \; \langle t-\tau\rangle^{-1} \|w\|_{W^{3,1}}
    \|\phi_t\|_{L^2}^2
    \nonumber\\
    &<&C\langle t \rangle^{-1}\|u_0\|_{W^{3,1}}
    +C \eta \|w\|_{W^{3,1}} \log \langle t\rangle \,.
\end{eqnarray}
Therefore,  with $t=T/\eta^2$,
\begin{eqnarray}
    \|\partial_t\psi_t\|_{L^2}
    &<&C\eta^2 \|v\|_{L^\infty}\|\phi_t\|_{L^2}^2
    +C \eta \|w\|_{L^1}\langle T/\eta^2 \rangle^{-1}\|u_0\|_{W^{3,1}}  \|\phi_t\|_{L^2}
    \nonumber\\
    &&+C \eta^2 \log \langle T/\eta^2 \rangle \|w\|_{L^1} \|w\|_{W^{3,1}} \|\phi_t\|_{L^2}
    \nonumber\\
    &<&C \eta \langle T/\eta^2\rangle^{-1}
    +C \eta^2\log \langle T/\eta^2 \rangle \,.
\end{eqnarray}
It follows that for any finite macroscopic time $T>0$,
\begin{eqnarray}\label{eq:der-t-Ak-apriori-1-0}
    \big(\sum_j|\partial_t A_j(t)|^2\big)^{1/2}
    \;=\; \|\partial_t\psi_t\|_{L^2}
    \;\leq\; C \eta \langle T/\eta^2\rangle^{-1}
    +C \eta^2\log \langle T/\eta^2 \rangle \,,
\end{eqnarray} 
uniformly in $k\geq0$.

In \eqref{Rmkmder-1-1}, the summation in the index $0\leq k'\leq N$ yields $\phi_t^<=\sum_{k'=0}^N e^{-it E_{k'}}A_{k'}(t)$. Taking the $\ell^2$ norm with respect to the index $k$, we then find
\begin{align}\label{Rmkmder-1-2} 
    \Big\|(\mathcal{Rem}^{(-),\mathrm{der}}_k(t))_k\Big\|_{\ell^2}
    &=
    \Big\|\sum_{j,j'}
    \Big(\int_0^t ds\,
    \Bigg\langle
        w*(\overline{\phi_t^<}\rchi_k),
        \Big(\frac{1}{
            \omega(-i\nabla)-(E_j-E_{j'})-i0^+
        }
        \nonumber\\
    &\hspace{3.0cm}\times
        e^{-i(t-s)\omega(-i\nabla)}
        w*(\overline{\rchi_{j'}}\rchi_j)
    \Bigg\rangle
    \nonumber\\
    &\hspace{1.5cm}\times
    e^{it E_k}
    e^{is(E_j-E_{j'})}
    \partial_s
    \left(
        A_j(s)\overline{A_{j'}(s)}
    \right) \;\Big)_k \Big\|_{\ell^2}
    \\
    &\leq C \sum_{j,j'}
    \int_0^t ds\, 
    \Big\| \Big( w* \Big(\frac{1}{
            \omega(-i\nabla)-(E_j-E_{j'})-i0^+
        }
        \nonumber\\
    &\hspace{3.0cm}\times
        e^{-i(t-s)\omega(-i\nabla)}
        w*(\overline{\rchi_{j'}}\rchi_j) \Big)\phi_t^< \Big\|_{L^2}
    \nonumber\\
    &\hspace{1.5cm}\times
    |\partial_s
    \left(
        A_j(s)\overline{A_{j'}(s)}
    \right)| \,.
    \label{Rmkmder-1-3} 
\end{align}
We next use 
\begin{eqnarray}\label{eq:triple-L2s-bd-1-0}
    \| (w*f) g \|_{L^2} \leq C \| w \|_{L^2_s}\|f\|_{L^2_{-s}}\|g\|_{L^2_s}
\end{eqnarray}
which follows from
\begin{eqnarray}
    | (w*f)(x) | \leq  C \| w(x - \; \cdot \; )\|_{L^2_s} \| f \|_{L^2_{-s}} 
    \leq C \langle x\rangle^s \| w \|_{L^2_s} \| f \|_{L^2_{-s}} \,.
\end{eqnarray}
The latter is obtained from 
\begin{eqnarray}
    \| w(x - \; \cdot \; )\|_{L^2_s}^2 &=& \int dy \; \langle y \rangle^{2s} \; |w(x-y)|^2
    \nonumber\\
    &\leq& C \langle x \rangle^{2s} \int dy \; \langle x-y \rangle^{2s} \; |w(x-y)|^2
\end{eqnarray}
using Peetre's inequality, $\langle y \rangle^{2s}\leq C \langle x \rangle^{2s}\langle x-y \rangle^{2s}$. 

Let
$$
    I_N:=\Big\{|E_j-E_{j'}|\;\Big|\;\,0\leq j,j'\leq N\;,\;j\neq j'\;\Big\} \,.
$$ 
Clearly,
\begin{eqnarray}
    0<\delta_N\leq \min I_N\,.
\end{eqnarray}
Therefore, the compact interval $J_N:=[\delta_N, E_N-E_0]\subset\R_+$  contains $I_N$. We may therefore use \lemref{lem:one-sided-half-wave-resolvent} with $\lambda_0=\delta_N$.

We split the sum into $\sum_{j=j'}$ and $\sum_{j\neq j'}$, where if $E_j>E_{j'}$,
\begin{eqnarray}
    \lefteqn{
    \eqref{Rmkmder-1-3}\big|_{\sum_{j\neq j'}}
    }
    \nonumber\\
    &\leq& 
    C \| w \|_{L^2_s} \| \phi_t^< \|_{L^2_s} \sum_{j\neq j'}
    \int_0^t ds\,
    \Big\|\frac{1}{
            \omega(-i\nabla)-(E_j-E_{j'})-i0^+
        } 
        e^{-i(t-s)\omega(-i\nabla)}\Big\|_{L^2_s\rightarrow L^2_{-s}}
        \nonumber\\
    &&\qquad\times 
    \|w*(\overline{\rchi_{j'}}\rchi_j)\|_{L^2_s}
        |\partial_s
    \left(
        A_j(s)\overline{A_{j'}(s)}
    \right)|
    \nonumber\\
    &\leq& C_{s,\gamma}(1+\delta_N^{-2}) 
    \| w \|_{L^2_s} \| \phi_t^< \|_{L^2_s} 
    \int_0^t ds\, \langle t-s\rangle^{-\gamma}
        \label{Rmkmder-1-4}\\
    &&\qquad\times 
    \|w\|_{L^2_s} \sum_{j\neq j'} \|\overline{\rchi_{j'}}\rchi_j\|_{L^2_s}
    \left(
        |\partial_s A_j(s)||A_{j'}(s)|+|A_j(s)||\partial_s A_{j'}(s)|
    \right)
    \nonumber\\
    &\leq& C_{s,\gamma}(1+\delta_N^{-2})
    \| w \|_{L^2_s}^2 \| \phi_t^< \|_{L^2_s} 
    \int_0^t ds \langle t-s \rangle^{-\gamma}
        \label{Rmkmder-1-5}\\
    &&\qquad\times 
    \sum_{j\neq j'} \sqrt{1+E_{\min\{j,j'\}} } 
    \left(
        |\partial_s A_j(s)||A_{j'}(s)|+|A_j(s)||\partial_s A_{j'}(s)|
    \right)
    \nonumber\\
    &\leq& C_{s,\gamma}(1+\delta_N^{-2}) \braket{T}\; 
    \| w \|_{L^2_s}^2 \| \phi_t^< \|_{L^2_s}  
    \int_0^t ds \; \langle t-s \rangle^{-\gamma} \|\partial_s A(s)\|_{\ell^2}
    \label{Rmkmder-1-6}\\
    &&\qquad\times 
    N\Big(\sum_{j} (1+E_j ) \, |A_j(s)|^2\Big)^{\frac12}
    \nonumber\\
    &\leq& C_{s,\gamma,\cE,w}(1+\delta_N^{-2}) \braket{T}\; N  \; 
    \eta^{2\gamma}\log \, \langle 1/\eta^2\rangle 
    \,.
    \label{Rmkmder-1-8}
\end{eqnarray}
using \lemref{lem:one-sided-half-wave-resolvent} to pass to \eqref{Rmkmder-1-4}. We recalled \eqref{eq:chi-L2s-bd-1-0}  to pass to \eqref{Rmkmder-1-5}. To pass to \eqref{Rmkmder-1-6}, we used Cauchy-Schwarz in the $j,j'$ summation. To obtain the last line, we used \eqref{eq:der-t-Ak-apriori-1-0}, and the fact that both $\| \phi_t^< \|_{L^2_s}<C_{\cE}$, $(\sum_{j} (1+E_j ) |A_j(s)|^2)^{\frac12}<C_{\cE}$ are uniformly bounded a constant that depends on the value of the conserved energy, and monotonicity $E_{\min\{j,j'\}}<E_{\max\{j,j'\}}$.  

If, on the other hand, $E_j<E_{j'}$, we use
\begin{eqnarray}\label{eq:Posen-dispbd-1-0}
    \lefteqn{
    \| e^{-it|\nabla|}(|\nabla|+|E_j-E_{j'}|-i\epsilon)^{-1}w\|_{L^2_{-s}}
    }
    \nonumber\\
    &\leq &C_s \braket{t}^{-s}
    \|(|\nabla|+|E_j-E_{j'}|)^{-1}w\|_{L^2_{s}} 
    \nonumber\\
    &\leq&C_s \braket{t}^{-s}
    \; (|E_j-E_{j'}|)^{-1} +|E_j-E_{j'}|)^{-2} ) 
    \; \|w\|_{L^2_{s}}
    \nonumber\\
    &\leq&C_s (1+\delta_N^{-2})\braket{t}^{-s}
    \;   \|w\|_{L^2_{s}}
\end{eqnarray}
because $|E_j-E_{j'}|^{-1}\leq \delta_N^{-1}$. Thus, the same upper bound \eqref{Rmkmder-1-8} is valid.

For the sum $\sum_{j=j'}$, we use the zero mode decay estimate in \lemref{lem:one-sided-half-wave-resolvent} to obtain
\begin{eqnarray}
\eqref{Rmkmder-1-3}\big|_{\sum_{j=j'}} &\leq& 
    C \| |\nabla|^{-\frac12} w \|_{L^2_s}^2 \| \phi_t^< \|_{L^2_s} \sum_{j}
    \int_0^t ds\,
    \langle t-s\rangle^{-s}
        \nonumber\\
    &&\qquad\times 
    \||\nabla|^{-\frac12} w\|_{L^2_s}
    \|(\overline{\rchi_{j}}\rchi_j)\|_{L^1_s}
        |\partial_s
    \left(
        A_j(s)\overline{A_{j'}(s)}
    \right)|
    \nonumber\\
    &\leq& C_{s,\gamma} \; 
    \| |\nabla|^{-\frac12} w \|_{L^2_s}^2 \| \phi_t^< \|_{L^2_s}
    \int_0^t ds\,
    \langle t-s\rangle^{-s}
    \|\partial_s A(s)\|_{\ell^2}
    \label{Rmkmder-1-7}\\
    &&\qquad\times 
    N\Big(\sum_{j} (1+E_j ) |A_j(s)|^2\Big)^{\frac12}   
    \nonumber\\
    &\leq& C_{s,\gamma,\cE,w} \; \braket{T}   \; 
    N \; \eta^{2\gamma} \log\langle1/\eta^2\rangle 
    \,,
\end{eqnarray}
using similar arguments as above for $\sum_{j\neq j'}$.

It thus follows that
\begin{eqnarray}
    \int_0^T dS\; \|(\mathcal{Rem}^{(\pm),\mathrm{der}}_k(S/\eta^2))_{k}\|_{\ell^2}
    < C_{s,\gamma,\cE,w} (1+\delta_N^{-2})\braket{T}
    \; N \eta^{2\gamma}\log\langle 1/\eta^2\rangle
\end{eqnarray}
as claimed.
\end{proof}

\begin{lemma}
\label{lm:Rem-bd-bd-1-0}
Let  $\epsilon>0$ and $N=N(\epsilon)=\ceil{ C_E\epsilon^{-2/\theta}}$ so that \lemref{lem:low-en-phi-1-0} holds. Recall  
\begin{eqnarray}
    \delta_N = \min_{0\leq k,k';j,j'\leq N}\{|\Delta E_{k,k';j,j'}|\;{\rm nonzero}\} \,,
\end{eqnarray}  
and the Diophantine condition in \hypref{asm:spectral-genericity}. Then, for any finite $0<T<\infty$,
\begin{eqnarray}\label{eq:Remkmbd-1-2}
    \int_0^T dS\; \|(\mathcal{Rem}^{(\pm),\mathrm{bd}}_k(S/\eta^2))_{k}\|_{\ell^2}
    < C_{s,\gamma,\cE,w} (1+\delta_N^{-2})\braket{T}    N  \eta^{2\gamma} \,.
\end{eqnarray}
Furthermore, we have
\begin{eqnarray}
    \Big\|\int_{0}^T dS\; \mathcal{Rem}^{\Delta E\neq0}(S/\eta^2)\Big\|_{\ell^2}
    &\leq& C_{s,v,w}\braket{T} 
    \frac{(1+E_N)}{\delta_N}  N^{\frac32}
    \eta^2   \log \langle 1/\eta^2 \rangle \,.
\end{eqnarray}
\end{lemma}

\begin{proof}
As in the previous lemma, let
$I_N=\{|E_j-E_{j'}|\;\big|\;\,0\leq j,j'\leq N,j\neq j'\}$, and we again use that $0<\delta_N\leq \min I_N$.
Because the compact interval $J_N:=[\delta_N, E_N-E_0]\subset\R_+$  contains $I_N$, we may use \lemref{lem:one-sided-half-wave-resolvent} with $\lambda_0=\delta_N$.
Moreover, assume $s>\frac12$, and let $0<\gamma<\min\{s-\frac12,1\}$. 

Summation in the index $0\leq k'\leq N$ in \eqref{eq:Remkmbd-1-1}  yields $\phi_t^<=\sum_{k'=0}^N e^{-it E_{k'}}A_{k'}(t)$.
The $\ell^2$ norm with respect to the index $k$ results in
\begin{align}
        \label{eq:Remkmbd-1-3}
    \Big\|\big(\mathcal{Rem}^{(-),\mathrm{bd}}_k(t)\big)_k\Big\|_{\ell^2} 
    =& \Big\| \sum_{j,j'} 
    \Big( \inp{w *(\cc{\phi_t^<}\rchi_k),  \frac{ e^{- i t |\nabla|}  }{|\nabla| - (E_{j} - E_{j'}) - i 0^+} 
    w *(\cc{\rchi_{j'}} \rchi_{j}) } \nonumber\\
    & \times e^{i t E_{k} } A_j(0) \cc{A_{j'}}(0) \; \Big)_k \Big\|_{\ell^2}
    \nonumber\\
    \leq&C\sum_{j,j'} \Big\| 
    \Big(w*\Big(\frac{ e^{- i t |\nabla|}  }{|\nabla| - (E_{j} - E_{j'}) - i 0^+} 
    w *(\cc{\rchi_{j'}} \rchi_{j}) \Big) \phi_t^<  \Big\|_{L^2} \nonumber\\
    & \qquad \times | A_j(0)| \; |A_{j'}(0)| 
\end{align}
We split the sum into $\sum_{j=j'}$ and $\sum_{j\neq j'}$.
Using \eqref{eq:triple-L2s-bd-1-0} and \eqref{eq:chi-L2s-bd-1-0}, we obtain that when $E_j>E_{j'}$, the sum $\sum_{j\neq j'}$ is bounded by
\begin{align}\label{eq:Rembd-bd-1-2}
    \eqref{eq:Remkmbd-1-3}\big|_{\sum_{j\neq j'}}
    \leq& C \sum_{j,j'} \|w\|_{L^2_s}  \Big\| \frac{ e^{- i t |\nabla|}  }{|\nabla| - (E_{j} - E_{j'}) - i 0^+} \Big\|_{L^2_s\rightarrow L^2_{-s}}
     \| \phi_t^< \|_{L^2_s} \nonumber\\
    & \qquad \times \|w *(\cc{\rchi_{j'}} \rchi_{j})\|_{L^2_s} \; | A_j(0)| \; |A_{j'}(0)|
    \nonumber\\
    \leq& C_{s,\gamma}(1+\delta_N^{-2}) \|w\|_{L^2_s} \langle T/\eta^2 \rangle^{-\gamma}  
    \nonumber\\
    & \qquad \times  \sum_{j,j'}  \|w\|_{L^2_s} \|\cc{\rchi_{j'}} \rchi_{j}\|_{L^1_s} \; | A_j(0)| \; |A_{j'}(0)|
    \nonumber\\ 
    \leq& C_{s,\gamma}(1+\delta_N^{-2})  \|w\|_{L^2_s}^2  
    \langle T/\eta^2 \rangle^{-\gamma}  
    \nonumber\\
    & \qquad \times \sum_{j,j'}   \sqrt{1+E_{\min\{j,j'\}}} \; | A_j(0)| \; |A_{j'}(0)|
    \nonumber\\
    \leq&  C_{s,\gamma}(1+\delta_N^{-2})  \|w\|_{L^2_s}^2   
    \langle T/\eta^2 \rangle^{-\gamma}  
    \nonumber\\
    & \qquad \times N\Big( \sum_{j}  \sqrt{1+E_{j}} \; | A_j(0)|^2\Big)^{\frac12} \; 
    \Big( \sum_{j'} |A_{j'}(0)|\Big)^{\frac12}
    \nonumber\\
    \leq&  C_{s,\gamma,\cE,w}(1+\delta_N^{-2})    N \langle T/\eta^2 \rangle^{-\gamma}   \,. 
\end{align}
If $E_j<E_{j'}$, we use \eqref{eq:Posen-dispbd-1-0}, and the same upper bound \eqref{eq:Rembd-bd-1-2} is valid.

For the sum $\sum_{j = j'}$, we find
\begin{align}\label{eq:Rembd-bd-1-3}
    \eqref{eq:Remkmbd-1-3}\big|_{\sum_{j = j'}}
    =&\Big\| \sum_{j} 
    \Big( \inp{w *(\cc{\phi_t^<}\rchi_k),  \frac{ e^{- i t |\nabla|}  }{|\nabla|- i 0^+} 
    w *(\cc{\rchi_{j'}} \rchi_{j}) } \nonumber\\
    &\qquad \times e^{i t (E_{k} - E_{k'})} A_j(0) \cc{A_{j'}}(0) \; \Big)_k \Big\|_{\ell^2}
    \nonumber\\
    \leq& C \sum_{j} \||\nabla|^{-\frac12}w\|_{L^2_s}  
     \| \phi_t^< \|_{L^2_s} \nonumber\\
    & \qquad \times \|(|\nabla|^{-\frac12} w) *(\cc{\rchi_{j}} \rchi_{j})\|_{L^2_s} \; | A_j(0)|^2
    \nonumber\\
    \leq& C  \||\nabla|^{-\frac12}w\|_{L^2_s}^2  \langle t\rangle^{-s}
     \| \phi_t^< \|_{L^2_s} N 
     \Big(\sum_j| A_j(0)|^2 \Big)^{\frac12}
     \Big(\sum_j(E_j+1)| A_j(0)|^2 \Big)^{\frac12}
     \nonumber\\
    \leq& C_{s,\cE,w}   \; N \langle T/\eta^2\rangle^{-s} \,.
\end{align}
Therefore,
\begin{eqnarray}
    \int_0^T dS\; \eqref{eq:Remkmbd-1-3} \leq 
    C_{s,\gamma,\cE,w} \braket{T}  (1+\delta_N^{-2}) N \; \eta^{2\gamma}
\end{eqnarray}
This is the desired bound.

Finally, we recall
\begin{eqnarray}
    \lefteqn{
    \mathcal{Rem}^{\Delta E\neq0}_k(S/\eta^2)
    }
    \nonumber\\
    &&=
    \sum_{\substack{ 0\leq k', j,j'\leq N \\\Delta E_{k,k';j,j'} \neq 0} }  P_N  M_{k,k';j,j'} 
    e^{iS\Delta E_{k,k';j,j'}/\eta^2}
     \cc{F_{j'}^\eta(S)} F_j^\eta(S) F_{k'}^\eta(S)\,
\end{eqnarray}
so that 
\begin{eqnarray}
 \label{eq:RemkDEd-1-0}
    \lefteqn{
    \Big\|\int_0^T dS \;\big(\mathcal{Rem}^{\Delta E\neq0}_k(S/\eta^2)\big)_k\Big\|_{\ell^2_{k\leq N}} 
    }
    \nonumber\\
    &=& \Big\| \Big(\sum_{k',j,j'\leq N} M_{k,k';j,j'}   \int_0^T dS\;
    e^{iS\Delta E_{k,k';j,j'}/\eta^2} 
    \cc{F_{j'}^\eta(S)} F_j^\eta(S) F_{k'}^\eta(S)
    \Big)_k 
     \Big\|_{\ell^2_{k\leq N}}
    \nonumber
\end{eqnarray}
The time integral is bounded by 
\begin{eqnarray}
    \lefteqn{
    \Big|\int_0^T dS\;
    e^{iS\Delta E_{k,k';j,j'}/\eta^2} 
    \cc{F_{j'}^\eta(S)} F_j^\eta(S) F_{k'}^\eta(S) \Big|
    }
    \nonumber\\
    &\leq&
    \Big|\frac{\eta^2}{i\Delta E_{k,k';j,j'}}
    \Big(e^{iT\Delta E_{k,k';j,j'}/\eta^2}
    \cc{F_{j'}^\eta(T)} F_j^\eta(T) F_{k'}^\eta(T)
    -\cc{F_{j'}^\eta(0)} F_j^\eta(0) F_{k'}^\eta(0)\Big)\Big|
    \nonumber\\
    &&+\Big|\frac{\eta^2}{i\Delta E_{k,k';j,j'}}
    \int_0^T dS\;
    e^{iS\Delta E_{k,k';j,j'}/\eta^2} 
    \partial_S\big(\cc{F_{j'}^\eta(S)} F_j^\eta(S) F_{k'}^\eta(S)\big)\Big|
    \nonumber\\
    &\leq&\frac{2\eta^2}{\delta_N}\Big(
    |\cc{F_{j'}^\eta(0)} F_j^\eta(0) F_{k'}^\eta(0)|
    + |\cc{F_{j'}^\eta(T)} F_j^\eta(T) F_{k'}^\eta(T)|
    \Big)
    \nonumber\\
    &&+\frac{\eta^2}{\delta_N}\int_0^T dS\;
    \big|\partial_S\big(\cc{F_{j'}^\eta(S)} F_j^\eta(S) F_{k'}^\eta(S)\big)\big|
    \,,
\end{eqnarray}
where
\begin{eqnarray}
    \delta_N := \min_{0\leq k,k';j,j'\leq N}\{|\Delta E_{k,k';j,j'}|
    \;{\rm nonzero}\} \,.
\end{eqnarray}
Next, we control $M_{k,k';j,j'}$. Taking the $\ell^2$ norm in the index $k\leq N$ while keeping $k',j,j'$ fixed,
\begin{eqnarray}
    \|M_{k,k';j,j'}\|_{\ell^2_{k\leq N}}
    &\leq& 
    \|\inp{\rchi_k, v *(\cc{ \rchi_{j'}} \rchi_j) \rchi_{k'}} \|_{\ell^2_{k\leq N}}
    \nonumber\\
    &&\quad+ 
    \Big\|\inp{w *(\cc{\rchi_{k'}}\rchi_k),  
    \frac{ 1 }{|\nabla| - (E_{j} - E_{j'}) - i 0^+} 
    w *(\cc{\rchi_{j'}} \rchi_{j}) }
    \Big)_k \Big\|_{\ell^2_{k\leq N}}
    \nonumber\\
    &&\quad+ 
    \Big\|\inp{w *(\cc{\rchi_{k'}}\rchi_k),  
    \frac{ 1 }{|\nabla| + (E_{j} - E_{j'}) + i 0^+} 
    w *(\cc{\rchi_{j'}} \rchi_{j}) }
    \Big)_k \Big\|_{\ell^2_{k\leq N}}
\end{eqnarray}
For the Hartree term,
\begin{eqnarray}
    \|\inp{\rchi_k, v *(\cc{ \rchi_{j'}} \rchi_j) \rchi_{k'}} \|_{\ell^2_{k\leq N}}
    &\leq&\|v *(\cc{ \rchi_{j'}} \rchi_j) \rchi_{k'}\|_{L^2}
    \nonumber\\
    &\leq&\|v\|_{L^\infty}
\end{eqnarray}
since $\|\rchi_{k'}\|_{L^2}=1$.
Moreover, we use for $j\neq j'$ and $E_j>E_{j'}$,
\begin{eqnarray}\label{eq:res-bd-Rem-1-1-0}
    \lefteqn{
    \Big\|\inp{w *(\cc{\rchi_{k'}}\rchi_k),  
    \frac{ 1 }{|\nabla| - (E_{j} - E_{j'}) - i 0^+} 
    w *(\cc{\rchi_{j'}} \rchi_{j}) }
    \Big)_k \Big\|_{\ell^2_{k\leq N}} 
    }
    \nonumber\\
    &\leq&C 
    \Big\|w *\Big( 
    \frac{ 1 }{|\nabla| - (E_{j} - E_{j'}) - i 0^+} 
    w *(\cc{\rchi_{j'}} \rchi_{j}) 
    \Big)\rchi_{k'} \Big\|_{L^2}
    \nonumber\\
    &\leq&C_s \|w\|_{L^2_s}\|w*(\cc{\rchi_{j'}} \rchi_{j}) \|_{L^2_s}
    \|\rchi_{k'} \|_{L^2_s}
    \nonumber\\
    &\leq&C_s \|w\|_{L^2_s}^2\|\cc{\rchi_{j'}} \rchi_{j} \|_{L^1_s}
    \|\rchi_{k'} \|_{L^2_s}
    \nonumber\\
    &\leq&C_s \|w\|_{L^2_s}^2
    \sqrt{1+E_{\min\{j,j'\}}}
    \sqrt{1+E_{k'}}
\end{eqnarray}
where we invoked \eqref{eq:triple-L2s-bd-1-0}, \lemref{lem:uniform-lap-for-R-eps-pm}, and similar arguments as in \eqref{eq:Rembd-bd-1-2}. For the opposite case $E_j<E_{j'}$, the operator norm of the resolvent is bounded, but for simplicity, we bound the corresponding expression by \eqref{eq:res-bd-Rem-1-1-0}, which majorizes it. The case with "$+$" instead of "$-$" in the resolvent is similar.

For $j=j'$, we use 
\begin{eqnarray}
    \lefteqn{
    \Big\|\inp{w *(\cc{\rchi_{k'}}\rchi_k),  
    \frac{ 1 }{|\nabla|} 
    w *(|\rchi_{j}|^2)  }
    \Big)_k \Big\|_{\ell^2_{k\leq N}}
    }
    \nonumber\\
    &=&
    \Big\|\inp{(|\nabla|^{-\frac12}w) *(\cc{\rchi_{k'}}\rchi_k),  
    (|\nabla|^{-\frac12}w) *(|\rchi_{j}|^2})
    \Big)_k \Big\|_{\ell^2_{k\leq N}}
    \nonumber\\
    &\leq&C 
    \Big\|(|\nabla|^{-\frac12}w) *\Big( 
    (|\nabla|^{-\frac12}w) *(| \rchi_{j}|^2) 
    \Big)\rchi_{k'} \Big\|_{L^2}
    \nonumber\\
    &\leq&C_s \|w\|_{\dot H^{-1/2}}^2\|\cc{\rchi_{j'}} \rchi_{j}\|_{L^1}
    \|\rchi_{k'} \|_{L^2}
    \nonumber\\
    &\leq&C_s \|w\|_{\dot H^{-1/2}}^2 
\end{eqnarray}
This implies that for $j\neq j'$,
\begin{eqnarray}
    \sup_{0\leq k',j,j'\leq N}\|M_{k,k';j,j'}\|_{\ell^2_{k\leq N}}
    \leq C_{s,v,w}(1+E_N)
\end{eqnarray}
by monotonicity of the eigenenergies, $E_k\leq E_N$ for all $k\leq N$, and
\begin{eqnarray}
    \sup_{0\leq k',j\leq N}\|M_{k,k';j,j}\|_{\ell^2_{k\leq N}}
    \leq C_{s,v,w} 
\end{eqnarray}
for $j=j'$.

Therefore,
\begin{eqnarray}
    \eqref{eq:RemkDEd-1-0}&\leq&C_{s,v,w}
    \frac{\eta^2}{\delta_N} (1+E_N)
    \nonumber\\
    &&\qquad\times
    \sum_{0\leq k',j,j'\leq N}  \Big(
    |\cc{F_{j'}^\eta(0)} F_j^\eta(0) F_{k'}^\eta(0)|
    + |\cc{F_{j'}^\eta(T)} F_j^\eta(T) F_{k'}^\eta(T)|
    \nonumber\\
    &&\qquad\qquad+\int_0^T dS\;
    \big|\partial_S\big(\cc{F_{j'}^\eta(S)} F_j^\eta(S) F_{k'}^\eta(S)|\Big)
    \nonumber\\
    &&+ C_{s,v,w}
    \frac{ \eta^2}{\delta_N}
    \sum_{0\leq k',j\leq N}  
    \Big( | F_j^\eta(0)|^2  |F_{k'}^\eta(0)|
    +| F_j^\eta(T)|^2  |F_{k'}^\eta(T)|
    \nonumber\\
    &&\qquad\qquad
    + \int_0^T dS\;
    \big|\partial_S\big(| F_j^\eta(S) |^2 F_{k'}^\eta(S))|\Big)
    \nonumber\\
    &\leq&C_{s,v,w}
    \frac{ \eta^2}{\delta_N} (1+E_N) N^{\frac32}
    \Big(\|F^\eta\|_{\ell^2}^3 + \int_0^T dS\;\|\partial_S F^\eta\|_{\ell^2}\|F^\eta\|_{\ell^2}^2\Big)
    \nonumber\\
    &&+C_{s,v,w}
    \frac{ \eta^2}{\delta_N} N
    \Big(\|F^\eta\|_{\ell^2}^3 + \int_0^T dS\;\|\partial_S F^\eta\|_{\ell^2}\|F^\eta\|_{\ell^2}^2\Big)
    \nonumber\\
    &\leq&C_{s,v,w}
    \frac{ \eta^2}{\delta_N} (1+E_N) N^{\frac32}
    \Big(1 + \int_0^T dS\;\|\partial_S F^\eta\|_{\ell^2}
    \Big)
    \nonumber\\
    &\leq&C_{s,v,w}\langle T\rangle
    \frac{(1+E_N)}{\delta_N}  N^{\frac32}
    \eta^2   \log \langle 1/\eta^2 \rangle  
\end{eqnarray}
using Cauchy Schwarz. Moreover, we used mass conservation $\|F^\eta\|_{\ell^2}=1$, and 
\begin{eqnarray}
    \|\partial_S F^\eta\|_{\ell^2}
    = \eta^{-2} \|\partial_t A^\eta\|_{\ell^2}
    \leq  C \eta^{-1} \langle T/\eta^2\rangle^{-1}
    +C  \log \langle T/\eta^2 \rangle 
\end{eqnarray} 
from \eqref{eq:der-t-Ak-apriori-1-0}.

In summary, we arrive at
\begin{eqnarray}
    \Big\|\int_0^T dS \;\mathcal{Rem}^{\Delta E\neq0}(S/\eta^2)
    \Big\|_{\ell^2}
    &\leq&C_{s,v,w}\langle T\rangle
    \frac{(1+E_N)}{\delta_N}  N^{\frac32}
    \eta^2   \log \langle 1/\eta^2 \rangle  .
\end{eqnarray}
This proves the claim.
\end{proof}
\bigskip

\appendix
\label{appendix}

\section{Weighted decay estimates and LAP}
\label{appendix:analytic-preliminaries}

\begin{lemma}[Weighted decay for the free half-wave resolvent]
\label{lem:one-sided-half-wave-resolvent}
Let 
\begin{equation}
 	\|f\|_{L^2_s}:=\|\langle x\rangle^sf\|_{L^2(\mathbb R^3)},
 	\qquad \langle x\rangle=(1+|x|^2)^{1/2},
\end{equation}
and fix a compact interval $J=[\lambda_0,\lambda_1]\subset(0,\infty)$ with $\lambda_0>0$ small.  If $\frac12<s<1$ and
\begin{equation}
 	0<\gamma<\min\left\{s-\frac12,1\right\},
\end{equation}
then for $t>0$ and any $M>s$,
\begin{equation}
\label{eq:one-sided-main}
 	\sup_{\substack{\lambda\in J\\0<\epsilon \le1}}
 	\|e^{-it|\nabla|}(|\nabla|-\lambda-i\epsilon )^{-1} \, w\|_{L^2_{-s}}
 	\le C_{s,\gamma} \, (1+\lambda_0^{-2}) \langle t\rangle^{-\gamma} \, \|w\|_{L^2_s} \,.
\end{equation}
The combination  of $t>0$ and the sign of the imaginary part in $(|\nabla|-\lambda-i\epsilon )^{-1}$ is essential.  

Moreover, the following zero mode decay estimate holds,
\begin{eqnarray}
    |\langle g, e^{-it|\nabla|} |\nabla|^{-1} f\rangle|
    \leq C_{s}\langle t\rangle^{-s}\||\nabla|^{-\frac12} f\|_{L^2_{s}}
    \||\nabla|^{-\frac12} g\|_{L^2_{s}}
\end{eqnarray}
for a constant depending only on $\frac12<s<1$.
\end{lemma}

\begin{proof}
As in \lemref{lem:translated-cauchy}, let $\widetilde\rchi\in C_c^\infty(\R)$ be a smooth bump function with 
\begin{eqnarray}
    \widetilde\rchi(\rho)=
    \left\{
    \begin{array}{cl}
        1 & \rho\in (-\frac12,\frac12) \\
        0 & \rho\in  (-\infty,-\frac{3}4)\cup(\frac{3}4,\infty)
    \end{array}\right. \,,
\end{eqnarray}
and let $\chi(\rho):=\widetilde\chi((\rho-\lambda)/{\lambda_0})$, so  that $\rm{supp}(\chi)\subset(\frac{\lambda_0}4,\lambda_1+\frac{3\lambda_0}4)$ for every $\lambda\in J$.  
We decompose
\begin{equation}
\label{eq:frequency-decomposition}
 	(|\nabla|-\lambda-i\epsilon )^{-1}
	 =\rchi(|\nabla|)\, (|\nabla|-\lambda-i\epsilon )^{-1}
 	+(1-\rchi(|\nabla|)) \,(|\nabla|-\lambda-i\epsilon )^{-1}.
\end{equation}
On the annulus $\supp(\rchi)$, we have  
\begin{align}
\label{eq:polar-comparison}
 	\|\rho\, \widehat f(\rho\, \omega)\|_{H^s_\rho L^2_\omega}
 	&\le C_s 
    \|\widehat f\|_{H^s_\xi},\\
 	\|\widehat u\|_{H^{-s}_\xi}
 	&\le C_s 
    \|\rho\widehat u(\rho \, \omega)\|_{H^{-s}_\rho L^2_\omega}.
\end{align}
The first inequality follows because full derivatives control radial derivatives; the second is its dual.  Moreover,
\begin{equation}
 	\|\widehat f\|_{H^s_\xi}=c\|f\|_{L^2_s},\qquad
 	\|\widehat u\|_{H^{-s}_\xi}=c\|u\|_{L^2_{-s}}.
\end{equation}
Therefore, we obtain
\begin{equation}
\label{eq:near-final}
 	\|e^{-it|\nabla|}\, \rchi(|\nabla|)\, (|\nabla|-\lambda-i\epsilon )^{-1} \, w\|_{L^2_{-s}}
 	\le C_s (1+\lambda_0^{-M})\langle t\rangle^{-\gamma}\|w\|_{L^2_s}.
\end{equation}
by using \lemref{lem:translated-cauchy} with
$\mathcal H=L^2(\mathbb S^2)$ and $M>s$.

For the complementary term, we define
\begin{equation}
 	a_{\lambda,\epsilon }(\rho)
 	:=\frac{1-\rchi(\rho)}{\rho-\lambda-i\epsilon }.
\end{equation}
Here we obtain $\|\chi\|_{W^{1,\infty}}\leq C\lambda_0^{-1}$, and moreover, $|\frac1{\rho-\lambda}|\leq C\lambda_0^{-1}$  $\rho\in\rm{supp}(1-\chi)$. 
Therefore, 
\begin{eqnarray}
    \| a_{\lambda,\epsilon }\|_{W^{1,\infty}} < C \lambda_0^{-2} \,.
\end{eqnarray}  
Let $\max\{\frac12,\gamma\}<\sigma<\min\{s,1\}$.
Therefore, we find that 
\begin{eqnarray}\label{eq:alameps-term-1-0}
    \Big\|e^{-it|\nabla|}a_{\lambda,\epsilon }(|\nabla|)w\Big\|_{L^2_{-s}}
    &\leq&\Big\|e^{-it|\nabla|}a_{\lambda,\epsilon}(|\nabla|)w\Big\|_{L^2_{-\sigma}}
    \nonumber\\
    &\leq&\big\|e^{-it|\nabla|}\big\|_{L^2_\sigma\rightarrow L^2_{-\sigma}}
    \|a_{\lambda,\epsilon }(|\nabla|)w\|_{L^2_{\sigma}}
    \nonumber\\
    &\leq&C \langle t\rangle^{-\sigma}\|a_{\lambda,\epsilon }\|_{W^{1,\infty}}\|w\|_{L^2_{\sigma}}
    \nonumber\\
    &\leq&C \langle t\rangle^{-\gamma}\|a_{\lambda,\epsilon }\|_{W^{1,\infty}}\|w\|_{L^2_{s}}
\end{eqnarray}
where we used  \thmref{thm:half-wave-decay} and 
\begin{eqnarray}
    \|a_{\lambda,\epsilon }(|\nabla|)w\|_{L^2_{\sigma}}
    &=&\|a_{\lambda,\epsilon }\widehat w\|_{H^\sigma}
    \nonumber\\
    &\leq&\|a_{\lambda,\epsilon }\|_{W^{1,\infty}} \|\widehat w \|_{L^2}
    + \|a_{\lambda,\epsilon }\|_{L^\infty} \|\widehat w \|_{H^\sigma}
    \nonumber\\
     &\leq&C\lambda_0^{-2}\|w\|_{L^2_\sigma} \,.
\end{eqnarray}
The bounds obtained so far hold uniformly for all $\lambda\in J$.
Combining \eqref{eq:near-final} and \eqref{eq:alameps-term-1-0}, we therefore get
\begin{equation}
\label{eq:one-sided-main-1-1}
 	\sup_{\substack{\lambda\in J\\0<\epsilon \le1}}
 	\|e^{-it|\nabla|}(|\nabla|-\lambda-i\epsilon )^{-1} \, w\|_{L^2_{-s}}
 	\le C_{s,\gamma} \,(\lambda_0^{-2}\,+ (1+\lambda_0^{-M})) \langle t\rangle^{-\gamma} \, \|w\|_{L^2_s} \,
\end{equation}
for small $\lambda_0$.
We may choose $M=1>s$, whereby we arrive at the claim.

We add the remark that, due to the jump formula
\begin{equation}
 	(|\nabla|-\lambda+i0)^{-1}-(|\nabla|-\lambda-i0)^{-1}
 	=-2\pi i\,\delta(|\nabla|-\lambda),
\end{equation} the resolvent with the opposite sign $+i\epsilon$ contains
$-2\pi i e^{-it\lambda}\delta(|\nabla|-\lambda)w$, which is nonzero in
$L^2_{-s}$ for generic $w$. Hence, it cannot decay. Thus, the combination of signs $-i\epsilon$ and $t>0$ is significant for the result proven here. 

The zero mode decay estimate follows immediately from \thmref{thm:half-wave-decay}.
\end{proof}

\begin{lemma}[Time decay for translated Cauchy kernel]
\label{lem:translated-cauchy}
Let $\cH$ be a Hilbert space, and $F\in H^s(\R,\cH)$ with $s>\frac12$. Let $J=[\lambda_0,\lambda_1]\subset\mathbb (0,\infty)$ a compact interval, 
and let $\widetilde\rchi\in C_c^\infty(\R)$ be a smooth bump function with
\begin{eqnarray}
    \widetilde\rchi(\rho)=
    \left\{
    \begin{array}{cl}
        1 & \rho\in (-\frac12,\frac12) \\
        0 & \rho\in  (-\infty,-\frac{3}4)\cup(\frac{3}4,\infty)
    \end{array}\right. \,.
\end{eqnarray}
Define $\chi(\rho):=\widetilde\chi((\rho-\lambda)/{\lambda_0})$, and we note that $\rm{supp}(\chi)\subset(\frac{\lambda_0}4,\lambda_1+\frac{3\lambda_0}4)$ for every $\lambda\in J$.  
If
\begin{equation}
 	q_{t,\lambda,\epsilon }(\rho)
 	:=\frac{e^{-it\rho}\rchi(\rho)}{\rho-\lambda-i\epsilon },
\end{equation}
then, for $0<\gamma<s-\frac12$ and $M>s$,
\begin{equation}
\label{eq:cauchy-estimate}
 	\|q_{t,\lambda,\epsilon }F\|_{H^{-s}(\mathbb R;\mathcal H)}
 	\le C_M (1+\lambda_0^{-M})\langle t\rangle^{-\gamma}
 	\|F\|_{H^s(\mathbb R;\mathcal H)},
\end{equation}
uniformly for $t\ge0$, $\lambda\in J$, and $0<\epsilon \le1$. Notably, the constant is independent of $\lambda_0$ and $\lambda_1$, so that the bound depends only on $\lambda_0$.
\end{lemma}

\begin{proof}
For notational convenience, let $\mathcal{F}(h)(r)=\int_{\R } e^{-ir\rho}h(\rho)\,d\rho$ denote the 1 dimensional Fourier transform, and $R_{-\lambda-i\epsilon }(\rho)=(\rho-\lambda-i\epsilon )^{-1}$.  Contour integration gives
\begin{equation}\label{eq:F1-res-ker-1-0}
 	\mathcal{F}( R_{-\lambda-i\epsilon })(r)
 	=c_0e^{-ir(\lambda+i\epsilon )}\mathbf1_{(-\infty,0)}(r).
\end{equation}
Let 
\begin{equation}
    R_{-\lambda-i\epsilon }^{(\rchi,t)}(\rho):=e^{-it\rho} \rchi(\rho) R_{-\lambda-i\epsilon }(\rho)
\end{equation}
and
\begin{equation}\label{eq:ktr-res-ker-1-0}
	k_t(r):= 
    \mathcal{F}(R_{-\lambda-i\epsilon }^{(\rchi,t)})(r)
\end{equation}
The factor $e^{-it\rho}$ translates the kernel \eqref{eq:F1-res-ker-1-0} by $t$, and multiplication
by $\rchi$ convolves it with the Schwartz function $\mathcal{F}(\rchi)$.  Hence,
\begin{eqnarray}
\label{eq:kernel-majorant}
 	|k_t(r)|&\leq& \big( \; |\mathcal{F}(\rchi)|*
 	\mathbf1_{(-\infty,0)} \; \big)(r+t) 
\end{eqnarray}
The factor $e^{\epsilon  r}$ is bounded by 1 in \eqref{eq:F1-res-ker-1-0} because the support condition enforces $r<0$, and remains uniformly bounded in \eqref{eq:ktr-res-ker-1-0}. Therefore, the constants are uniform in $\epsilon $.

Given $M>0$, let
\begin{eqnarray}\label{eq:BMchi-def-1-0}
    B_{M,\chi} &:=& \|\mathcal{F}(\chi) \|_{L^1} 
    + \sup_{x\geq 1}x^M\int_x^\infty dy |\mathcal{F}( \chi)|
    \nonumber\\
    &\leq& C_{M,\widetilde\chi}\max\{1,\lambda_0^{-M} \} \,.
\end{eqnarray}
This is is a consequence of scaling, $|\mathcal{F}(\chi)(r)|=\lambda_0|\mathcal{F}(\widetilde\chi)(\lambda_0 r)|$, and \eqref{eq:BMchi-def-1-0} follows from a change of variables $r\rightarrow \lambda_0 r$ and the rapid decay of $\mathcal{F}(\widetilde\chi)$ (on a scale independent of $\lambda_0$).
The definition \eqref{eq:BMchi-def-1-0} implies that
\begin{eqnarray}
\label{eq:kernel-majorant-1-1}
 	|k_t(r-r')|&\leq& C \, B_{M,\chi}\,\big(
 	\mathbf1_{r-r'\leq -t+1} + \langle t+r-r'\rangle^{-M}\mathbf1_{r-r'\geq-t+1} \big) \,.
\end{eqnarray}
That is, $k_t(r)$ decays rapidly outside $\{ r \in \R \;|\; r+t \leq 1\}$.

Subsequently, we find 
\begin{eqnarray}
	\mathcal{F}(q_{t,\lambda,\epsilon }F)(r) = \int_{\R } dr' k_t(r-r')\mathcal{F}(F)(r') \,.
\end{eqnarray}
Thus, the operator conjugated by Sobolev weights is given by $T_t\otimes \mathbf1_{\mathcal H}$ with
\begin{eqnarray}
    (T_t f)(r):=\int_{\R}dr' K_t(r,r')f(r')
\end{eqnarray} 
and
\begin{equation}
	 K_t(r,r')=\langle r\rangle^{-s}k_t(r-r')
 	\langle r'\rangle^{-s} .
\end{equation}
The operator norm of $T_t\otimes\mathbf1_{\mathcal H}$ is bounded by 
\begin{eqnarray}
    \| T_t \otimes\mathbf1_{\mathcal H} \|^2 \leq \| T_t \|_{HS}^2 = 
    \int 
    \langle r\rangle^{-2s} |k_t(r-r')|^2
 	\langle r'\rangle^{-2s}\,dr\,dr' \,.
\end{eqnarray} 
Using Cauchy-Schwarz, we obtain two terms corresponding to the right-hand side of \eqref{eq:kernel-majorant-1-1}. 
The contribution from the term supported on $r-r'\leq-t+1$ in \eqref{eq:kernel-majorant-1-1} is bounded by
\begin{eqnarray}
\label{eq:volterra-hs}
 	\iint_{r-r'<-t+1}\langle r\rangle^{-2s}
 	\langle r'\rangle^{-2s}\,dr\,dr'
    &=&
    \int_{u\leq-t+1}du\;\Big(\int_{\R}dr 
    \langle r\rangle^{-2s}\langle r-u\rangle^{-2s}\Big)
    \nonumber\\
    &\leq&
    C_s \int_{u\leq-t+1}du\; 
    \langle u \rangle^{-2s}  
    \nonumber\\
    &\leq& C_s\langle t\rangle^{1-2s}
\end{eqnarray}
where we recall that $s>\frac12$.
The contribution for the term  in \eqref{eq:kernel-majorant-1-1} supported on $r-r'\geq-t+1$ is bounded by
\begin{eqnarray}
\label{eq:volterra-hs-1-1}
 	\lefteqn{
    \int_{r-r'\geq -t+1}\,dr\,dr' \;
    \langle r\rangle^{-2s}
    \langle t+r-r'\rangle^{-2M}
 	\langle r'\rangle^{-2s}
    }
 	\nonumber\\
    &=& 
    \int_{u\geq -t+1}du \;
    \langle t+u\rangle^{-2M}
    \int_{\R}dr\,\langle r\rangle^{-2s}
 	\langle r-u\rangle^{-2s} 
    \nonumber\\
    &\leq& 
    C_s
    \int_{u\geq -t+1}du \;
    \langle t+u\rangle^{-2M} \langle u\rangle^{-2s}
 	\nonumber\\
    &\leq& C_s\langle t\rangle^{-2s}
\end{eqnarray}
for $M>s$.
Here we substituted $u=r-r'$ to pass to the second line, and used the convolution bound $\int_{\R}dr\,\langle r \rangle^{-2s} \langle r-u\rangle^{-2s}\leq C_s \langle u\rangle^{-2s}$.
    
Thus, we arrive at
\begin{eqnarray}
    \| T_t \otimes\mathbf1_{\mathcal H}  \|  \leq C_s \, B_{M,\chi}  \,\langle t\rangle^{-\gamma}
\end{eqnarray}
with $0<\gamma<s-\frac12$.
This implies \eqref{eq:cauchy-estimate}.
\end{proof}

\begin{theorem}[Weighted decay for the free half-wave propagator]
\label{thm:half-wave-decay}
Let \(1/2<s<1\).  There exists a constant \(C_s>0\) such that
\begin{equation}
    \norm{e^{-it|\nabla|}f}_{L^2_{-s}(\R^3)}
    \leq C_s\langle t\rangle^{-s}
    \norm{f}_{L^2_s(\R^3)}
    \label{eq:weighted-half-wave-decay}
\end{equation}
for every \(t\in\R\) and every \(f\in L^2_s(\R^3)\).   
\end{theorem}

The proof uses the following radial lifting estimate.

\begin{lemma}[Radial lifting]
\label{lem:radial-lifting}
For \(F\in H^\sigma(\R^3)\), define
\begin{equation}
    (\mathcal RF)(\rho,\omega)=\rho F(\rho \, \omega),
    \qquad \rho>0,\quad \omega\in S^2.
\end{equation}
If \(0\leq\sigma\leq1\), then
\begin{equation}
    \norm{\mathcal RF}_{H^\sigma(\R_+;L^2(S^2))}
    \leq C_\sigma\norm{F}_{H^\sigma(\R^3)}.
    \label{eq:radial-lifting}
\end{equation}
If \(1/2<\sigma\leq1\), the boundary trace of \(\mathcal RF\) at
\(\rho=0\) vanishes.  Consequently, its extension by zero to \(\R\),
denoted by \(E_0\mathcal RF\), satisfies
\begin{equation}
    \norm{E_0\mathcal RF}_{H^\sigma(\R;L^2(S^2))}
    \leq C_\sigma\norm{F}_{H^\sigma(\R^3)}.
    \label{eq:zero-extension}
\end{equation}
\end{lemma}

\begin{proof}
At \(\sigma=0\), polar coordinates give
\begin{equation}
    \norm{\mathcal RF}_{L^2(\R_+;L^2(S^2))}^2
    =\int_0^\infty\int_{S^2}
      \rho^2|F(\rho \, \omega)|^2\,d\omega\,d\rho
    =\norm{F}_{L^2(\R^3)}^2.
    \label{eq:radial-isometry}
\end{equation}
For \(F\in C_c^\infty(\R^3)\), differentiation in \(\rho\) yields
\begin{equation}
    \partial_\rho\bigl(\rho F(\rho \, \omega)\bigr)
    =F(\rho \, \omega)+\rho\,\omega\cdot\nabla F(\rho \, \omega).
    \label{eq:radial-derivative}
\end{equation}
The second term on the right satisfies
\begin{equation}
    \int_0^\infty\int_{S^2}
    \rho^2|\omega\cdot\nabla F(\rho \, \omega)|^2\,d\omega\,d\rho
    \leq \norm{\nabla F}_{L^2(\R^3)}^2.
\end{equation}
For the first term, Hardy's inequality in dimension three gives
\begin{equation}
    \begin{aligned}
    \int_0^\infty\int_{S^2}|F(\rho \, \omega)|^2\,d\omega\,d\rho
    &=\int_{\R^3}\frac{|F(\xi)|^2}{|\xi|^2}\,d\xi\\
    &\leq4\norm{\nabla F}_{L^2(\R^3)}^2.
    \end{aligned}
\end{equation}
Together with \eqref{eq:radial-isometry}, this proves
\eqref{eq:radial-lifting} for \(\sigma=1\).  Complex interpolation between
the cases \(\sigma=0\) and \(\sigma=1\) proves
\eqref{eq:radial-lifting} for all \(0\leq\sigma\leq1\).

Suppose now that \(1/2<\sigma\leq1\).  The trace map
\begin{equation}
    H^\sigma(\R_+;L^2(S^2))\longrightarrow L^2(S^2),
    \qquad u\longmapsto u(0),
\end{equation}
is continuous.  For smooth \(F\), one has
\(\rho F(\rho\,\cdot)\to0\) in \(L^2(S^2)\) as \(\rho\downarrow0\).
The trace therefore vanishes for general \(F\in H^\sigma(\R^3)\) by
density and \eqref{eq:radial-lifting}.  Finally, when
\(1/2<\sigma<3/2\), extension by zero is bounded on the closed subspace of
\(H^\sigma(\R_+;L^2(S^2))\) consisting of functions with zero boundary
trace.  This proves \eqref{eq:zero-extension}.
\end{proof}

\begin{lemma}[Fourier decay from translation ]
\label{lem:translation-decay}
Let $\frac12<s<1$, and suppose that \(b\in L^1(\R)\) satisfies
\begin{equation}
    \norm{b(\,\cdot+h)-b}_{L^1(\R)}\leq M|h|^s,
    \qquad h\in\R.
    \label{eq:translation-modulus}
\end{equation}
Then
\begin{equation}
    \left|\int_\R e^{-it\rho}b(\rho)\,d\rho\right|
    \leq C_s\langle t\rangle^{-s}
    \bigl(\norm{b}_{L^1(\R)}+M\bigr).
    \label{eq:fourier-decay-modulus}
\end{equation}
\end{lemma}

\begin{proof}
Let $\widehat b_1(t)=\int_\R e^{-it\rho}b(\rho)\,d\rho$.  Translation
gives
\begin{equation}
    (e^{ith}-1)\widehat b_1(t)
    =\int_\R e^{-it\rho}\bigl(b(\rho+h)-b(\rho)\bigr)\,d\rho.
\end{equation}
For $|t|\geq1$, choose $h=\pi/|t|$.  Since $|e^{ith}-1|=2$, equation \eqref{eq:translation-modulus} gives
\begin{equation}
    |\widehat b_1(t)|\leq C_sM|t|^{-s}.
\end{equation}
For $|t|\leq1$, the assertion follows from
$$|\widehat b_1(t)|\leq\norm{b}_{L^1}$$.
\end{proof}

\begin{proof}[Proof of \thmref{thm:half-wave-decay}]
It suffices initially to take \(f,g\in\mathcal S(\R^3)\).  Define the
zero-extended radial Fourier profiles
\begin{equation}\label{eq:uf-gf-def-1-0}
    u_f(\rho,\omega)
    =\begin{cases}
      \rho\widehat f(\rho \, \omega),&\rho>0,\\
      0,&\rho<0,
    \end{cases}
    \qquad
    u_g(\rho,\omega)
    =\begin{cases}
      \rho\widehat g(\rho \, \omega),&\rho>0,\\
      0,&\rho<0.
    \end{cases}
\end{equation}
\lemref{lem:radial-lifting} and the identity
\begin{equation}
    \norm{\widehat f}_{H^s(\R^3)}=\norm{f}_{L^2_s(\R^3)}
    \label{eq:fourier-weight-identity}
\end{equation}
give
\begin{equation}
    \norm{u_f}_{H^s(\R;L^2(S^2))}
    \leq C_s\norm{f}_{L^2_s},
    \qquad
    \norm{u_g}_{H^s(\R;L^2(S^2))}
    \leq C_s\norm{g}_{L^2_s}.
    \label{eq:profile-Hs-bound}
\end{equation}
Introducing the scalar spectral density
\begin{equation}
    b_{f,g}(\rho)
    =\ip{u_f(\rho)}{u_g(\rho)}_{L^2(S^2)}.
    \label{eq:spectral-density}
\end{equation}
Cauchy--Schwarz in $\rho$ implies
\begin{equation}
    \norm{b_{f,g}}_{L^1(\R)}
    \leq\norm{u_f}_{L^2(\R;L^2(S^2))}
         \norm{u_g}_{L^2(\R;L^2(S^2))}.
    \label{eq:density-L1}
\end{equation}
Moreover,
\begin{equation}
    \begin{aligned}
    b_{f,g}(\rho+h)-b_{f,g}(\rho)
    ={}&\ip{u_f(\rho+h)-u_f(\rho)}{u_g(\rho+h)}_{L^2(S^2)}\\
       &+\ip{u_f(\rho)}{u_g(\rho+h)-u_g(\rho)}_{L^2(S^2)}.
    \end{aligned}
\end{equation}
The Hilbert-valued Sobolev translation estimate
\begin{equation}
    \norm{u(\,\cdot+h)-u}_{L^2(\R;\mathcal H)}
    \leq C_s|h|^s\norm{u}_{H^s(\R;\mathcal H)},
    \qquad 0<s<1,
    \label{eq:sobolev-translation}
\end{equation}
therefore yields
\begin{equation}
    \norm{b_{f,g}(\,\cdot+h)-b_{f,g}}_{L^1(\R)}
    \leq C_s|h|^s
    \norm{f}_{L^2_s}\norm{g}_{L^2_s}.
    \label{eq:density-translation}
\end{equation}
Applying \lemref{lem:translation-decay} to $b_{f,g}$, and using
\eqref{eq:density-L1}--\eqref{eq:density-translation}, gives
\begin{equation}
    \left|\int_\R e^{-it\rho}b_{f,g}(\rho)\,d\rho\right|
    \leq C_s\langle t\rangle^{-s}
    \norm{f}_{L^2_s}\norm{g}_{L^2_s}.
    \label{eq:density-fourier-decay}
\end{equation}
On the other hand, Plancherel and polar coordinates show that
\begin{equation}
    \begin{aligned}
    \ip{e^{-it|D|}f}{g}_{L^2(\R^3)}
    &=\int_{\R^3}e^{-it|\xi|}\widehat f(\xi)
      \overline{\widehat g(\xi)}\,d\xi\\
    &=\int_0^\infty e^{-it\rho}\rho^2
      \int_{S^2}\widehat f(\rho \, \omega)
      \overline{\widehat g(\rho \, \omega)}\,d\omega\,d\rho\\
    &=\int_\R e^{-it\rho}b_{f,g}(\rho)\,d\rho.
    \end{aligned}
    \label{eq:spectral-representation}
\end{equation}
Combining \eqref{eq:density-fourier-decay} and
\eqref{eq:spectral-representation}, we obtain
\begin{equation}
    \left|\ip{e^{-it|D|}f}{g}_{L^2}\right|
    \leq C_s\langle t\rangle^{-s}
    \norm{f}_{L^2_s}\norm{g}_{L^2_s}.
\end{equation}
Since $(L^2_{-s})^*=L^2_s$, duality proves
\eqref{eq:weighted-half-wave-decay}.  The general case follows by density.
\end{proof}

The following is a limiting absorption principle for the half-wave operator in weighted $L^2_s$ spaces. For more details see \cite{agmon1975spectral,goldberg2004limiting,ionescu2006agmon,burak2019limiting,rodnianski2015effective}.

\begin{lemma}[Uniform LAP for the half-wave operator]
\label{lem:uniform-lap-for-R-eps-pm}
For $\frac12<s<1$, 
one has
\begin{equation}
\label{eq:one-sided-main-in-uni-LAP}
    \sup_{\substack{\lambda\in\R}}
 	\|(|\nabla|-\lambda-i\epsilon )^{-1}\|_{L^2_s\rightarrow L^2_{-s}}
 	\le C_{s}  \,,  
\end{equation}
uniformly in $\epsilon>0$ and $\lambda\in\R$.

\end{lemma}

\begin{proof}
Let $f,g\in\mathcal S(\R^3)$, and we recall the
zero-extended radial Fourier profiles as in \eqref{eq:uf-gf-def-1-0}.
Then,
\begin{eqnarray}
    \Big\langle g, (|\nabla|-\lambda-i\epsilon)^{-1}f \Big\rangle
    &=&\int d\Omega(\omega) \int_0^\infty d\rho
    \frac1{\rho-\lambda-i\epsilon} 
    \overline{u_{ g}(\rho,\omega)} 
    u_{ f}(\rho,\omega) \,.
\end{eqnarray}
From \lemref{lm:unif-halfline-bd-1-0} follows that this is bounded by
\begin{eqnarray}
    &\leq&C_s \int d\Omega(\omega) 
    \|u_{ f}(\rho,\omega)\|_{H^s(\mathbb R_+)}
    \|u_{ g}(\rho,\omega)\|_{H^s(\mathbb R_+)}
    \nonumber\\
    &\leq&
    C_s  \|\widehat f\|_{H^s(\R^3)}
    \|\widehat g\|_{H^s(\R^3)}
    \nonumber\\
    &=&
    C_s  \|f\|_{L^2_s(\R^3)}
    \|g\|_{L^2_s(\R^3)}
\end{eqnarray}
using \lemref{lem:radial-lifting} for $\frac12<s<1$. This proves the claim.
\end{proof}

\begin{lemma}[Uniform half-line Cauchy-transform bound]
\label{lm:unif-halfline-bd-1-0}
Let $s>\frac12$ and let $f,g\in H^s(\mathbb R_+)$ satisfy
$f(0)=g(0)=0$. Then
\begin{equation}
 \sup_{\lambda\in\mathbb R,\,\epsilon>0}
 \left|\int_0^\infty \frac{f(r)g(r)}{r-\lambda-i\epsilon}\,dr\right|
 \leq C_s\|f\|_{H^s(\mathbb R_+)}\|g\|_{H^s(\mathbb R_+)}.
 \label{eq:uniform-cauchy-bound}
\end{equation}
In particular, the constant is uniform for $\lambda\in\R$ and $\epsilon>0$.
\end{lemma}

\begin{proof}
Fix
\begin{equation}
 0<\alpha<\min\left\{1,s-\frac12\right\},
\end{equation}
and set $h=fg$ on $\mathbb R_+$.  By Sobolev embedding and the product
estimate in $C^{0,\alpha}$,
\begin{equation}
 \|h\|_{L^1(\mathbb R_+)}+\|h\|_{L^\infty(\mathbb R_+)}
 +[h]_{C^{0,\alpha}(\mathbb R_+)}
 \leq C_{s,\alpha}
 \|f\|_{H^s(\mathbb R_+)}\|g\|_{H^s(\mathbb R_+)}.
 \label{eq:product-holder-bound}
\end{equation}
Here, the $L^1$ term follows from Cauchy--Schwarz. Since $h(0)=0$, its
extension by zero,
\begin{equation}
 H(r):=\begin{cases}h(r),&r\geq0,\\0,&r<0,\end{cases}
\end{equation}
belongs to $L^1(\mathbb R)\cap C^{0,\alpha}(\mathbb R)$ and satisfies the
same bound as in \eqref{eq:product-holder-bound}, up to a fixed constant.

For $\lambda\in\mathbb R$ and $\epsilon>0$, changing variables $x=r-\lambda$ and
splitting at $|x|=1$ gives
\begin{align}
 \int_{\mathbb R}\frac{H(r)}{r-\lambda-i\epsilon}\,dr
 &=\int_{|x|<1}\frac{H(\lambda+x)-H(\lambda)}{x-i\epsilon}\,dx
   +H(\lambda)\int_{-1}^{1}\frac{dx}{x-i\epsilon} \notag\\
 &\quad+\int_{|x|\geq1}\frac{H(\lambda+x)}{x-i\epsilon}\,dx.
 \label{eq:cauchy-decomposition}
\end{align}
The three terms on the right satisfy, uniformly in $\lambda$ and $\epsilon$,
\begin{align}
 \left|\int_{|x|<1}\frac{H(\lambda+x)-H(\lambda)}{x-i\epsilon}\,dx\right|
 &\leq [H]_{C^{0,\alpha}}
       \int_{-1}^{1}|x|^{\alpha-1}\,dx
 =\frac{2}{\alpha}[H]_{C^{0,\alpha}},\\
 \left|H(\lambda)\int_{-1}^{1}\frac{dx}{x-i\epsilon}\right|
 &=2|H(\lambda)|\arctan(\epsilon^{-1})
 \leq\pi\|H\|_{L^\infty},\\
 \left|\int_{|x|\geq1}\frac{H(\lambda+x)}{x-i\epsilon}\,dx\right|
 &\leq\|H\|_{L^1}.
\end{align}
Because $H$ vanishes on $(-\infty,0)$, the integral on the left-hand side
of \eqref{eq:cauchy-decomposition} is precisely the integral in
\eqref{eq:uniform-cauchy-bound}.  Combining the preceding estimates with
\eqref{eq:product-holder-bound} proves the claim.
\end{proof}

\begin{remark}
The proof only requires $f(0)g(0)=0$; hence, it is sufficient that one of the two factors has a zero trace at the origin.
\end{remark}


\section{Well-posedness and proof of \texorpdfstring{\propref{prop:lwp-mf-strong-final}}{Proposition~\ref*{prop:lwp-mf-strong-final}}}
\label{sec-appB-0-1}

\begin{lemma}
    \label{lm-Vgeneric-spec-1-0}
In the space 
\begin{equation}
    \mathcal V(\mathbb R^3)
    :=
    \left\{
        V\in L^\infty_{\mathrm{loc}}(\mathbb R^3;\mathbb R):
        \lim_{|x|\to\infty}V(x)=+\infty
    \right\}
\end{equation} of confining potentials endowed with the $L^\infty$ topology on bounded perturbations, the subset of potentials having simple, rationally independent eigenvalues is residual.
    
\end{lemma}

\begin{proof}
For \(K\in\mathbb N\) and
\(q=(q_0,\ldots,q_K)\in\mathbb Q^{K+1}\setminus\{0\}\), define
\begin{equation}
    \mathcal O_{K,q}
    :=
    \left\{
        V\in\mathcal V(\mathbb R^3):
        E_0(V),\ldots,E_K(V)\text{ are simple and }
        \sum_{j=0}^{K}q_jE_j(V)\neq0
    \right\}.
\end{equation} 
Proposition~3.2 of Mason and Sigalotti \cite{MasonSigalotti2010}
states that every $\mathcal O_{K,q}$ is open and dense in
$\mathcal V(\mathbb R^3)$.  Since the collection of pairs
$(K,q)$ is countable, the set
\begin{equation}
    \mathcal G_{\mathrm{nr}}
    :=
    \bigcap_{K\geq0}
    \bigcap_{q\in\mathbb Q^{K+1}\setminus\{0\}}
    \mathcal O_{K,q}
\end{equation}
is residual.  Thus, every $V\in\mathcal G_{\mathrm{nr}}$ has a simple,
rationally independent spectrum.
\end{proof}

\begin{lemma}[Product estimate in $Y^1$]
\label{lem:Y1-product-estimate}
For any $\vfi_t \in Y^1(\mathbb{R}^3)$ and $f \in W^{1,\infty}(\mathbb{R}^3)$, the product
$f\vfi_t$ belongs to $Y^1(\mathbb{R}^3)$, and there exists a universal constant $C>0$
such that
\begin{equation}
\|f\vfi_t\|_{Y^1} \le C \|f\|_{W^{1,\infty}} \|\vfi_t\|_{Y^1}.
\end{equation}
\end{lemma}
\begin{proof}
Directly expanding the $Y^1$ norm, we find
\begin{align*}
\|f\vfi_t\|_{Y^1}^2
=&\|\nabla(f\vfi_t)\|_{L^2}^2 + \int_{\mathbb{R}^3} V |f\vfi_t|^2 \, dx + (C_0+1)\|f\vfi_t\|_{L^2}^2 \\
&\le 2\|f\|_{L^\infty}^2 \|\nabla \vfi_t\|_{L^2}^2 + 2\|\nabla f\|_{L^\infty}^2 \|\vfi_t\|_{L^2}^2
+ \|f\|_{L^\infty}^2 \int_{\mathbb{R}^3} (V + C_0 + 1) |\vfi_t|^2 \, dx.
\end{align*}
Since $V + C_0 + 1 \ge 1$, we clearly have $\|\vfi_t\|_{L^2}^2 \le \|\vfi_t\|_{Y^1}^2$ and 
$\|\nabla\vfi_t\|_{L^2}^2 \le \|\vfi_t\|_{Y^1}^2$, as well as $\int (V+C_0+1)|\vfi_t|^2 \le \|\vfi_t\|_{Y^1}^2$. 
This immediately yields $\|f\vfi_t\|_{Y^1}^2 \le 5 \|f\|_{W^{1,\infty}}^2 \|\vfi_t\|_{Y^1}^2$.
\end{proof}

\bigskip

\begin{proof}[\textbf{Proof of \propref{prop:lwp-mf-strong-final}}]
\label{app:lwp-mf-strong-final}
Fix $t_0>0$. Define the Banach space
\begin{equation}
\cX_{t_0} := C([0,t_0];Y^1(\mathbb{R}^3))\times C([0,t_0];H^{1/2}(\mathbb{R}^3))
\end{equation}
equipped with the norm
\begin{equation}
\|(\vfi,u)\|_{\cX_{t_0}}:=\|\vfi\|_{L^\infty_t Y^1(\mathbb{R}^3)}+\|u\|_{L^\infty_t H^{1/2}(\mathbb{R}^3)}.
\end{equation}
Define the map $\Phi: \cX_{t_0} \to \cX_{t_0}$ by 

\begin{equation}\Phi(\vfi_t,u_t) := (\Phi_1(\vfi_t,u_t), \Phi_2(\vfi_t,u_t)),\end{equation} 
where 
\begin{align}
\Phi_1(\vfi_t,u_t)(t) &:= e^{-it\cH_0}\vfi_0 - i\int_0^t e^{-i(t-\tau)\cH_0} \Bigl( \lam (v*|\vfi_\tau|^2) + \eta \, w*(u_\tau+\overline{u_\tau}) \Bigr)\vfi_\tau \, d\tau, \\
\Phi_2(\vfi_t,u_t)(t) &:= e^{-it\omega(-i\nabla)}u_0 - i\eta\int_0^t e^{-i(t-\tau)\omega(-i\nabla)} (w*|\vfi_\tau|^2) \, d\tau.
\end{align}
Since $e^{-it\omega(-i\nabla)}$ is unitary on $H^{1/2}(\mathbb{R}^3)$ and $Y^1 \hookrightarrow H^1 \hookrightarrow L^2$, Young's inequality gives
\begin{equation}
\|w*|\vfi_\tau|^2\|_{H^{1/2}}
\le
\|w\|_{H^{1/2}}\||\vfi_\tau|^2\|_{L^1}
=
\|w\|_{H^{1/2}}\|\vfi_\tau\|_{L^2}^2
\le
\|w\|_{H^{1/2}}\|\vfi_\tau\|_{Y^1}^2.
\end{equation}
Integrating this bound over the time interval $[0,t_0]$ yields
\begin{equation}
\label{eq:local-u-bound}
\|\Phi_2(\vfi,u)\|_{L^\infty_t H^{1/2}} \le \|u_0\|_{H^{1/2}}+ \eta\,t_0\,\|w\|_{H^{1/2}}\,
\|\vfi\|_{L^{\infty}_t Y^1}^2.
\end{equation}
This estimate also demonstrates the local Lipschitz continuity of $\Phi_2$. By Young's inequality for convolutions, Cauchy-Schwarz, and the embedding $Y^1 \hookrightarrow L^2$, we have for any two pairs $(\vfi_t,u_t)$ and $(\tilde{\vfi}_t,\tilde{u}_t)$,
\begin{equation}
    \|\Phi_2(\vfi_t,u_t) - \Phi_2(\tilde{\vfi}_t,\tilde{u}_t)\|_{L^\infty_t H^{1/2}} \leq \eta t_0 \|w\|_{H^{1/2}}  \|\vfi_t + \tilde{\vfi}_t\|_{L^\infty_t Y^1} \|\vfi_t - \tilde{\vfi}_t\|_{L^\infty_t Y^1}.
\end{equation}
For the particle component $\Phi_1$, we estimate the nonlinearities using  \lemref{lem:Y1-product-estimate}. 
For the Hartree term, let $f = v*|\vfi_t|^2$. Since $v\in W^{1,\infty}(\mathbb{R}^3)$,
\begin{equation}
\|f\|_{W^{1,\infty}} \le \|v\|_{W^{1,\infty}} \|\vfi_t\|_{L^2}^2 \le C_v \|\vfi_t\|_{Y^1}^2.
\end{equation}
Applying \lemref{lem:Y1-product-estimate}, we bound the Hartree nonlinearity
\begin{equation}
\|(v*|\vfi_t|^2)\vfi_t\|_{Y^1}
\le
C_v \|\vfi_t\|_{Y^1}^3.
\label{eq:Hartree-Y1-bound}
\end{equation}
For the field coupling term, let $f = w*(u_t +\bar u_t)$. Given that $w \in H^{1/2}$, using the Cauchy-Schwarz inequality, one gets
\begin{align}
    \|f\|_{W^{1,\infty}} = & \norm{w*(u_t+\bar u_t)}_{W^{1,\infty}} 
    \nonumber\\
    \lesssim & \norm{w}_{H^{1/2}} \norm{u_t+\bar u_t}_{H^{1/2}} 
\end{align}
Applying \lemref{lem:Y1-product-estimate} again,
\begin{equation}
\|(w*(u_t+\bar u_t))\vfi_t\|_{Y^1}
\le
C_w\,\|u_t\|_{H^{1/2}}\,\|\vfi_t\|_{Y^1}.
\label{eq:coupling-Y1-bound}
\end{equation}
Since $e^{-it\cH_0}$ is an isometry on $Y^1(\mathbb{R}^3)$, integrating the sum of \eqref{eq:Hartree-Y1-bound} and \eqref{eq:coupling-Y1-bound} yields
\begin{equation}
\label{eq:local-phi-bound}
\|\Phi_1(\vfi_t,u_t)\|_{L^\infty_t Y^1} \le \|\vfi_0\|_{Y^1} + C t_0 \left(
|\lam| \|\vfi_t\|_{L^\infty_t Y^1}^3 + \eta \|u_t\|_{L^\infty_t H^{1/2}} \|\vfi_t\|_{L^\infty_t Y^1} \right).
\end{equation}
The local Lipschitz continuity of $\Phi_1$ follows from analogous estimates using \lemref{lem:Y1-product-estimate}. That is, for any two pairs $(\vfi_t,u_t)$ and $(\tilde{\vfi}_t,\tilde{u}_t)$, we have 
\begin{align}
    \|\Phi_1(\vfi_t,u_t) - \Phi_1(\tilde{\vfi}_t,\tilde{u}_t)\|_{L^\infty_t Y^1}  &\leq  t_0 C \Big(\|\vfi_t - \tilde{\vfi}_t\|_{L^\infty_t Y^1}
    +\|u_t-\widetilde{u}_t\|_{L^\infty_t H^{1/2}}\Big),
\end{align}
where $C$ depends only on the norms of the data in the local space.

Therefore, define the closed ball $B_R:=\{(\vfi,u)\in \cX_{t_0}:\|(\vfi,u)\|_{\cX_{t_0}}\le R\}$. By the estimates \eqref{eq:local-u-bound} and \eqref{eq:local-phi-bound}, we can choose $t_0$ sufficiently small such that $\Phi(B_R) \subseteq B_R$. Moreover, the local Lipschitz continuity of $\Phi$ implies that $\Phi$ is a strict contraction on $B_R$. Therefore, by the Banach fixed-point theorem, there exists a unique fixed point of $\Phi$ in $B_R$, corresponding to a unique mild solution on $[0,t_0]$.
\end{proof}

\begin{lemma}[Propagation of the $H^{\frac12}$-regularity of the field]
\label{lem:field-Hs-propagation}
Let $(\vfi_t,u_t)$ be a mild solution on $[0,T_{\max})$. Assume $w\in H^{\frac12}(\mathbb R^3)$ and that the boson mass $\|\vfi_t\|_{L^2}$ is conserved. Then, for every $t\in [0,T_{\max})$,
\begin{equation}
\label{eq:Hs-propagation-final}
\|u_t\|_{H^{1/2}(\mathbb R^3)}
\le
\|u_0\|_{H^{1/2}(\mathbb R^3)}
+
\eta\,t\,\|w\|_{H^{1/2}(\mathbb R^3)}\,
\|\vfi_0\|_{L^2(\mathbb R^3)}^2.
\end{equation}
\end{lemma}
\begin{proof}
Since $e^{-it\omega(-i\nabla)}$ is unitary on $H^{1/2}$, the Duhamel formula and Young's inequality yield
\begin{equation}
\|u_t\|_{H^{1/2}}\le\|u_0\|_{H^{1/2}}+\eta\int_0^t \|w\|_{H^{1/2}}\|\vfi_\tau\|_{L^2}^2\,d\tau.
\end{equation}
By $L^2$-mass conservation, $\|\vfi_\tau\|_{L^2}^2 = \|\vfi_0\|_{L^2}^2$, yielding the linear growth bound.
\end{proof}

The global theory relies on the coercive lower bound for the full conserved energy. 

\begin{lemma}[Coercive lower bound for the full energy]
\label{lem:energy-coercive-lower-bound}
Assume that the potentials satisfy
\begin{equation}
V(x)\ge -C_0 \quad \text{on } \mathbb{R}^3,
\qquad
v\in L^\infty(\mathbb{R}^3),
\qquad
w\in L^1(\mathbb{R}^3)\mcap L^\infty(\mathbb{R}^3).
\end{equation}
Assume the dispersion relation $\omega(\xi)=|\xi|$ is strictly positive almost everywhere and defines a field energy space $X_\omega := \bigl\{ u \in \mathcal{S}'(\mathbb{R}^3; \mathbb{C}) : \|\omega(-i\nabla)^{1/2}u\|_{L^2} < \infty \bigr\}$ that embeds continuously into $L^3(\mathbb{R}^3)$. Let $C_\omega>0$ denote the sharp embedding constant such that
\begin{equation}
\|u\|_{L^3(\mathbb{R}^3)} \le C_\omega \|\omega(-i\nabla)^{1/2}u\|_{L^2(\mathbb{R}^3)} \quad \text{for all } u \in X_\omega.
\end{equation}
For any state $(\vfi, u) \in Y^1(\mathbb{R}^3) \times X_\omega$, the full energy
\begin{align}
\label{eq:full-energy-of-mf}
\mathcal{E}[\vfi, u] =\;& \frac{1}{2} \int_{\mathbb{R}^3} |\nabla \vfi(x)|^2 \, dx + \frac{1}{2} \int_{\mathbb{R}^3} V(x) |\vfi(x)|^2 \, dx + \frac{\lambda}{4} \int_{\mathbb{R}^3} (v * |\vfi|^2)(x) |\vfi(x)|^2 \, dx \notag \\
& + \frac{\eta}2 \int_{\mathbb{R}^3} (w * (\bar{u} + u))(x) |\vfi(x)|^2 \, dx + \frac{1}{2} \int_{\mathbb{R}^3} \bar{u} \, \omega(-i \nabla) u \, dx
\end{align}
satisfies the lower bound
\begin{align}
\mathcal{E}[\vfi, u] + \frac{C_0+1}{2}\|\vfi\|_{L^2(\mathbb R^3)}^2
\ge\;&
\frac12\|\vfi\|_{Y^1(\mathbb{R}^3)}^2
+
\frac14\|\omega(-i\nabla)^{1/2}u\|_{L^2(\mathbb{R}^3)}^2 \notag\\
&\; -
\left(
\frac{|\lambda|}{4}\|v\|_{L^\infty}
+
4 C_\omega^2 \eta^2\|w\|_{L^{3/2}}^2
\right)
\|\vfi\|_{L^2(\mathbb{R}^3)}^4.
\label{eq:energy-lower-bound-final}
\end{align}
In particular, the energy is coercive in the sense that
\begin{equation}
\label{eq:energy-lower-bound-coercive-form}
\mathcal{E}[\vfi, u]
\ge
c_1\|\vfi\|_{Y^1(\mathbb{R}^3)}^2
+
c_2\|\omega(-i\nabla)^{1/2}u\|_{L^2(\mathbb{R}^3)}^2
-
C\bigl(1+\|\vfi\|_{L^2(\mathbb{R}^3)}^4\bigr)
\end{equation}
for suitable constants $c_1,c_2,C>0$ depending only on
$C_0,\lambda,\eta,\|v\|_{L^\infty}$, and $\|w\|_{L^{3/2}}$.
\end{lemma}

\begin{proof}
We observe that adding $\frac{C_0+1}{2}\|\vfi\|_{L^2}^2$ to the linear components of the energy exactly reconstructs the form domain norm defined in \eqref{eq:Y1-norm-def}
\begin{equation}
\frac12\|\nabla\vfi\|_{L^2}^2 + \frac12\int_{\mathbb{R}^3}V|\vfi|^2\,dx + \frac{C_0+1}{2}\|\vfi\|_{L^2}^2 = \frac12\|\vfi\|_{Y^1}^2.
\end{equation}
By Young's inequality, the Hartree term is bounded below by $-\frac{|\lambda|}{4}\|v\|_{L^\infty}\|\vfi\|_{L^2}^4$. 
For the coupling term, Hölder's inequality and the $X_\omega \hookrightarrow L^3$ embedding yield
\begin{equation}
\left| \eta\int_{\mathbb R^3}(w*(u+\bar u))|\vfi|^2 \right|
\le
2 \eta C_\omega\|w\|_{L^{3/2}}\|\vfi\|_{L^2}^2
\|\omega(-i\nabla)^{1/2}u\|_{L^2}.
\end{equation}
Applying Young's product inequality $ab\le \frac14 b^2+a^2$, we obtain
\begin{equation}
\eta\int_{\mathbb R^3}(w*(u+\bar u))|\vfi|^2
\ge
-
\frac14\|\omega(-i\nabla)^{1/2}u\|_{L^2}^2
-
4C_\omega^2\eta^2\|w\|_{L^{3/2}}^2\|\vfi\|_{L^2}^4.
\end{equation}
Summing these interaction bounds with the reconstructed linear $Y^1$ term yields the desired bound \eqref{eq:energy-lower-bound-final}.
\end{proof}

  \bigskip
\noindent{\bf Acknowledgments:} 
The authors thank Prof. Kenji Nakanishi for helpful comments.
T.C. gratefully acknowledges support by the NSF through the grant DMS-2009800, and the RTG Grant DMS-1840314 - {\em Analysis of PDE}. GPT 5.6 has assisted streamlining the presentation of Appendix A; the results were derived by the authors.

\bibliographystyle{abbrvnat}
\bibliography{references}
\end{document}